\documentclass[noinfoline,11pt]{imsart}

\usepackage[OT1]{fontenc}
\usepackage{amsthm,amsmath,graphicx,latexsym,amssymb,array,mathabx,mathrsfs,tcolorbox}
\usepackage{natbib}
\usepackage[colorlinks,citecolor=blue,urlcolor=blue]{hyperref}
\usepackage{adjustbox,blindtext}
\usepackage{mathdots}
\usepackage{verbatim}
\usepackage{array,makecell}  

\theoremstyle{plain}
\usepackage[normalem]{ulem}
\allowdisplaybreaks

\usepackage{fullpage}

\newcommand{\E}{\mathbb{E}}

\newcommand{\verti}[1]{{\left\vert\kern-0.25ex\left\vert\kern-0.25ex\left\vert #1 
    \right\vert\kern-0.25ex\right\vert\kern-0.25ex\right\vert}}

\DeclareMathOperator{\diag}{diag}

\usepackage{subfigure}
\usepackage{layout}

\usepackage{dsfont}
\usepackage[mathscr]{eucal}
\usepackage[toc,page]{appendix}
\usepackage{mathrsfs}
\usepackage{color}
\usepackage{pifont}
\usepackage{bm}
\usepackage{latexsym}
\usepackage{amsfonts}
\usepackage{amssymb}
\usepackage{epsfig}
\usepackage{graphicx}
\usepackage{multirow}

\newtheorem{theorem}{Theorem}
\newtheorem{proposition}{Proposition}
\newtheorem{corollary}{Corollary}
\newtheorem{lemma}{Lemma}

\newtheorem{remark}{Remark}

\newcommand{\transpose}{^{\top}}
\startlocaldefs
\numberwithin{equation}{section}
\theoremstyle{plain}
\endlocaldefs

\begin{document}

\begin{frontmatter}

\title{{\large Testing for the Rank of a Covariance Operator}}

\runtitle{Testing for the Rank of a Covariance Operator}

\begin{aug}
\author{\fnms{Anirvan} \snm{Chakraborty}\ead[label=e1]{anirvan.c@iiserkol.ac.in}} \and
\author{\fnms{Victor M.} \snm{Panaretos}\ead[label=e2]{victor.panaretos@epfl.ch}}

\runauthor{A. Chakraborty \& V.M. Panaretos}

\affiliation{Ecole Polytechnique F\'ed\'erale de Lausanne}

\address{Indian Institute of Science Education and Research Kolkata \&
Ecole Polytechnique F\'ed\'erale de Lausanne\\
\printead{e1}, \printead*{e2}}

\end{aug}

\begin{abstract}  How can we discern whether the covariance operator of a stochastic process is of reduced rank, and if so, what its precise rank is? And how can we do so at a given level of confidence? This question is central to a great deal of methods for functional data, which require low-dimensional representations whether by functional PCA or other methods. The difficulty is that the determination is to be made on the basis of i.i.d. replications of the process observed discretely and with measurement error contamination. This adds a ridge to the empirical covariance, obfuscating the underlying dimension. We build a matrix-completion inspired test statistic that circumvents this issue by measuring the best possible least square fit of the empirical covariance's off-diagonal elements, optimised over covariances of given finite rank. For a fixed grid of sufficiently large size, we determine the statistic's asymptotic null distribution as the number of replications grows. We then use it to construct a bootstrap implementation of a stepwise testing procedure controlling the family-wise error rate corresponding to the collection of hypotheses formalising the question at hand. Under minimal regularity assumptions we prove that the procedure is consistent and that its bootstrap implementation is valid. The procedure circumvents smoothing and associated smoothing parameters, is indifferent to measurement error heteroskedasticity, and does not assume a low-noise regime. An extensive simulation study reveals an excellent practical performance, stably across a wide range of settings, and the procedure is further illustrated by means of two data analyses.
\end{abstract}



\end{frontmatter}

{ \begin{footnotesize}
\tableofcontents
\end{footnotesize}
}

\section{Introduction} \label{sec1}

Principal component analysis (PCA) plays a fundamental role in statistics due to its ability to focus on a parsimonious data subspace that is most relevant for many practical purposes. In the case of functional data, it assumes an even more prominent role because a reduction in the data dimension maps the statistical problem back to a more familiar multivariate setting. Furthermore, regularization techniques that are necessary for functional regression, testing, prediction, and classification typically hinge on the identification of the most prominent sources of variation in the data.

\indent One of the main drawbacks of principal component analysis for any kind of data is that the  procedure of estimation/selection of the number of components to retain is often exploratory in nature. Indeed, one has to either inspect the scree plot or select the first few components that explain, say $85\%$ of the total variation (see, e.g., \cite{Joll02}). There are but a few confirmatory procedures to this end (see, e.g., \cite{Horn65}, \cite{Veli76} and \cite{PJS05}). However, each of these procedures rely on its own assessment of what is an appropriate definition of the dimension of the data corresponding to how many components to retain. In this paper, we view the problem from the perspective of hypothesis testing. Indeed, the high level or \textit{global} problem that is being considered is that of $\{H_0 : \mbox{rank}(\Sigma) < d\}$ versus $\{H_1 : \mbox{rank}(\Sigma) = d\}$, where $\Sigma$ is the covariance matrix of the $d$-dimensional distribution that generates the data. Thus, we want to test whether the data gives us enough evidence to conclude that its intrinsic variation is lower dimensional. If this null hypothesis is rejected based on the observed data, one can also consider a more detailed analysis and test the \textit{local} hypotheses $\{H_{0,q} : \mbox{rank}(\Sigma) = q\}$ versus $\{H_{1,q} : \mbox{rank}(\Sigma) > q\}$ for each $q = 1,2,\ldots, (d-1)$.

\indent When dealing with functional data, the object of interest is a covariance kernel $k_X$ (rather than a covariance matrix), which a priori is an infinite dimensional object. One would then wish to test $\{H_0 : \mbox{rank}(k_X) \leq d\}$ versus $\{H_1 : \mbox{rank}(k_X) > d\}$ at a global level, for some finite integer $d$ where the rank of $k_X$ is at most $d$ if and only if its Mercer expansion has no more than $d$ terms.  
 
 In practice, we can observe each of a sample of $n$ curves on a finite number, say $L$, of grid nodes, so $d$ will certainly have to be at most $(L\wedge n)-1$ for the testing problem to not be vacuous. The local hypotheses will then be $\{H_{0,q} : \mbox{rank}(k_X) = q\}$ versus $\{H_{1,q} : \mbox{rank}(k_X) > q\}$ for $1\leq q\leq d$. Based on the $n$ sample curves, and assuming that the are observed on the same grid, the simplest procedure would be to look at the rank of the $L\times L$ empirical covariance matrix. 
 
\indent At first sight, this problem seems simple: if the number of observations exceeds $d$, then perfect inference will be feasible. However, functional data are most often additively corrupted by unobservable measurement errors that are usually modelled as independent random variables indexed by the grid points for each sample function. This additional noise adds a ``ridge" to the true covariance. More specifically, the covariance matrix of the observed erroneous data is of \textit{full} rank. Clearly, this gives rise to a problem of the true rank being confounded by the additive noise. One way of removing the effect of the errors is to use some smoothing procedure on the data (see e.g., \cite{RS05}). But this smoothing step obfuscates the problem since the relationship between the rank of $k_X$ and the rank of the smoothed data is unknown, and further depends on the choice of tuning parameter(s) used for smoothing. At this stage, it would seem that the problem is ``\textit{almost insolubly difficult}" as pointed out by \cite{HV06}, who further concluded that ``\emph{conventional approaches based on formal hypothesis testing will not be effective}". As a workaround, \cite{HV06} considered a ``\textit{low noise}'' setting (assuming the noise variance vanishes as the number of observations increases) and used an unconventional rank selection procedure based on the amount of \textit{unconfounded noise variance}. A weakness of the procedure was that it required the analyst to provide \textit{acceptable values} of the noise variance for the procedure to be implemented in practice, and these bounds are to be selected in an ad hoc manner.

\indent An alternative approach altogether is to view the problem not as one of testing, but rather as one of model selection. For instance, as part of their PACE method, and assuming Gaussian data, \cite{yao2005} offer a solution based on a pseudo-AIC criterion applied to a smoothed covariance whose diagonal has been removed.  Later work by \cite{LWC13} provides estimates of the effective dimension based on a BIC criterion employing and the estimate of the error variance obtained using the PACE approach with the difference being that they used a adaptive penalty term in place of that used in the classical BIC technique. For densely observed functional data, \cite{LWC13} also studied a modification of the AIC technique in \cite{yao2005} by assuming a Gaussian likelihood for the data. \cite{LWC13} finally considered versions of information theoretic criteria studied earlier by \cite{BN02} in the context of factor models in econometrics, where the latter method is used to choose the number of factors. For all of the procedures studied by \cite{yao2005} and \cite{LWC13}, the main drawback is the involvement of smoothing parameters which enter due to use of smoothing prior to dimension estimation. The asymptotic consistency of these procedures also depends on specific decay rates for the smoothing parameters, as well as on assumptions on the regularity of the true mean and covariance functions. {In early work, not explicitly framed in the context of FDA, \cite{kneip1994} used several smooth versions of the data constructed using a progression of smoothing parameters, to select a dimension based on a sum of residual estimated eigenvalues. Here too, the method's performance and asymptotics depend on regularity assumptions and decay rates for smoothing parameters.}

\indent In this paper, we steer back to a formal hypothesis testing perspective for the dimension problem. We demonstrate that it is possible to construct a valid test that circumvents the smoothing step entirely, by means of matrix completion. The proposed test statistic measures the best possible least square fit of the empirical covariance's off-diagonal elements by nonnegative matrices of a given finite rank, exploiting the fact that the corruption affects only the diagonal. 

Compared to smoothing based alternatives, our approach presents the following advantages:
\begin{itemize}
\item It provides a genuine testing procedure, inferring the rank with confidence guarantees.  

\item It does not rely on pre-smoothing and consequently on the choice of smoothing parameters.  

\item  
It rests on minimal regularity, in particular continuity of the covariance and sample paths.  

\item It can handle heteroskedastic measurement errors, which are detrimental to smoothing.

\item It does not require a ``low noise" regime, indeed the noise variances can be aribtrary.

\item It exhibits excellent finite sample performance, stably across a wide range of scenarios. 
\end{itemize}

\indent The paper is organized as follows. In subsection \ref{subsec2-1}, we discuss the problem statement and setup in detail. We then develop a key identifiability result in subsection \ref{section-identifiability}, which elucidates how the rank can be identified. Exploiting this result, Section \ref{section-procedure} describes the testing procedure. The asymptotic distribution of the test statistic, and a valid bootstrap-based calibration approach are introduced in \ref{subsec2-3} and \ref{subsec2-4}. Practical and computational aspects of its implementation are discussed in Section \ref{implementation}. An extensive simulation study is presented in section \ref{simulations}, where we benchmark the performance of our procedure relative to those studied by \cite{yao2005} and \cite{LWC13}. Two illustrative data analyses are presented in \ref{real-data}. Proofs of formal statements are collected in Section \ref{proofs}, and further technical details are given in sections \ref{critical_value} and \ref{hessian_invertibility}.

\section{Methodology} \label{sec2}

\subsection{Problem Statement and Background} \label{subsec2-1}

Let $X = \{X(t) : t \in [0,1]\}$ be the stochastic process in question, assumed zero mean and with continuous covariance kernel on $[0,1]^2$,
$$k_X(s,t)=\mathbb{E}[X(s)X(t)],\qquad (s,t)\in[0,1]^2.$$
Continuity of $k_X$ implies that it admits the Mercer expansion,
\begin{equation}\label{mercer_expansion}
k_X(s,t)=\sum_{m\geq 1}\lambda_m \varphi_m(s)\varphi_m(t)
\end{equation}
with the series on the right hand side converging uniformly and absolutely. Consequently, $X$ is mean square continuous and admits a Karhunen-Lo\`eve expansion
\begin{equation}\label{KL_expansion}
X(t)=\sum_{m\geq 1}Y_m\varphi_m(t),
\end{equation}
where $\{Y_m\}$ is a sequence of uncorrelated zero-mean random variables with variances $\lambda_m$, respectively. Convergence of the series is in the mean square sense, uniformly in $t$.  Given $n$ i.i.d. replications $\{X_{1}, \ldots, X_{n}\}$ of $X$, we observe the noise-corrupted discrete measurements
\begin{equation}\label{measurements}
W_{ij} = X_{i}(t_{j}) + \varepsilon_{ij},\qquad i=1,\ldots, n,\,j=1,\ldots,L,
\end{equation}
for a grid of $L$ points
$$0 \leq t_{1} < t_{2} < \ldots < t_{L} \leq 1.$$ 
We will assume that the {grid nodes are regularly spaced, i.e. $(j-1)/L\leq t_{j}<  j/L$}, for the sake of simplifying our statements, but this can be considerably relaxed. We assume that the $n\times L$ random variables $\epsilon_{ij}$'s are continuous random variables, independent of the $X_{i}$'s and themselves independent across both indices, with moments up to second order given by
$$\mathbb{E}[\varepsilon_{ij}]=0\quad\&\quad\mathrm{var}[\varepsilon_{ij}]=\sigma^2_j<\infty,\qquad i=1,\ldots, n,\,j=1,\ldots,L.$$
Note, in particular that the $\varepsilon_{ij}$ are allowed to be \emph{heteroskedastic} in $j$, i.e. the measurement precision may vary over the grid points. The measured vectors $\{(W_{i1},\ldots,W_{iL})\transpose  \}_{i=1}^{n}$ are now i.i.d. random vectors in $\mathbb{R}^L$ with $L\times L$ covariance matrix  
$$K_{W,L}=\mathrm{cov}\{(X_{1}(t_{1}),X_{1}(t_{2}),\ldots,X_{1}(t_{L}))\transpose  \}=K_{X,L}+D,$$
where:
\begin{itemize}
\item[--] $K_{X,L} := \{k_{X}(t_{p},t_{q})\}_{p,q=1}^{L}$ is the $L\times L$ matrix obtained by pointwise evaluation of $k_X(\cdot,\cdot)$ on the pairs $(t_i,t_j)$, and
\item[--] $D=\mathrm{diag}\{\sigma_{1}^{2},\sigma_{2}^{2},\ldots,\sigma_{L}^{2}\}$ is the $L\times L$ covariance matrix of the $L$-vector $(\varepsilon_{i1},\ldots,\varepsilon_{iL})\transpose  $.
\end{itemize}
In this setup, we wish to use the observations $\{W_{ij}:i\leq n,j\leq L\}$ in order to \emph{infer} whether the stochastic process $X$ is, in fact, \emph{low dimensional}, and if so what its dimension might be. We use the term \emph{infer} in its formal sense, i.e. we wish to be able to make statements in the form of hypothesis tests with a given level of significance. Concretely, the question posed pertains to whether the covariance $k_X$ is of reduced rank, in the sense of a finite Mercer expansion \eqref{mercer_expansion}, and if so of what rank. 

Formally, for some dimension $d<\infty$, we wish to test the hypothesis pair
\begin{equation}\label{global_hypotheses}
\left\{\begin{array}{r@{}l}
    H_{0}&{}:\mathrm{rank}(k_X) \leq d \\
    H_{1}&{}: \mathrm{rank}(k_X)> d
\end{array}\right\}
\end{equation}

Notice that we can never actually choose $d=\infty$, since we have finite data, which is why we have to settle with a $d< L\wedge n$. Typically $n\gg L$ so that $L\wedge n=L$. This \emph{global} hypothesis pair is related to the sequence of \emph{local} hypotheses
\begin{equation}\label{local_hypotheses}
\left\{\begin{array}{r@{}l}
    H_{0,q}&{}:\mathrm{rank}(k_X) = q \\
    H_{1,q}&{}: \mathrm{rank}(k_X)> q
\end{array}\right\},\qquad q=1,\ldots,d.
\end{equation}
In particular, if we can sequentially test all $d$ local hypotheses with a controlled family-wise error rate, then we will have a test for the global hypothesis, and a means to infer what the rank is, when $H_0$ is valid (more details in the next section). In any case, $k_X$ can be replaced by $K_{X,L}$ in the null hypotheses $\{H_{0,q}\}_{q=1}^{d}$, provided $L$ is sufficiently large relative to $d$:

\begin{proposition} \label{lemma-discrete-rank}
{Let $k_X:[0,1]^2\rightarrow\mathbb{R}$ be a continuous covariance kernel and $K_{X,L}=\{k_X(t_i,t_j)\}_{i,j=1}^{L}$. 
If $\mathrm{rank}(k_X)\geq d$ there exists $L_* <\infty$ such that $\mathrm{rank}(K_{X,L}) \geq d$ whenever $L \geq L_{*}$.}

\end{proposition}

As noted in the introduction, while this question is of clear intrinsic theoretical interest, it also arises very prominently when carrying out a functional PCA as a first step for further analysis, in particular when evaluating a scree plot to choose a truncation level: the choice of a truncation dimension $q$ can be translated into testing whether the rank of $k_X$ is equal to $q$.\\

\noindent The frustrating tradeoff faced by the statistician in the context of this problem is that:
\begin{enumerate}
\item Without any smoothing, the noise covariance $D$ confounds the the problem by the addition of a ridge to the empirical covariance, leading to an inflation of the underlying dimensionality. Specifically, the rank of $K_{W,L}=K_{X,L}+D$ is at most $n \wedge L$, with probability 1.

\item Attempts to denoise $K_{W,L}$ and approximate $K_{X,L}$ by means of smoothing will obfuscate the the problem, since the choice of smoothing/tuning parameters will interfere with the problem of rank selection.

\end{enumerate}

It is this tradeoff that \cite{HV06} presumably had in mind when referring to this problem of rank inference as ``\emph{almost insolubly difficult}". Despite the apparent difficulty, we wish to challenge their statement that ``\emph{conventional approaches based on formal hypothesis testing will not be effective}", demonstrating that this can be achieved via \emph{matrix completion}. The crucial obervation is that the corrupted diagonal can be \emph{entirely disregarded}, while still being able to identify the rank, owing to the continuity of the problem. How precisely is described in the next section.

\subsection{Identifiability} \label{section-identifiability}

The main idea we wish to put forward here is that it is feasible make inferences about the rank of $K_{X,L}$ without resorting to smoothing or low noise assumptions, simply by focussing on the off-diagonal elements of the matrix $K_{W,L}$ for any sufficiently large but finite grid size $L$. The point is that we have no information whatsoever on the diagonal matrix $D$, and cannot attempt to annhilate it by means of smoothing without biasing inference on the rank. Still, we  have
$$K_{X,L}(i,j)=K_{W,L}(i,j),\qquad\forall\,i\neq j$$
i.e. the matrices are equal off the diagonal, even if their relationship on the diagonal is completely unknown. So the rank of $K_{X,L}$ may still be identifiable from its off-diagonal entries. The first of our main results shows that this is indeed the case, owing to the continuity of $k_X$.

\begin{theorem}[Identifiability] \label{theorem-identifiability}
{Assume that the kernel $k_{X}$ is continuous on $[0,1]^{2}$, let $d\geq 1$, and let $q\in\{1,\ldots,d\}$. Then, there exists a critical $L_{\dagger}=L_\dagger(d) <\infty$   such that, for all $L>L_{\dagger}$, the functional}
\begin{center}$\Theta \mapsto \sum_{i\neq j} \Big(K_{W,L}(i,j)-\Theta(i,j)\Big)^2$\end{center}
restricted on the set $\mathcal{M}_q=\{\Theta\in\mathbb{R}^{L\times L}:\mathrm{rank}(\Theta)\leq q\}$ of matrices of rank at most $q$,
\begin{enumerate}
\item Vanishes uniquely at $K_{X,L}$ when $\mathrm{rank}(k_X)=q$.

\item Is bounded below by a positive constant when $\mathrm{rank}(k_X)>q$.

\end{enumerate}
\end{theorem}

\begin{remark}[Notation]
{The sum-of-squares term $\sum_{i\neq j} \left(K_{W,L}(i,j)-\Theta(i,j)\right)^2$  is simply the squared Frobenius distance between $\Theta$ and $K_{W,L}$ when disregarding their diagonal entries. We can re-write it more compactly as $\|P_{L} \circ (K_{W,L} - \Theta)\|_{F}^2$, where $P_L=\{ \mathbf{1}\{i\neq j\} \}_{i,j=1}^{L}$, $\|A\|_F=\sqrt{\mathrm{trace}(A\transpose  A)}$ is the Frobenius matrix norm, and `$\circ$' denotes the Hadamard (element-wise) product.}
\end{remark}

\begin{remark}[Critical Grid Size]\label{grid_remark}

The precise critical value $L_{\dagger}<\infty$ in Theorem \ref{theorem-identifiability} will generally depend on the on the boundary value $d$ in the global hypothesis pair \eqref{global_hypotheses}, and the spectrum of $k_X$. For most scenarios encountered in functional data analysis, the value
$$L_{\dagger}= 2d+1$$
suffices. This includes polynomial or trigonometric eigensystems and warped versions thereof, systems comprised of splines or other piecewise (non-vanishing) analytic basis elements,  and more generally systems with eigenfunctions that are linearly independent over sets of positive Lebesgue measure. Note that it is not the regularity of the eigenfunctions that is elemental here -- for instance, the last class described can include eigenfunctions that are nowhere differentiable. See Section \ref{critical_value} for a detailed discussion. \end{remark}

The theorem affirms that sequentially checking whether the rank of $k_X$ is equal to $q$ or exceeds $q$, for $q\in\{1,...,d\}$, is feasible by means of the off-diagonal entries of $K_{X,L}$ alone, and indeed for any finite grid $L> L_{\dagger}$. That is, the collection of local hypothesis pairs $\{H_{0,q},H_{1,q}\}_{q=1}^{d}$ is identifiable non-asymptotically in the grid size, even when observation is discrete and noisy. Consequently, we will henceforth be working in a framework where $L$ is assumed fixed but sufficiently large relative to $d$ (i.e. $L>L_\dagger$, where $L_\dagger$ is as in Theorem~\ref{theorem-identifiability}). 

Indeed, the identifiability is \emph{constructive}, in that if we had access to the true matrix $K_{W,L}$, starting with $q=1$ and proceeding sequentially, we could discern all $d$ hypothesis pairs as follows:
\begin{enumerate}

\item For any candidate rank $q\leq d$, we check whether
 $$ \min_{\Theta: \mathrm{rank}(\Theta)\leq q}\|P_{L} \circ (K_{W,L} - \Theta)\|_{F}^2=0.$$
 
\item If the minimum is positive we are certain that rank$(K_{X,L})> q$.

\end{enumerate}

\begin{figure}[hhhhh]\label{functional_schematic}
\begin{center}
\includegraphics[scale=0.35]{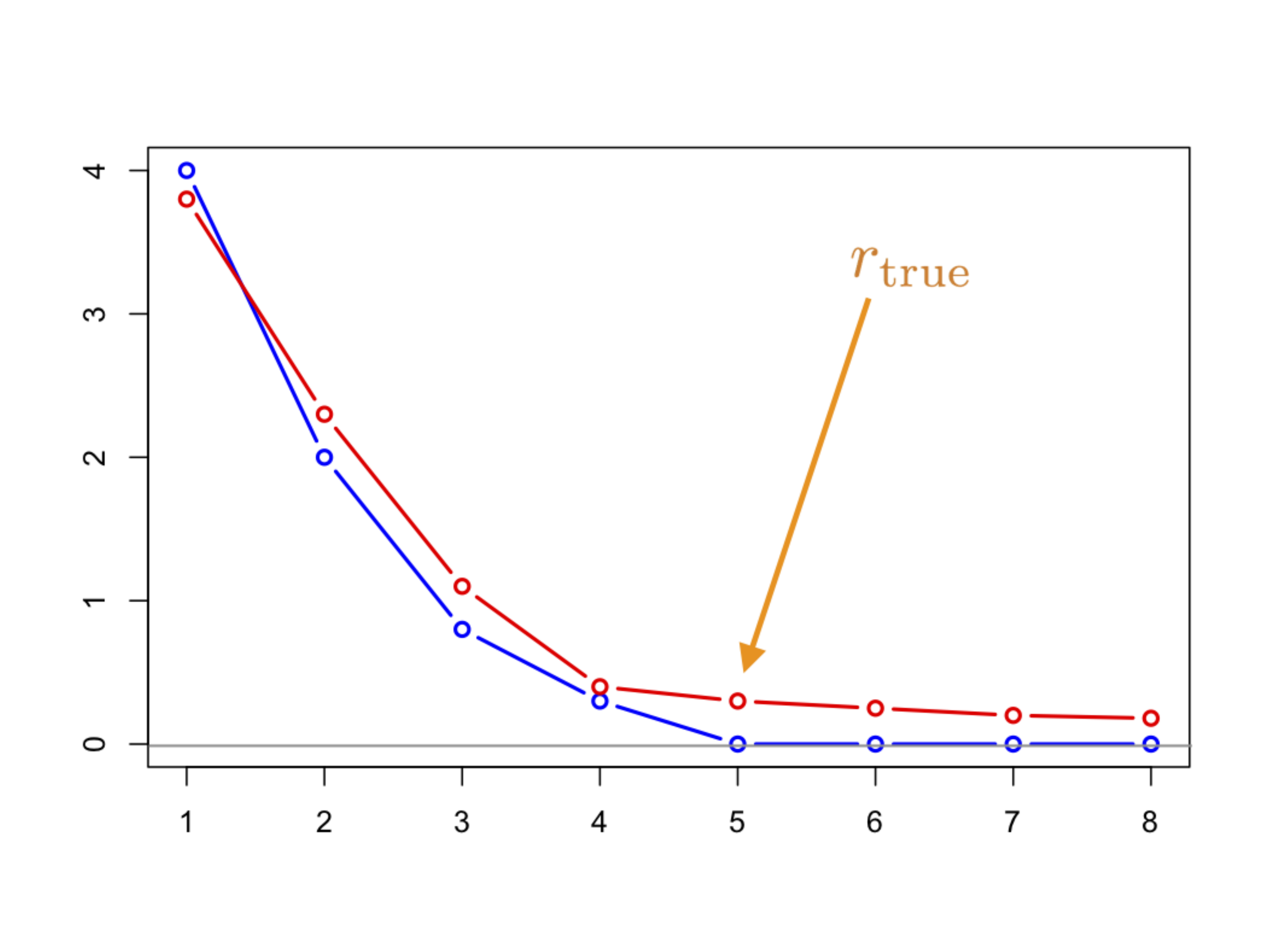}
\caption{Schematic illustration of the use of the rank-constrained off-diagonal ``residual sum of squares" as a means to discerning the true rank. The curve in blue represents the population functional $q\mapsto \min_{\Theta^{L \times L} : \mathrm{rank}(\Theta) \leq q} \|P_{L} \circ ({K}_{W,L} - \Theta)\|^{2}_{F}$, which is strictly positive for $q<r_{\mathrm{true}}$ and zero for $q\geq r_{\mathrm{true}}$. The curve in red represents the empirical functional $q\mapsto \min_{\Theta^{L \times L} : \mathrm{rank}(\Theta) \leq q} \|P_{L} \circ (\widehat{K}_{W,L} - \Theta)\|^{2}_{F}$.}
\end{center}
\end{figure}

\subsection{The Testing Procedure}\label{section-procedure}
This constructive identifiability  can be leveraged to construct a testing procedure. Of course, in practice the  matrix $K_{W,L}$ is \emph{unobservable} and we must rely on  $\widehat{K}_{W,L}$, the empirical covariance of the observed vector $(W_{i1},\dots,W_{iL})\transpose  $,
$$\widehat K_{W,L}:=\frac{1}{n}\sum_{i=1}^{n} ~(W_{i1},\dots,W_{iL})(W_{i1},\dots,W_{iL})\transpose  .$$ 
This motivates testing the local hypothesis pair $\{H_{0,q},H_{1,q}\}$ by means of the test statistic
 \begin{equation}\label{test_statistic}
 T_q = \min_{\Theta^{L \times L} : \mathrm{rank}(\Theta) \leq q} \|P_{L} \circ (\widehat{K}_{W,L} - \Theta)\|^{2}_{F},
 \end{equation}
rejecting $H_{0,q}$ in favour of $H_{1,q}$ for large values of $T_q$. Note the interpretation of the test statistic: to test whether the rank is $q$, we measure the best possible fit of the off-diagonal elements of the empirical covariance $\widehat{K}_{W,L}$ by a matrix of rank $q$. We reject when this fit is poor, and the calibration of $T_q$ is considered in the next two sections, via an asymptotic analysis based on $M$-estimation, and hinging on Theorem \ref{theorem-identifiability}.\\

For the moment, though, assume that we can obtain a $p$-value $p_q$ for $T_q$ (or some appropriately re-scaled version, e.g. $n\times T_q$) under the hypothesis $H_{0,q}$. In order to be able to test the global pair $\{H_0,H_1\}$ \eqref{global_hypotheses}, and infer the rank when the global null $\{H_0:\mathrm{rank}(k_X)\leq d\}$ is valid, we consider a stepwise procedure, for a given significance level $\alpha$:
\begin{itemize}
\item[\bf Step 1:] Test $H_{0,1} : \mathrm{rank}(K_{X,L}) = 1$ vs $H_{1} :  \mathrm{rank}(K_{X,L}) > 1$ by means of $T^{(1)}$. \

Stop if the corresponding $p$-value, $p_1$ exceeds $\alpha$; otherwise continue to Step 2. \\

\item[\bf Step 2:] Test $H_{0,2} : \mathrm{rank}(K_{X,L}) = 2$ vs $H_{1} : \mathrm{rank}(K_{X,L}) > 2$ by means of $T^{(2)}$. 

Stop if the corresponding $p$-value, $p_2$,  exceeds $\alpha$; otherwise continue similarly. \\

\item[$\qquad\vdots$]
\end{itemize}

\noindent 

We reject the global null $\{H_0:\mathrm{rank}(k_X)\leq d\}$ in \eqref{global_hypotheses} if and only if the sequential procedure terminates with the rejection of all local hypotheses up to and including the $d$-th one. If the procedure terminates earlier, the global null is not rejected, and we subsequently declare the rank of the functional data to be the value 
$$\widehat{r} := \min\{q \geq 1 : {p}_{q} > \alpha\},$$ i.e. the smallest $q$ for which we fail to reject $H_{0,q}$. This stepwise procedure strongly controls the Family Wise Error Rate (FWER) at level $\alpha$ (see \cite{MHL95} and \cite{LGSF17}). Indeed, observe that at most one of the hypotheses $\{H_{0,q}\}_{q=1}^{d}$ can be true, and suppose it corresponds to $q = q_{0}$. Then, if $V$ denotes the number of false discoveries among the number of rejections, one has
$$\{V > 0\} \Leftrightarrow \{V = 1\} \Leftrightarrow \{H_{0,q_0}  \, \mbox{has been rejected}\}  \Leftrightarrow \{{p}_{r_{0}} \leq \alpha\}.$$
So, FWER = $P(V > 0) = P({p}_{r_{0}} > \alpha) \leq \alpha$, where the probabilities are calculated under the given configuration of true and false null hypotheses, equivalently, under the assumption that $\{H_{0,q_0} : \mathrm{rank}(K_{X,L}) = q_{0}\}$ is true (which automatically ensures that the other hypotheses are false). Since $q_{0}$ is arbitrary, the FWER is controlled at level $\alpha$.  

Finally, if $\mathrm{rank}(k_X)< d$, we have 
$$P(\widehat{r} > r_{\mathrm{true}}) \leq P({p}_{r_{\mathrm{true}}} \leq \alpha) = P(V > 0) \leq \alpha.$$ 
Thus, the control over the FWER translates into \emph{a control over the probability of over-estimating the true rank}.

To implement the procedure, we will require the $p$-values $\{p_q\}$ corresponding to (an appropriately re-scaled version of) the test statistic $T_q$ under $H_{0,q}$. To this aim, the next two sections determine the large-$n$ sampling distribution of $n\times T_q$ under $H_{0,q}$ and describe a valid bootstrap procedure  for approximating $p$-values $\{p^*_q\}$ under $H_{0,q}$ in practice. En route, they also establish the consistency of the resulting test (and bootstrap procedure) as $n\rightarrow\infty$ under $H_{1,q}$, for all $L$ sufficiently large.

\subsection{Asymptotic Theory} \label{subsec2-3}

To justify the use of the test statistic $T_q$ for testing $\{H_{0,q}\,\mathrm{vs}\,H_{1,q}\}$ (for some given $q\leq d$), we will derive its asymptotic distribution under the null $H_{0,q}$ and the alternative $H_{1,q}$ as $n\rightarrow\infty$ for any $q\leq d$ and $L>L_\dagger$, after appropriate re-scaling (by $n$, in particular). To this aim, we introduce the functional,
$$ \Psi: \mathbb{R}^{L\times q} \rightarrow [0,\infty),\quad \Psi(C)= \|P_{L} \circ (K_{W,L} - CC\transpose )\|_{F}^2.$$

\noindent Furthermore, we collect the following assumptions:

\begin{description}
\item[\textbf{Assumption (C)}:] The covariance kernel $k_{X}(\cdot,\cdot)$ is continuous on $[0,1]^2$, the grid nodes $\{t_1,...,t_L\}$ are regularly spaced, and $\mathrm{var}[\varepsilon_{ij}]=\sigma^2_j\in [0,\infty)$.
\item[\textbf{Assumption (H)}:] Under $H_{0,q}$, there exists a factor $C_{0}\in\mathbb{R}^{L\times q}$ of $K_{X,L}$, i.e. $K_{X,L}=C_{0}C_{0}\transpose $ , such that the Hessian $\nabla^{2}\Psi(C_{0})$ is non-singular.

\end{description}

\begin{remark}[On The Hessian Condition]\label{hessian-remark}
{A sufficient condition for (H) to hold true is}
\begin{quote}
\begin{description}
\item[\textbf{Assumption (E)}:] The $q$ leading eigenvectors of $K_{X,L}$ have non-zero entries. 
\end{description}
\end{quote}
{In particular, if (E) is valid, then $C_0$ can be taken to be  equal to $V\Lambda^{1/2}$ where $K_{X,L}=V\Lambda V\transpose$ is the eigendecomposition of $K_{X,L}$, and the Hessian $\nabla^2\Psi(V\Lambda^{1/2})$ is provably non-singular. } {Condition (E), and hence Assumption (H), is automatically satisfied in all the settings listed in Remark \ref{grid_remark}. 
See Section \ref{hessian_invertibility} for more details. }
\end{remark}

\noindent We can now state our second main result:

\begin{theorem}[Asymptotic Distribution of the Test Statistic] \label{thm1}
Suppose that Assumptions (C) and (H) hold and let $q\leq d\leq L_\dagger<\infty$ be as in Theorem \ref{theorem-identifiability}. Denote the weak (centered Gaussian) limit of $\sqrt{n}(\widehat{K}_{W,L} - K_{W,L})$ by the random matrix $Z$. Then, for any $L>L_\dagger$,
\begin{itemize}
\item When $H_{0,q}$ is valid, we have as $n \rightarrow \infty$
$$ nT_q \stackrel{d}{\rightarrow} \|P_{L} \circ Z\|_{F}^{2} - 8~(\mathrm{vec}(P_{L} \circ Z))\transpose \{(C_{0} \otimes I_{L})(\nabla^{2}\Psi(C_{0}))^{-1}(C_{0}\transpose  \otimes I_{L})\}\mathrm{vec}(P_{L} \circ Z)$$

\item When $H_{1,q}$ is valid, $nT_q$ diverges to infinity as $n\rightarrow \infty$. 
\end{itemize}
\end{theorem}

The theorem justifies the use of $nT_q$ as a test statistic: though $T_q$ will not be precisely zero even when the true rank is $q$, the test statistic will converge to zero under $H_{0,q}$, with an asymptotic variance of the order of $n^{-2}$. The diffuse limiting law of $nT_q$ under $H_{0,q}$ in principle allows for calibration (though it does depend on unknown quantities, see the next Section). That $nT_q$ diverges under $H_{1,q}$ establishes the consistency of a test based on $nT_q$.

\subsection{Bootstrap Calibration} \label{subsec2-4}

Since the limiting null distribution of $nT_q$  established in Theorem \ref{thm1} depends on unknown quantities, we consider a bootstrap strategy in order to generate approximate $p$-values of $nT_q$ for testing the pair $\{H_{0,q},H_{1,q}\}$. If $H_{0,q}$ is truly valid, then a na\"ive bootstrap would suffice. But if $H_{1,q}$ is actually valid instead, a na\"ive bootstrap will fail to correctly approximate the sought $p$-values under $H_{0,q}$. In effect, we need a \emph{re-centering} (or rather, \emph{re-ranking}) scheme in order to generate bootstrap replications ``conforming" to $H_{0,q}$, even when $H_{1,q}$ holds true in reality. The purpose of this section is to present such a scheme and establish its validity.

\medskip
\noindent The proposed bootstrap scheme is:

\begin{tcolorbox}

\begin{description}
\begin{footnotesize}

\item[(1)] Find a minimizer $\widehat\Theta_q$ of $\|P_{L} \circ (\widehat{K}_{W,L} - \Theta)\|_{F}$ over nonnegative definite matrices $\Theta$ satisfying $\mathrm{rank}(\Theta) \leq q$.

\medskip
\item[(2)] For each $1 \leq i \leq n$, define
$$\widehat{m}({\bf W}_{i}) = \overline{{\bf W}} + \widehat\Theta_q\widehat{K}_{W,L}^{-1}({\bf W}_{i} - \overline{{\bf W}}).$$ 
where $\overline{{\bf W}} = n^{-1}\sum_{i=1}^{n} {\bf W}_{i}$. Under the null hypothesis $\{H_{0,q}:\mathrm{rank}(K_{X,L})=q\}$, this is an estimator of the best linear predictor of the discretely sampled curve ${\bf X}_{i}$ given the noise-corrupted version $\bf{W}_{i}$, i.e. $m({\bf W}_{i}) = \overline{{\bf W}} +  K^{(q)}_{X,L}{K}_{W,L}^{-1}({\bf W}_{i} - \overline{{\bf W}}).$

\medskip
\item[(3)] Estimate $D$ by $\widehat{D}$, defined as the diagonal matrix with $j$th diagonal element defined as $$\widehat{D}(j,j)=\max\{\widehat{K}_{W,L}(j,j)-\widehat\Theta_M(j,j),0\},$$ where $\widehat\Theta_M$ is a minimiser of $\|P_{L} \circ (\widehat{K}_{W,L} - \Theta)\|_{F}$ over nonnegative definite matrices $\Theta$ satisfying $\mathrm{rank}(\Theta) \leq M$, and 
$$M=m_n \mathbf{1}\{m_n < d\}+d\mathbf{1}\{m_n\geq d\}$$
with 
$${m_n=\min\left\{m\geq q: T_m\leq \epsilon \frac{{\log n}}{n} \right\}},$$
{and $0<\epsilon\leq 1$ an arbitrary constant.}

\medskip
\item[(4a)] Draw $n$ bootstrap observations $U^{*}_1, U^{*}_2, \ldots, U^{*}_n$ from $\{\widehat{m}({\bf W}_i) : 1 \leq i \leq n\}$.

\medskip
\item[(4b)] Draw $n$ i.i.d. observations $V^{*}_{1}, V^{*}_{2}, \ldots, V^{*}_{n}$ from an $L$-dimensional centered Gaussian distribution with covariance matrix $\widehat D + \widehat{A}$, where $\widehat{A} := \widehat\Theta_q - \widehat\Theta_q\widehat{K}_{W,L}^{-1}\widehat\Theta_q$. 

\medskip
\item[(5)] Define the $L$-vectors $\bm{\zeta}_{j} = U^{*}_j + V^{*}_{j}$ for $j = 1,2,\ldots,n$. 

\medskip
\item[(6)] Let $F^*_q$ be the law of
\begin{center}$T_{q}^* = \min_{\Theta^{L \times L} : \mathrm{rank}(\Theta) \leq q} \left\|P_{L} \circ \left(\frac{1}{n}\sum_{j=1}^{n}\bm\zeta_j \bm\zeta_j\transpose - \Theta\right)\right\|^{2}_{F}.$\end{center}

\medskip
\item[(7)] To test the pair $\{H_{0,q},H_{1,q}\}$  
use the bootstrap $p$-value $p^*_q=F^*_q(  T_q).$

\end{footnotesize}
\end{description}

\end{tcolorbox}
Of course, in practice we use $B<\infty$ random samples $\{\bm\zeta_{1,b},...,\bm\zeta_{n,b}\}_{b=1}^{B}$ 
to approximate the $p$-value $p^*$ in Step (7) by
$$p^*_{q,B}=\frac{1}{B}\sum_{b=1}^{B}\mathbf{1}\{T^*_{q,b} \leq  T_q\}=F^*_{q,B}(T_q)$$
where
$$T_{q,b}^* = \min_{\Theta^{L \times L} : \mathrm{rank}(\Theta) \leq q} \left\|P_{L} \circ \left(\frac{1}{n}\sum_{j=1}^{n}\bm\zeta_{j,b} \bm\zeta_{j,b}\transpose - \Theta\right)\right\|^{2}_{F}.$$
If one is willing to assume that the measurement errors are heteroskedastic, one can replace $\widehat D$ in Step 3 by its diagonally averaged version, 
$$\check D=\diag\left\{L^{-1}\sum_{j=1}^L\widehat{D}(j,j),\hdots,L^{-1}\sum_{j=1}^L\widehat{D}(j,j)\right\}.\label{hom}$$

The next remark explains the heuristic behind the bootstrap procedure, and the theorem succeeding it establishes the bootstrap procedure's validity. The procedure's finite sample performance is investigated thoroughly in the next Section.

\begin{remark}[Bootstrap Heuristic]
 Assume that the errors $\varepsilon_{ij}$ in \eqref{measurements} are Gaussian. Let $X^{(q)}_i(u)=\sum_{m=1}^{q}\langle X_i,\varphi_m\rangle_{L^2}\varphi_m(u)$ be the $q$-truncated Karhunen-Lo\`eve expansion of the curve $X_i$ and $\mathbf{X}_i^{(q)}=\{X_i^{(q)}(t_j)\}_{j=1}^{L}$ its discrete version when evaluated at the $\{t_j\}_{j=1}^{L}$. If we had access to the $L$-vectors $\{\mathbf{X}_i^{(q)}\}_{i=1}^{n}$ and $\{\bm\varepsilon_i\}_{i=1}^{n}$, then we would generate a bootstrap sample conforming to $H_{0,q}$ by means of constructing $n$ random $L$-vectors 
 $$\mathbf{W}^{(q)}_{i}=\mathbf{X}^{(q)}_i+\bm\delta_i,\qquad \bm{\delta}_i \mbox{ sampled randomly with replacement from }\{\bm\varepsilon_1,...,\bm\varepsilon_n\}.$$
These bootstrapped vectors would have covariance matrix $K^{(q)}_{X,L}+D$, where 
$$K^{(q)}_{X,L}(i,j)=\sum_{m=1}^{q}\lambda_m\varphi_m(t_i)\varphi_m(t_j).$$
If instead of observing  $\{\mathbf{X}_i^{(q)}\}_{i=1}^{n}$ and $\{\bm\varepsilon_i\}_{i=1}^{n}$, we only had access to their covariance $K^{(q)}_{X,L}$ and $D$, then we would do the ``next best thing", i.e. replace $\mathbf{X}^{(q)}_i$ by its best linear predictor given the actual observations,
$$m(\mathbf{W}_i)=\overline{{\bf W}} + K^{(q)}_{X,L}{K}_{W,L}^{-1}({\bf W}_{i} - \overline{{\bf W}})$$ 
and replace $\bm\delta_i$ by
$$V_i\sim \mathrm{N}_L(0,D + K_{X,L}^{(q)} - K_{X,L}^{(q)}K_{W,L}^{-1}K_{X,L}^{(q)}).$$
The reason this is the ``next best thing" is that the resulting $m(\mathbf{W}_i)+V_i$ has zero mean and covariance matrix 
$$\mathrm{Cov}\{m(\mathbf{W}_i)\}+D + K_{X,L}^{(q)} - K_{X,L}^{(q)}K_{W,L}^{-1}K_{X,L}^{(q)}=K_{X,L}^{(q)}+D=\mathrm{Cov}\{\mathbf{W}^{(q)}_{i}\}.$$
In other words, $\zeta_{i}=m(\mathbf{W}_i)+V_i$ is a 
\begin{center}
``rank $q$ proxy version of $\bm X_i$ + Gaussian measurement error".
\end{center}
whose first and second moments match those of the ideal (but unobservable) bootstrap samples $\mathbf{W}^{(q)}_{i}=\mathbf{X}^{(q)}_i+\bm\delta_i$ (and thus when $X$ and $\varepsilon$ are Gaussian, their laws match, too).

The idea of the bootstrap procedure is to materialise this heuristic, replacing the unknown matrices $\{K_{W,L}^{-1},K^{(q)}_{X,L},D\}$ by their ``hat counterparts" $\{\widehat{K}^{-1}_{W,L},\widehat\Theta_q,\widehat D\}$. In particular, as part of the next theorem, the informal statement that the bootstrap scheme generates samples conforming to $H_{0,q}$ even when $H_{1,q}$ is true will be made rigorous, by means of establishing validity of the bootstrap.
\end{remark}

\begin{theorem}[Bootstrap Validity]  \label{thm3}
Let $q\leq d\leq L_{\dagger}<\infty$ be as in Theorem \ref{theorem-identifiability} and assume that (C) and (H) hold true. Let $p^*_q=F^*_q(T_q)$ be the bootstrapped $p$-value as defined in Step (7) of the bootstrap procedure above. Then, for all $L>L_\dagger$,
\begin{itemize}
\item When $H_{0,q}$ holds true, one has 
$$\mathbb{P}\{p^*_q\leq u\}\stackrel{n\rightarrow\infty}{\longrightarrow}u,\qquad\forall\, u\in [0,1],$$
provided the underlying processes $\{X_i\}$ and errors $\{\varepsilon_{ij}\}$ are Gaussian.

\medskip
\item When $H_{1,q}$ holds true, one has 
$$ \mathbb{P}\{p^*_q \leq u \ \mbox{eventually as} \ n \rightarrow \infty \} = 1,\qquad \forall\, u\in[0,1].$$

\end{itemize}
\end{theorem}

\begin{remark}\label{gaussian-remark}
Regardless of whether or not the $\{\mathbf{W}_i\}$ are Gaussian, as part of the proof of the theorem we establish that under $H_{0,q}$ the (random) bootstrap law $F_q^*$ converges pointwise almost surely to the distribution function of the random variable
$$ \|P_{L} \circ Z_{\dagger}\|_{F}^{2} - 8~(\mathrm{vec}(P_{L} \circ Z_{\dagger}))\transpose \{(C_{0} \otimes I_{L})(\nabla^{2}\Psi(C_{0}))^{-1}(C_{0}\transpose  \otimes I_{L})\}\mathrm{vec}(P_{L} \circ Z_{\dagger}),$$ 
where $C_0$ is as in Assumption (H), and the random vector $Z_{\dagger}$ is the (centred Gaussian) weak limit of $\sqrt{n}\{\frac{1}{n}\sum_{j=1}^{n}\zeta_j \zeta_j\transpose-(\widehat{\Theta}_q + \widehat{D})\}$ under $H_{0,q}$. When  the $\{\mathbf{W}_i\}$ are Gaussian, the covariance of $Z_\dagger$ coincides with that of the centred Gaussian $Z$ encountered in Theorem \ref{thm1}, and so the the bootstrap distribution asymptotically coincides with the limiting law of $nT_q$ under $H_{0,q}$ as given by Theorem \ref{thm1}.

When  the $\{\mathbf{W}_i\}$ are not Gaussian, it is not guaranteed the centred Gaussians $Z_\dagger$ and $Z$ will share the same covariance. Hence the large $n$ limit of  $p^*_q=F^*_q(T_q)$ (given by $G(T_q)$) may not behave as a uniform random variable under $H_{0,q}$, leading to a significance level different than the nominal one. We investigate the potential effect of non-Gaussianity on calibration of the bootstrap in our simulation study (Section \ref{simulations}), and find that this effect is negligible (in fact undetectable). We expect that Gaussianity can be weakened to higher-order moment conditions, at the expense of an even lengthier proof.
\end{remark}

\subsection{Practical Implementation}\label{implementation} {We now discuss practical aspects related to the implementation of our procedure.}

\subsubsection{{Hypothesis Boundary, Grid Size, Bootstrap Parameters}}\label{choice-of-d}

Recall that the global hypothesis pair \eqref{global_hypotheses} to be tested is given by $\{H_0:\mathrm{rank}(k_X) \leq d\}$ versus $\{H_1:\mathrm{rank}(k_X) > d\}$ for some prescribed $d<\infty$. Notice, furthermore, that the bottom-up nature of our iterative testing procedure (Section \ref{section-procedure}) translates to the FWER \emph{remaining invariant} to the choice of boundary value $d$ in the global hypothesis pair $\eqref{global_hypotheses}$. This means that as far as FWER control is concerned, we may choose $d$ as we wish.  Indeed, we are even free to ``data snoop" when choosing $d$ to set up the global hypothesis pair, i.e. formulate our hypothesis boundary by looking at the data. 

The only constraint on the choice of $d$ is the need to ensure that the grid size $L$ is sufficiently large relative to $d$ for our identifiability result (Theorem \ref{theorem-identifiability}) to hold true. As per Remark \ref{grid_remark}, it suffices to have grid size $L\geq 2d+1$ for virtually any type of covariance operator encountered in FDA practice, so it is prudent to always respect the constraint $d\leq \lfloor(L-1)/2\rfloor$. Of course, one can always choose $d$ to be smaller if an inspection of the data suggests so: for instance we can set $d$ to be a value near an elbow of the off-diagonal scree plot\footnote{use of the off-diagonal rather than the classical scree plot is recommended, since the former is immune to the presence of measurement errors when $n$ is large} 
$$j\mapsto T^{(j)}-T^{(j-1)},\qquad j=1,...,L,$$
provided this choice not exceed $\lfloor(L-1)/2\rfloor$.

The value $M$ in Step (3) of the bootstrap procedure can similarly be chosen by inspection of the off-diagonal scree plot, as its formal definition suggests: it should represent an elbow of the graph, but can be taken no larger than our choice of $d$.

\smallskip
\noindent These observations motivate the following practical recommendations:

\begin{itemize}

\item[(I)] The boundary $d$ should be no larger than $\lfloor(L-1)/2\rfloor$.

\medskip
\item[(II)] In particular, $d$ can be chosen empirically, for instance as a value near an elbow of the off-diagonal scree-plot $j\mapsto T^{(j)}-T^{(j-1)}$. 

\medskip
\item[(III)] If the empirical choice is equivocal or exceeds $\lfloor(L-1)/2\rfloor$, we simply recommend fixing $d=\lfloor(L-1)/2\rfloor$. 

\medskip
\item[(IV)] Either way, we recommend setting $M$ in Step (3) of the bootstrap procedure as the minimum of $d$ or a value slightly above an elbow of the off-diagonal scree-plot $j\mapsto T^{(j)}-T^{(j-1)}$.

\end{itemize}

In our simulations, we set $d= \lfloor(L-1)/2\rfloor$ throughout for reasons of automation. As for $M$, we inspected the off-diagonal scree plots from a sample simulation run in each scenario, and fixed the value of $M$  as a value distinctly above an apparent elbow in that run's plot, to accommodate potential variation in other realisations of the plot (unless this exceeded $d$, in which case we took $M=d$). This yielded excellent results irrespectively of the simulation setting.

\subsubsection{Computation}\label{computation}

Recall that evaluation of the test statistic $T^{(j)}$ requires the solution of the optimisation problem
$$ \min_{\Theta^{L \times L}: \ \mathrm{rank}(\Theta) \leq j} \|P_{L} \circ (\widehat{K}_{W,L} - \Theta)\|^{2} = \min_{C \in \mathbb{R}^{L \times j}} \|P_{L} \circ (\widehat{K}_{W,L} - CC\transpose )\|^{2}.$$
This being a \textit{non-convex} optimization problem, we cannot ensure that standard techniques like gradient descent will converge to a global minimum (note that there are infinitely many minima when using the parametrisaion $CC\transpose  $due to the fact that if $C_1$ is a minimum, so is $C_1V$ for any $j \times j$ orthogonal matrix $V$. 

However, recent work by \citet{CW15} shows that projected gradient descent methods with a suitable starting point have a high probability of returning a ``good" local optimum in factorised matrix completion problems. For our simulation study, we used the in-built solver \texttt{optim} in the \texttt{R} software with starting point $C_1 = U_{j}\Sigma_{j}^{1/2}$, where $U \Sigma U\transpose $ is the spectral decomposition of $\widehat{K}_{W,L}$, $U_{j}$ is the matrix obtained by retaining the first $j$ columns of $U$, and $\Sigma_{j}$ is the matrix obtained by keeping the first $j$ rows and columns of $\Sigma$. Although we do not exactly use the approach by \cite{CW15}, it is seen in the simulations that our chosen method of optimisation converges reasonably quickly and yields stable results.

Although our procedure bootstraping a statistic whose value is the solution of a non-convex problem, its implementation was feasible in quite reasonable computational time in all the simulations that we carried out. A single implementation of our bootstrap test procedure, when run on a 64-bit Intel(R) Core(TM) i7-8550U CPU @ 1.80 GHz machine with 16 GB RAM, typically took about 10 seconds when the sample size was $n = 150$ and the grid size was $L = 50$.

\section{Simulation study} \label{simulations}

\indent We will now investigate the finite sample performance of our procedure. Recall that in our notation
$$X(t) = \mu(t) + \sum_{j=1}^{r_{\mathrm{true}}} Y_{j}\varphi_{j}(t), \qquad t \in [0,1]$$
where $\{\varphi_j,\lambda_j\}$ are the eigenfunction/eigenvalue pairs of $k_X$ and the principal component scores $Y_j=\int_0^1 X(u)\varphi_j(u)du$ satisfy $E(Y_{j}) = 0$ and $Var(Y_{j}) = \lambda_{j}$ for all $1 \leq j \leq r_{\mathrm{true}}$. We observe $W_{ij} = X_{i}(t_{j}) + \epsilon_{ij}$ for $1 \leq i \leq n$ and $1 \leq j \leq L$, where $0 < t_{1} < t_{2} < \ldots< t_{L} < 1$ are equispaced grid points. For the purposes of the simulation, the errors $\{\varepsilon_{ij}\}$ are taken to be independent and normally distributed, potentially heteroskedastic in the grid index, $\varepsilon_{ij} \stackrel{\mathrm{i.i.d.}}{\sim} N(0,\sigma_{j}^{2})$ for each $1\leq j\leq L$. We will initially define our simulation scenarios with homoskedastic errors, and in a later section switch to heteroskedastic regimes.

\subsection{Homoskedastic errors} \label{sim-homo}

\indent In the case of homoskedastic measurement errors, we consider the following models (and we comment on their features as we define them):

\begin{description}
\item[\textbf{Model A1}] $r_{\mathrm{true}} = 3$, $\mu(t) = 5(t-0.6)^{2}$, $(\lambda_{1},\lambda_{2},\lambda_{3}) = (0.6,0.3,0.1)$, $Y_{j} \sim N(0,\lambda_{j})$, $\varphi_{1}(t) = 1$, $\varphi_{2}(t) = \sqrt{2}\sin(2{\pi}t)$, $\varphi_{3}(t) = \sqrt{2}\cos(2{\pi}t)$, and $\sigma_{j}^{2} = 1$ for all $j$.  

\smallskip
\item[\textbf{Model A2}] Same as Model A1 except that we now set $\varphi_{3}(t) = \sqrt{2}\cos(4{\pi}t)$, and $Y_{j}$ now has a mixture distribution that is $N(2\sqrt{\lambda_{j}/3},\lambda_{j}/3)$ with probability $1/3$ and $N(-\sqrt{\lambda_{j}/3},\lambda_{j}/3)$ with probability $2/3$. Thus, the $X$-paths are somewhat ``curvier" and the principal component scores follow skewed Gaussian mixture models. The latter is chosen to investigate the behaviour of the bootstrap procedure for non-Gaussian processes (see Remark \ref{gaussian-remark}). 

\smallskip
\item[\textbf{Model A3}] $r_{\mathrm{true}} = 3$, $\mu(t) = 12.5(t-0.5)^{2} - 1.25$, $(\lambda_{1},\lambda_{2},\lambda_{3}) = (4,2,1)$, $Y_{j} \sim N(0,\lambda_{j})$, $\varphi_{1}(t) = 1$, $\varphi_{2}(t) = \sqrt{2}\cos(2{\pi}t)$, $\varphi_{3}(t) = \sqrt{2}\sin(4{\pi}t)$, and $\sigma_{j}^{2} = 2$ for all $j$.  
\smallskip
\item[\textbf{Model A4}] Same Model A3 but with principal component scores having a skewed Gaussian mixture law as in Model A2. 

\smallskip
\item[\textbf{Model A5}] $r_{\mathrm{true}} = 6$, $\mu(t) = 0$, $(\lambda_{1},\lambda_{2},\lambda_{3},\lambda_{4},\lambda_{5},\lambda_{6}) = (4,3.5,3,2.5,2,1.5)$, $Y_{j} \sim N(0,\lambda_{j})$, $\varphi_{1}(t) = 1$, $\varphi_{2k}(t) = \sqrt{2}\sin(2k{\pi}t)$ for $k=1,2,3$, $\varphi_{2k+1}(t) = \sqrt{2}\cos(2k{\pi}t)$ for $k=1,2$, and $\sigma_{j}^{2} = 3$ for all $j$. 
\end{description}

\noindent Models (A1)-(A3) are similar to those considered in \cite{LWC13}. To go beyond globally defined eigenfunctions, the next set of models feature piecewise polynomial eigenfunctions.  

\begin{description}
\item[\textbf{Model S1}] $r_{\mathrm{true}} = 6$, $\mu(t) = 5(t-0.6)^{2}$, $(\lambda_{1},\lambda_{2},\lambda_{3},\lambda_{4},\lambda_{5},\lambda_{6}) = (2,1.7,1.4,1.1,0.8,0.5)$, $Y_{j} \sim N(0,\lambda_{j})$, the eigenfunctions $\varphi_{t}$ are orthonormalised functions obtained from the basis of cubic splines with knots at $(0.3,0.5,0.7)$, and $\sigma_{j}^{2} = 3$ for all $j$. 

\smallskip
\item[\textbf{Model S2}] The model parameters are the same as in Model S1, with the only difference that being that the principal component scores are now distributed according to the skewed Gaussian mixture form in Model A2. 

\smallskip
\item[\textbf{Model S3}] $r_{\mathrm{true}} = 4$, $\mu(t) = 5(t-0.6)^{2}$, $(\lambda_{1},\lambda_{2},\lambda_{3},\lambda_{4}) = (1.4,1.1,0.8,0.5)$, $Y_{j} \sim N(0,\lambda_{j})$, the $\varphi_{t}$'s are orthonormalized functions obtained from the basis of quadratic splines with knots at $(0.2,0.6)$, and $\sigma_{j}^{2} = 2$ for all $j$. 

\smallskip
\item[\textbf{Model S4}] The model parameters are the same as in Model S3 except that now the principal component scores now have a skewed Gaussian mixture distribution as in Model A2. 

\smallskip
\item[\textbf{Model S5}] $r_{\mathrm{true}} = 3$, $\mu(t) = 5(t-0.6)^{2}$, $(\lambda_{1},\lambda_{2},\lambda_{3}) = (1.1,0.8,0.5)$, the $\varphi_{t}$'s are orthonormalized functions obtained from the basis of linear splines with knots at $(0.2,0.6)$, and $\sigma_{j}^{2} = 1$ for all $j$. The principal component scores now have the same skewed Gaussian mixture form as in Model A2.

\end{description}

\indent For each of these models, we have considered two combinations of sample size $n$ and grid size $L$, namely $(n,L)$  equal to $(150,25)$ and $(150,50)$, to emulate more sparsely/densely observed settings.  
 The parameter $M$ described in the bootstrap algorithm in the previous section is set to $M = 10$ for all the simulations in this and the next sub-section. As discussed in the previous section, we can choose $M$ by visual inspection of the off-diagonal scree plot $j\mapsto T^{(j)}-T^{(j-1)}$. When using this approach in a trial runs from each scenario, the plot was suggestive of $M = 9$ for models A5, S1 and S2, and $M = 6$ for the other models. The fixed value of $M = 10$ was thus chosen for use across the simulation scenarios.

\indent To probe the performance of the bootstrap procedure, we set the number of bootstrap samples to $B = 500$ and set the significance level to $\alpha = 0.05$. For each model, we carried out $100$ independent replications to report the empirical distribution of the estimated rank. 
 \indent We benchmark the performance of our procedure with two well-known techniques for selecting the rank of a functional data, 
namely: the AIC based criterion ($AIC_{yao}$) in Sec. 2.5 of \cite{yao2005}; the modified AIC ($AIC_m$) and modified BIC ($BIC_m$) criteria proposed in equations (16) and (6), respectively, in \cite{LWC13}; and the modified information theoretic criteria $PC_{p1}$ and $IC_{p1}$ given in equation (20) in \cite{LWC13}. The information theoretic criteria are inspired by similar techniques in \cite{BN02} who used them to estimate the number of factors in an approximate factor model. We underline that these procedures are used purely for the purpose of benchmarking,  since these are procedures whose purpose is model selection and thus are geared toward inducing parsimony, though some come with theoretical guarantees of consistently selecting the true rank asymptotically, if the rank is truly finite. The results are tabulated in Tables \ref{Tab3}--\ref{Tab6}. 

\begin{table}[!t] 
\caption{Table showing the true rank (in bold) and the empirical distribution of the estimated rank under Models A1--A5 with homoskedastic errors for $(n,L) = (150,25)$}
\label{Tab3}
\begin{center}
\scalebox{0.85}{
\begin{tabular}{c*{18}{c}}
\hline
              & \multicolumn{5}{c}{Model A1} & & \multicolumn{5}{c}{Model A2} & & \multicolumn{5}{c}{Model A3} \\
\hline 
Selected rank & 1 & 2 & \bf{3} & 4 & $\geq 5$& & 1 & 2 & \bf{3} & 4 & $\geq 5$& & 1 & 2 & \bf{3} & 4 & $\geq 5$ \\
\hline
Proposed test      & 0 & 2 & 97 & 0 & 1           & & 0 & 1 & 98 & 1 & 0           & & 0 & 0 & 100 & 0 & 0 \\
$AIC_{yao}$   & 0 & 0 & 13 & 59 & 26         & & 0 & 0 & 16 & 64 & 20           & & 0 & 0 & 74 & 25 & 1 \\
$AIC_m$       & 34 & 54 & 12 & 0 & 0          & & 41 & 53 & 6 & 0 & 0          & & 80 & 20 & 0 & 0 & 0 \\
$BIC_m$       & 0 & 0 & 100 & 0 & 0          & & 0 & 0 & 100 & 0 & 0          & & 0 & 0 & 100 & 0 & 0 \\
$PC_{p1}$     & 67 & 33 & 0 & 0 & 0         & & 70 & 30 & 0 & 0 & 0          & & 99 & 1 & 0 & 0 & 0 \\
$IC_{p1}$     & 53 & 44 & 3 & 0 & 0          & & 55 & 45 & 0 & 0 & 0         & & 92 & 8 & 0 & 0 & 0 \\
\hline 
              & \multicolumn{5}{c}{Model A4} & & \multicolumn{8}{c}{Model A5}              & & & \\
\hline
$\widehat{r}$ & 1 & 2 & \bf{3} & 4 & $\geq 5$& & 1 & 2 & 3 & 4 & 5 & \bf{6} & 7 & $\geq 8$ & & & \\
 \hline
Proposed test      & 0 & 0 & 100 & 0 & 0           & & 0 & 0 & 0 & 0 & 0 & 97 & 2 & 1           & & & \\
$AIC_{yao}$   & 0 & 0 & 68 & 32 & 1           & & 0 & 0 & 0 & 0 & 0 & 100 & 0 & 0           & & & \\
$AIC_m$       & 77 & 23 & 0 & 0 & 0          & & 77 & 19 & 4 & 0 & 0 & 0 & 0 & 0           & & & \\
$BIC_m$       & 0 & 0 & 100 & 0 & 0          & & 0 & 0 & 0 & 0 & 33 & 67 & 0 & 0           & & &\\
$PC_{p1}$     & 92 & 8 & 0 & 0 & 0           & & 89 & 9 & 2 & 0 & 0 & 0 & 0 & 0           & & &\\
$IC_{p1}$     & 90 & 10 & 0 & 0 & 0          & & 90 & 8 & 2 & 0 & 0 & 0 & 0 & 0           & & & \\
 \hline
\end{tabular}}
\end{center}
\end{table}

\begin{table}[!t] 
\caption{Table showing the true rank (in bold) and the empirical distribution of the estimated rank under Models A1--A5 with homoskedastic errors for $(n,L) = (150,50)$}
\label{Tab4}
\begin{center}
\scalebox{0.85}{
\begin{tabular}{c*{18}{c}}
\hline
              & \multicolumn{5}{c}{Model A1} & & \multicolumn{5}{c}{Model A2} & & \multicolumn{5}{c}{Model A3} \\
\hline 
Selected rank & 1 & 2 & \bf{3} & 4 & $\geq 5$& & 1 & 2 & \bf{3} & 4 & $\geq 5$& & 1 & 2 & \bf{3} & 4 & $\geq 5$ \\
\hline
Proposed test      & 0 & 0 & 100 & 0 & 0            & & 0 & 0 & 100 & 0 & 0            & & 0 & 0 & 100 & 0 & 0 \\
$AIC_{yao}$   & 0 & 0 & 0 & 0 & 100          & & 0 & 0 & 0 & 1 & 99           & & 0 & 0 & 0 & 2 & 98 \\
$AIC_m$       & 9 & 38 & 52 & 1 & 0          & & 8 & 37 & 55 & 0 & 0          & & 47 & 46 & 7 & 0 & 0 \\
$BIC_m$       & 0 & 0 & 100 & 0 & 0          & & 0 & 0 & 100 & 0 & 0          & & 0 & 0 & 100 & 0 & 0 \\
$PC_{p1}$     & 41 & 51 & 8 & 0 & 0          & & 37 & 54 & 9 & 0 & 0          & & 76 & 24 & 0 & 0 & 0 \\
$IC_{p1}$     & 21 & 49 & 30 & 0 & 0         & & 12 & 51 & 37 & 0 & 0          & & 60 & 39 & 1 & 0 & 0 \\
\hline  
              & \multicolumn{5}{c}{Model A4} & & \multicolumn{8}{c}{Model A5}              & & & \\
\hline
Selected rank & 1 & 2 & \bf{3} & 4 & $\geq 5$& & 1 & 2 & 3 & 4 & 5 & \bf{6} & 7 & $\geq 8$ & & & \\
 \hline
Proposed test      & 0 & 0 & 100 & 0 & 0            & & 0 & 0 & 0 & 0 & 0 & 100 & 0 & 0           & & &\\
$AIC_{yao}$   & 0 & 0 & 0 & 1 & 99           & & 0 & 0 & 0 & 0 & 0 & 92 & 8 & 0           & & &\\
$AIC_m$       & 37 & 49 & 14 & 0 & 0         & & 24 & 28 & 40 & 6 & 2 & 0 & 0 & 0        & & &  \\
$BIC_m$       & 0 & 0 & 100 & 0 & 0          & & 0 & 0 & 0 & 0 & 0 & 100 & 0 & 0           & & & \\
$PC_{p1}$     & 78 & 22 & 0 & 0 & 0          & & 40 & 40 & 19 & 1 & 0 & 0 & 0 & 0          & & & \\
$IC_{p1}$     & 62 & 36 & 2 & 0 & 0          & & 46 & 37 & 16 & 1 & 0 & 0 & 0 & 0          & & & \\
 \hline
\end{tabular}}
\end{center}
\end{table}

\begin{table}[!t] 
\caption{Table showing the true rank (in bold) and the empirical distribution of the estimated rank under Models S1--S5 with homoskedastic errors for $(n,L) = (150,25)$}
\label{Tab5}
\begin{center}
\scalebox{0.85}{
\begin{tabular}{c*{20}{c}}
\hline
              & \multicolumn{8}{c}{Model S1}         & & \multicolumn{8}{c}{Model S2}         & &\\
\hline 
Selected rank & 1 & 2 & 3 & 4 & 5 & \bf{6} & 7 & $\geq 8$ & & 1 & 2 & 3 & 4 & 5 & \bf{6} & 7 & $\geq 8$ & &\\
\hline
Proposed test      & 0 & 0 & 0 & 0 & 0 & 99 & 1 & 0        & & 0 & 0 & 0 & 0 & 0 & 95 & 4 & 1      & &\\

$AIC_{yao}$   & 0 & 0 & 0 & 0 & 0 & 100 & 0 & 0      & & 0 & 0 & 0 & 0 & 0 & 100 & 0 & 0      & &\\
$AIC_m$       & 31 & 31 & 38 & 0 & 0 & 0 & 0 & 0     & & 33 & 35 & 32 & 0 & 0 & 0 & 0 & 0     & &\\
$BIC_m$       & 0 & 0 & 0 & 0 & 77 & 23 & 0 & 0      & & 0 & 0 & 0 & 2 & 77 & 21 & 0 & 0      & &\\
$PC_{p1}$     & 33 & 36 & 31 & 0 & 0 & 0 & 0 & 0     & & 34 & 37 & 29 & 0 & 0 & 0 & 0 & 0     & &\\
$IC_{p1}$     & 65 & 31 & 4 & 0 & 0 & 0 & 0 & 0     & & 63 & 32 & 5 & 0 & 0 & 0 & 0 & 0    & &\\
\hline  
              & \multicolumn{6}{c}{Model S3} & & \multicolumn{6}{c}{Model S4} & & \multicolumn{5}{c}{Model S5} \\
\hline
Selected rank & 1 & 2 & 3 & \bf{4} & 5 & $\geq 6$ & & 1 & 2 & 3 & \bf{4} & 5 & $\geq 6$ & & 1 & 2 & \bf{3} & 4 & $\geq 5$ \\
 \hline
Proposed test      & 0 & 0 & 0 & 93 & 5 & 2       & & 0 & 0 & 0 & 94 & 4 & 2      & & 0 & 0 & 94 & 4 & 2 \\
$AIC_{yao}$   & 0 & 0 & 0 & 100 & 0 & 0       & & 0 & 0 & 0 & 100 & 0 & 0      & & 0 & 0 & 77 & 22 & 1 \\
$AIC_m$       & 1 & 8 & 56 & 35 & 0 & 0      & & 1 & 8 & 53 & 38 & 0 & 0        & & 1 & 20 & 79 & 0 & 0 \\
$BIC_m$       & 0 & 0 & 0 & 100 & 0 & 0      & & 0 & 0 & 0 & 100 & 0 & 0        & & 0 & 0 & 100 & 0 & 0 \\
$PC_{p1}$     & 1 & 10 & 60 & 29 & 0 & 0      & & 2 & 11 & 55 & 32 & 0 & 0      & & 1 & 20 & 79 & 0 & 0 \\
$IC_{p1}$     & 7 & 14 & 64 & 15 & 0 & 0     & & 16 & 17 & 55 & 12 & 0 & 0     & & 3 & 23 & 74 & 0 & 0 \\
 \hline
\end{tabular}}
\end{center}
\end{table}

\begin{table}[!t] 
\caption{Table showing the true rank (in bold) and the empirical distribution of the estimated rank under Models S1--S5 with homoskedastic errors for $(n,L) = (150,50)$}
\label{Tab6}
\begin{center}
\scalebox{0.85}{
\begin{tabular}{c*{20}{c}}
\hline
              & \multicolumn{8}{c}{Model S1}         & & \multicolumn{8}{c}{Model S2}         & &\\
\hline 
Selected rank & 1 & 2 & 3 & 4 & 5 & \bf{6} & 7 & $\geq 8$ & & 1 & 2 & 3 & 4 & 5 & \bf{6} & 7 & $\geq 8$ & &\\
\hline
Proposed      & 0 & 0 & 0 & 0 & 0 & 100 & 0 & 0        & & 0 & 0 & 0 & 0 & 0 & 100 & 0 & 0      & &\\

$AIC_{yao}$   & 0 & 0 & 0 & 0 & 0 & 83 & 17 & 0      & & 0 & 0 & 0 & 0 & 0 & 80 & 20 & 0      & &\\
$AIC_m$       & 0 & 2 & 3 & 13 & 57 & 25 & 0 & 0     & & 0 & 0 & 7 & 14 & 51 & 28 & 0 & 0     & &\\
$BIC_m$       & 0 & 0 & 0 & 0 & 0 & 100 & 0 & 0      & & 0 & 0 & 0 & 0 & 0 & 100 & 0 & 0      & &\\
$PC_{p1}$     & 0 & 3 & 8 & 17 & 54 & 18 & 0 & 0     & & 0 & 0 & 12 & 21 & 49 & 18 & 0 & 0     & &\\
$IC_{p1}$     & 1 & 11 & 27 & 36 & 25 & 0 & 0 & 0     & & 1 & 10 & 27 & 37 & 25 & 0 & 0 & 0    & &\\
\hline  
              & \multicolumn{6}{c}{Model S3} & & \multicolumn{6}{c}{Model S4} & & \multicolumn{5}{c}{Model S5} \\
\hline
Selected rank & 1 & 2 & 3 & \bf{4} & 5 & $\geq 6$ & & 1 & 2 & 3 & \bf{4} & 5 & $\geq 6$ & & 1 & 2 & \bf{3} & 4 & $\geq 5$ \\
 \hline
Proposed      & 0 & 0 & 0 & 100 & 0 & 0       & & 0 & 0 & 0 & 100 & 0 & 0      & & 0 & 0 & 100 & 0 & 0 \\
$AIC_{yao}$   & 0 & 0 & 0 & 3 & 42 & 55      & & 0 & 0 & 0 & 3 & 43 & 54      & & 0 & 0 & 0 & 1 & 99 \\
$AIC_m$       & 0 & 1 & 16 & 83 & 0 & 0       & & 0 & 0 & 9 & 91 & 0 & 0      & & 0 & 8 & 92 & 0 & 0 \\
$BIC_m$       & 0 & 0 & 0 & 100 & 0 & 0      & & 0 & 0 & 0 & 100 & 0 & 0      & & 0 & 0 & 100 & 0 & 0 \\
$PC_{p1}$     & 0 & 2 & 17 & 81 & 0 & 0      & & 0 & 0 & 9 & 91 & 0 & 0      & & 0 & 7 & 93 & 0 & 0 \\
$IC_{p1}$     & 0 & 2 & 15 & 83 & 0 & 0      & & 0 & 0 & 18 & 82 & 0 & 0      & & 0 & 9 & 91 & 0 & 0 \\
 \hline
\end{tabular}}
\end{center}
\end{table}

\indent It is observed from Tables \ref{Tab3}--\ref{Tab6} that the proposed method selects the true rank in at least $90\%$ of the iterations for \emph{all} of the chosen models, irrespective of whether the true rank is large/small, the observation grid is sparse/dense, the distribution is Gaussian or not, the signal is smooth/rough, and the noise is large/small compared to the signal. In fact, the when $(n,L) = (150,50)$, the bootstrap procedure chooses the true rank in all the $100$ iterations under all of the above simulation models. Moreover, the evidence (as seen from the magnitude of the $p$-values) is quite strong. In cases where the detection of the true rank is not perfect, we found that on making the test procedure more conservative (by choosing a smaller $\alpha$, e.g., $\alpha = 0.01$ or $0.001$), the rate of correct identification of the rank surged to $100\%$.

\indent It is observed from the results shown in Tables \ref{Tab3}--\ref{Tab6} that $AIC_{yao}$ estimates the true rank accurately if the rank is large (equal to $6$ as in Models (A5), (S1) and (S2)). When the rank is small, the performance of $AIC_{yao}$ varies depending on the model. Investigating a bit more, it may be observed that it overestimates the rank under Models A1 and A2 (rank = 3), where the error dominates the leading eigenvalue of the signal. On the other hand, for Models (S3) and (S4) (rank = 4), $AIC_{yao}$ accurately selects the true rank. When the rank is small (equal to $3$ or $4$) but the grid is dense ($L = 50$), it is seen that $AIC_{yao}$ grossly over-estimates the true rank in almost all models. This over-estimation is exacerbated when the eigenfunctions are trigonometric, which is surprising since one would expect this to be an easier setting than in (S3), (S4) and (S5). The over-estimation of the rank by $AIC_{yao}$ was also observed by \cite{LWC13}. \\
\indent The $AIC_m$, $PC_{p1}$ and $IC_{p1}$ criteria do not perform well in general and mostly under-estimate the rank irrespective of the sample size and the sparse/dense regime. The $BIC_m$ procedure, on the other hand, yields the same perfect estimation results as our procedure when the grid is dense ($L = 50$). It does so also when the grid is sparse provided that the true rank is small (equal to $3$ or $4$). However, for Models (S1) and (S2) with $L = 25$, where the rank is large (equal to $6$) and the grid is sparse ($L = 25$), the $BIC_m$ criterion mostly selects a smaller rank. This is different from its performance under Model (A5) with $L = 25$ (which is also of rank $6$), where it selects the true rank in $67\%$ of iterations. The difference in this behaviour of $BIC_m$ may be attributed to the fact that for Model (A5), the eigenfunctions are smooth, while they are only twice continuously differentiable for Models (S1) and (S2) due to the presence of knots. 

\indent Summarising the observations from Tables \ref{Tab3}--\ref{Tab6}, it may be concluded that the $BIC_m$ and the $AIC_{yao}$ criteria are most appropriate among the competing information-based procedures. Some tentative conclusions on the two methods are as follows. While the latter works well when the rank is large (irrespective of the sparsity/denseness of the grid), the former is suited when the grid is dense (irrespective of the magnitude of the rank). The $BIC_m$ procedure also works very well when the grid is sparse, provided that the rank is small. However, both procedures appear to be quite sensitive to departures from the above situations -- $AIC_{yao}$ grossly over-estimates, while $BIC_m$ mildly under-estimates. Note that the difference in performance is observed between $L = 25$ and $L = 50$. This change in number of observations is not so stark so as to be classified immediately as sparse versus dense, and the fact that the performance of these two procedures vary in such a moderate change of grid size is concerning.   
We also mention in passing that the performance of the $AIC_{yao}$ and the $BIC_m$ procedures crucially depend on the choice of the smoothing parameters. Indeed, \cite{LWC13} considered models similar to Models (A1)-(A5) but worked with an undersmoothing choice of the bandwidth parameter, and the relative performance of the above two procedures differs from that observed in our simulation results.

\indent On the other hand, Table \ref{Tab3}--\ref{Tab6} shows that our proposed procedure always selects the true rank in at least $90\%$ of the iterations (the percentage being much higher in most cases), irrespective of the magnitude of the rank and the sparsity/denseness of the grid. Thus, the proposed method seems to provide an effective and stable alternative. Beyond this advantage, our method also comes with a probabilistic guarantee on overestimation, and hence provides an automatic quantification of uncertainty about the true rank, while not relying on smoothing.

\subsection{Heteroskedastic errors}

\indent Our theory suggests that our testing procedure automatically adapts to a heteroskedastic variance structure for the measurement errors. We therefore the same model scenarios as before, but this time with heteroskedastic errors in order to gauge how this translates into practical performance.
All else being the same, the measurement error variances are now given by
$$\sigma_{(p-1)U+k}^{2} = U^{-1} \sum_{l=(p-1)U+1}^{pU} k_{X}(t_{l},t_{l})/1.5,$$
where $U = L/5$, $k=1,2,\ldots,U$ and $p=1,2,\ldots,5$. This specific error structure may be viewed from the perspective of a local averaging of the signal along with a downscaling by a factor of $3/2$. For these simulation models, the results obtained are provided in Tables \ref{Tab-hetero1} to \ref{Tab-hetero4}. It is observed that the performance of the proposed procedure remains invariant to the presence of  homoskedasticity, as our theory predicts.

\begin{table}[!t] 
\caption{Table showing the true rank (in bold) and the empirical distribution of the estimated rank under Models A1--A5 with heteroskedastic errors for $(n,L) = (150,25)$}
\label{Tab-hetero1}
\begin{center}
\scalebox{0.85}{
\begin{tabular}{c*{18}{c}}
\hline
              & \multicolumn{5}{c}{Model A1} & & \multicolumn{5}{c}{Model A2} & & \multicolumn{5}{c}{Model A3} \\
\hline 
Selected rank & 1 & 2 & \bf{3} & 4 & $\geq 5$& & 1 & 2 & \bf{3} & 4 & $\geq 5$& & 1 & 2 & \bf{3} & 4 & $\geq 5$ \\
\hline
Proposed      & 0 & 0 & 95 & 4 & 1           & & 0 & 0 & 94 & 5 & 1           & & 0 & 0 & 93 & 6 & 1 \\
$AIC_{yao}$   & 0 & 0 & 25 & 57 & 18         & & 0 & 0 & 39 & 56 & 5           & & 0 & 0 & 22 & 62 & 16 \\
$AIC_m$       & 21 & 45 & 34 & 0 & 0          & & 27 & 51 & 22 & 0 & 0          & & 93 & 7 & 0 & 0 & 0 \\
$BIC_m$       & 0 & 0 & 100 & 0 & 0          & & 0 & 0 & 100 & 0 & 0          & & 0 & 0 & 100 & 0 & 0 \\
$PC_{p1}$     & 60 & 38 & 2 & 0 & 0         & & 63 & 37 & 0 & 0 & 0          & & 100 & 0 & 0 & 0 & 0 \\
$IC_{p1}$     & 33 & 47 & 20 & 0 & 0          & & 38 & 52 & 10 & 0 & 0         & & 100 & 0 & 0 & 0 & 0 \\
\hline 
              & \multicolumn{5}{c}{Model A4} & & \multicolumn{8}{c}{Model A5}              & & & \\
\hline
Selected rank & 1 & 2 & \bf{3} & 4 & $\geq 5$& & 1 & 2 & 3 & 4 & 5 & \bf{6} & 7 & $\geq 8$ & & & \\
 \hline
Proposed      & 0 & 0 & 94 & 5 & 1           & & 0 & 0 & 0 & 0 & 0 & 100 & 0 & 0           & & & \\
$AIC_{yao}$   & 0 & 0 & 25 & 55 & 20           & & 0 & 0 & 0 & 0 & 0 & 100 & 0 & 0           & & & \\
$AIC_m$       & 95 & 5 & 0 & 0 & 0          & & 86 & 13 & 1 & 0 & 0 & 0 & 0 & 0           & & & \\
$BIC_m$       & 0 & 1 & 99 & 0 & 0          & & 0 & 0 & 0 & 0 & 45 & 55 & 0 & 0           & & &\\
$PC_{p1}$     & 99 & 1 & 0 & 0 & 0           & & 99 & 1 & 0 & 0 & 0 & 0 & 0 & 0           & & &\\
$IC_{p1}$     & 98 & 2 & 0 & 0 & 0          & & 97 & 3 & 0 & 0 & 0 & 0 & 0 & 0           & & & \\
 \hline
\end{tabular}}
\end{center}
\end{table}

\begin{table}[!t] 
\caption{Table showing the true rank (in bold) and the empirical distribution of the estimated rank under Models A1--A5 with heteroskedastic errors for $(n,L) = (150,50)$}
\label{Tab-hetero2}
\begin{center}
\scalebox{0.85}{
\begin{tabular}{c*{18}{c}}
\hline
              & \multicolumn{5}{c}{Model A1} & & \multicolumn{5}{c}{Model A2} & & \multicolumn{5}{c}{Model A3} \\
\hline 
Selected rank & 1 & 2 & \bf{3} & 4 & $\geq 5$& & 1 & 2 & \bf{3} & 4 & $\geq 5$& & 1 & 2 & \bf{3} & 4 & $\geq 5$ \\
\hline
Proposed      & 0 & 0 & 100 & 0 & 0           & & 0 & 0 & 100 & 0 & 0           & & 0 & 0 & 100 & 0 & 0 \\
$AIC_{yao}$   & 0 & 0 & 0 & 0 & 100         & & 0 & 0 & 0 & 4 & 96           & & 0 & 0 & 0 & 1 & 99 \\
$AIC_m$       & 4 & 30 & 65 & 1 & 0          & & 6 & 24 & 70 & 0 & 0          & & 62 & 37 & 1 & 0 & 0 \\
$BIC_m$       & 0 & 0 & 100 & 0 & 0          & & 0 & 0 & 100 & 0 & 0          & & 0 & 0 & 100 & 0 & 0 \\
$PC_{p1}$     & 35 & 53 & 12 & 0 & 0         & & 27 & 51 & 22 & 0 & 0          & & 84 & 16 & 0 & 0 & 0 \\
$IC_{p1}$     & 10 & 41 & 49 & 0 & 0          & & 7 & 44 & 49 & 0 & 0         & & 79 & 21 & 0 & 0 & 0 \\
\hline 
              & \multicolumn{5}{c}{Model A4} & & \multicolumn{8}{c}{Model A5}              & & & \\
\hline
Selected rank & 1 & 2 & \bf{3} & 4 & $\geq 5$& & 1 & 2 & 3 & 4 & 5 & \bf{6} & 7 & $\geq 8$ & & & \\
 \hline
Proposed      & 0 & 0 & 100 & 0 & 0           & & 0 & 0 & 0 & 0 & 0 & 100 & 0 & 0           & & & \\
$AIC_{yao}$   & 0 & 0 & 0 & 1 & 99           & & 0 & 0 & 0 & 0 & 0 & 65 & 35 & 0           & & & \\
$AIC_m$       & 63 & 34 & 3 & 0 & 0          & & 22 & 29 & 40 & 7 & 2 & 0 & 0 & 0           & & & \\
$BIC_m$       & 0 & 0 & 100 & 0 & 0          & & 0 & 0 & 0 & 0 & 0 & 100 & 0 & 0           & & &\\
$PC_{p1}$     & 86 & 14 & 0 & 0 & 0           & & 37 & 40 & 22 & 1 & 0 & 0 & 0 & 0           & & &\\
$IC_{p1}$     & 85 & 15 & 0 & 0 & 0          & & 44 & 38 & 17 & 1 & 0 & 0 & 0 & 0           & & & \\
 \hline
\end{tabular}}
\end{center}
\end{table}

\begin{table}[!t] 
\caption{Table showing the true rank (in bold) and the empirical distribution of the estimated rank under Models S1--S5 with heteroskedastic errors for $(n,L) = (150,25)$}
\label{Tab-hetero3}
\begin{center}
\scalebox{0.85}{
\begin{tabular}{c*{20}{c}}
\hline
              & \multicolumn{8}{c}{Model S1}         & & \multicolumn{8}{c}{Model S2}         & &\\
\hline 
Selected rank & 1 & 2 & 3 & 4 & 5 & \bf{6} & 7 & $\geq 8$ & & 1 & 2 & 3 & 4 & 5 & \bf{6} & 7 & $\geq 8$ & &\\
\hline
Proposed      & 0 & 0 & 0 & 0 & 3 & 97 & 0 & 0        & & 0 & 0 & 0 & 0 & 1 & 96 & 3 & 0      & &\\

$AIC_{yao}$   & 0 & 0 & 0 & 0 & 0 & 100 & 0 & 0       & & 0 & 0 & 0 & 0 & 0 & 100 & 0 & 0      & &\\
$AIC_m$       & 60 & 34 & 6 & 0 & 0 & 0 & 0 & 0       & & 65 & 31 & 4 & 0 & 0 & 0 & 0 & 0     & &\\
$BIC_m$       & 1 & 0 & 0 & 4 & 69 & 23 & 0 & 0       & & 20 & 0 & 0 & 5 & 51 & 24 & 0 & 0      & &\\
$PC_{p1}$     & 64 & 31 & 5 & 0 & 0 & 0 & 0 & 0       & & 77 & 21 & 2 & 0 & 0 & 0 & 0 & 0     & &\\
$IC_{p1}$     & 82 & 18 & 0 & 0 & 0 & 0 & 0 & 0       & & 89 & 11 & 0 & 0 & 0 & 0 & 0 & 0    & &\\
\hline  
              & \multicolumn{6}{c}{Model S3} & & \multicolumn{6}{c}{Model S4} & & \multicolumn{5}{c}{Model S5} \\
\hline
Selected rank & 1 & 2 & 3 & \bf{4} & 5 & $\geq 6$ & & 1 & 2 & 3 & \bf{4} & 5 & $\geq 6$ & & 1 & 2 & \bf{3} & 4 & $\geq 5$ \\
 \hline
Proposed      & 0 & 0 & 0 & 93 & 5 & 2       & & 0 & 0 & 0 & 94 & 6 & 0        & & 0 & 0 & 100 & 0 & 0 \\
$AIC_{yao}$   & 0 & 0 & 0 & 66 & 34 & 0       & & 0 & 0 & 0 & 71 & 29 & 0        & & 0 & 0 & 65 & 34 & 1 \\
$AIC_m$       & 41 & 37 & 22 & 0 & 0 & 0       & & 43 & 34 & 23 & 0 & 0 & 0        & & 2 & 23 & 75 & 0 & 0 \\
$BIC_m$       & 0 & 0 & 0 & 100 & 0 & 0       & & 0 & 0 & 0 & 100 & 0 & 0        & & 0 & 0 & 100 & 0 & 0 \\
$PC_{p1}$     & 38 & 38 & 23 & 1 & 0 & 0       & & 38 & 34 & 27 & 1 & 0 & 0        & & 4 & 23 & 73 & 0 & 0 \\
$IC_{p1}$     & 64 & 31 & 5 & 0 & 0 & 0       & & 66 & 31 & 3 & 0 & 0 & 0        & & 8 & 30 & 62 & 0 & 0 \\
 \hline
\end{tabular}}
\end{center}
\end{table}

\begin{table}[!t] 
\caption{Table showing the true rank (in bold) and the empirical distribution of the estimated rank under Models S1--S5 with heteroskedastic errors for $(n,L) = (150,50)$}
\label{Tab-hetero4}
\begin{center}
\scalebox{0.85}{
\begin{tabular}{c*{20}{c}}
\hline
              & \multicolumn{8}{c}{Model S1}         & & \multicolumn{8}{c}{Model S2}         & &\\
\hline 
Selected rank & 1 & 2 & 3 & 4 & 5 & \bf{6} & 7 & $\geq 8$ & & 1 & 2 & 3 & 4 & 5 & \bf{6} & 7 & $\geq 8$ & &\\
\hline
Proposed      & 0 & 0 & 0 & 0 & 0 & 100 & 0 & 0        & & 0 & 0 & 0 & 0 & 0 & 100 & 0 & 0      & &\\

$AIC_{yao}$   & 0 & 0 & 0 & 0 & 0 & 14 & 57 & 29       & & 0 & 0 & 0 & 0 & 0 & 12 & 62 & 26      & &\\
$AIC_m$       & 0 & 8 & 20 & 35 & 35 & 2 & 0 & 0       & & 1 & 7 & 22 & 37 & 33 & 0 & 0 & 0     & &\\
$BIC_m$       & 0 & 0 & 0 & 0 & 0 & 100 & 0 & 0       & & 0 & 0 & 0 & 0 & 0 & 100 & 0 & 0      & &\\
$PC_{p1}$     & 0 & 9 & 27 & 38 & 23 & 3 & 0 & 0       & & 1 & 9 & 25 & 38 & 27 & 0 & 0 & 0     & &\\
$IC_{p1}$     & 13 & 26 & 39 & 20 & 2 & 0 & 0 & 0       & & 16 & 21 & 43 & 19 & 1 & 0 & 0 & 0    & &\\
\hline  
              & \multicolumn{6}{c}{Model S3} & & \multicolumn{6}{c}{Model S4} & & \multicolumn{5}{c}{Model S5} \\
\hline
Selected rank & 1 & 2 & 3 & \bf{4} & 5 & $\geq 6$ & & 1 & 2 & 3 & \bf{4} & 5 & $\geq 6$ & & 1 & 2 & \bf{3} & 4 & $\geq 5$ \\
 \hline
Proposed      & 0 & 0 & 0 & 100 & 0 & 0       & & 0 & 0 & 0 & 100 & 0 & 0        & & 0 & 0 & 100 & 0 & 0 \\
$AIC_{yao}$   & 0 & 0 & 0 & 0 & 4 & 96       & & 0 & 0 & 0 & 0 & 5 & 95        & & 0 & 0 & 1 & 6 & 93 \\
$AIC_m$       & 0 & 1 & 18 & 81 & 0 & 0       & & 0 & 0 & 13 & 87 & 0 & 0        & & 0 & 4 & 96 & 0 & 0 \\
$BIC_m$       & 0 & 0 & 0 & 100 & 0 & 0       & & 0 & 0 & 0 & 100 & 0 & 0        & & 0 & 0 & 100 & 0 & 0 \\
$PC_{p1}$     & 0 & 2 & 18 & 80 & 0 & 0       & & 0 & 0 & 16 & 84 & 0 & 0        & & 0 & 5 & 95 & 0 & 0 \\
$IC_{p1}$     & 0 & 4 & 34 & 62 & 0 & 0       & & 0 & 2 & 28 & 70 & 0 & 0        & & 0 & 8 & 92 & 0 & 0 \\
 \hline
\end{tabular}}
\end{center}
\end{table}

\subsection{Spiked functional data}  \label{spiked}

\indent One may also consider a spiked covariance model, in analogy to high-dimensional statistics (see, e.g., \cite{John01}, \cite{Paul07}) where some of the eigenvalues are considerably larger than the rest \citep{AW12}. 
One instance of the latter setting is the Tecator data set considered in Section \ref{real-data}. This is a particularly challenging setting: heuristically, a prominent bend is expected to appear in the scree plot, well before the index value of the true rank (see Figure \ref{eigen-plot-spiked} which shows the scree plots for the spiked models considered immediately below). To probe the performance of our method in this setting, we also consider the following spiked scenarios:

\begin{description}
\item[\textbf{Model SF1}] Model (A1) is modified to now have $(\lambda_1,\lambda_2,\lambda_3) = (4,0.2,0.1)$ and $\sigma_{j}^{2} = 1$ for all $j$. Here the first eigenvalue explains about $93\%$ of the total variation in the signal. Note that the error variance is five and ten times the size of the penultimate and last eigenvalue, respectively.
\smallskip
\item[\textbf{Model SF2}] Model (A5) is modified to have $(\lambda_1,\lambda_2,\lambda_3,\lambda_4,\lambda_5,\lambda_6) = (5,4,0.2,0.2,0.1,0.1)$ and $\sigma_{j}^{2} = 1$. Here the top two eigenvalues explain $93.75\%$ of the total variation in the signal, and there are three more trailing eigenvalues are of order between 1/5 and 1/10 the size of the noise variance.

\smallskip
\item[\textbf{Model SF3}] Model (A5) is modified to have $(\lambda_1,\lambda_2,\lambda_3,\lambda_4,\lambda_5,\lambda_6) = (4,3.5,3,0.3,0.2,0.1)$ and $\sigma_{j}^{2} = 3$. Here the top three eigenvalues explain about $95\%$ of the total variation in the signal, and the last three are between 1/10 and 1/300 the size of the of the noise variance. This is a very challenging setup which features several trailing eigenvalues the last of which has size negligible relative to the noise variance.\\

\end{description}

\begin{figure}[!t]
\begin{center}
\includegraphics[scale=0.375]{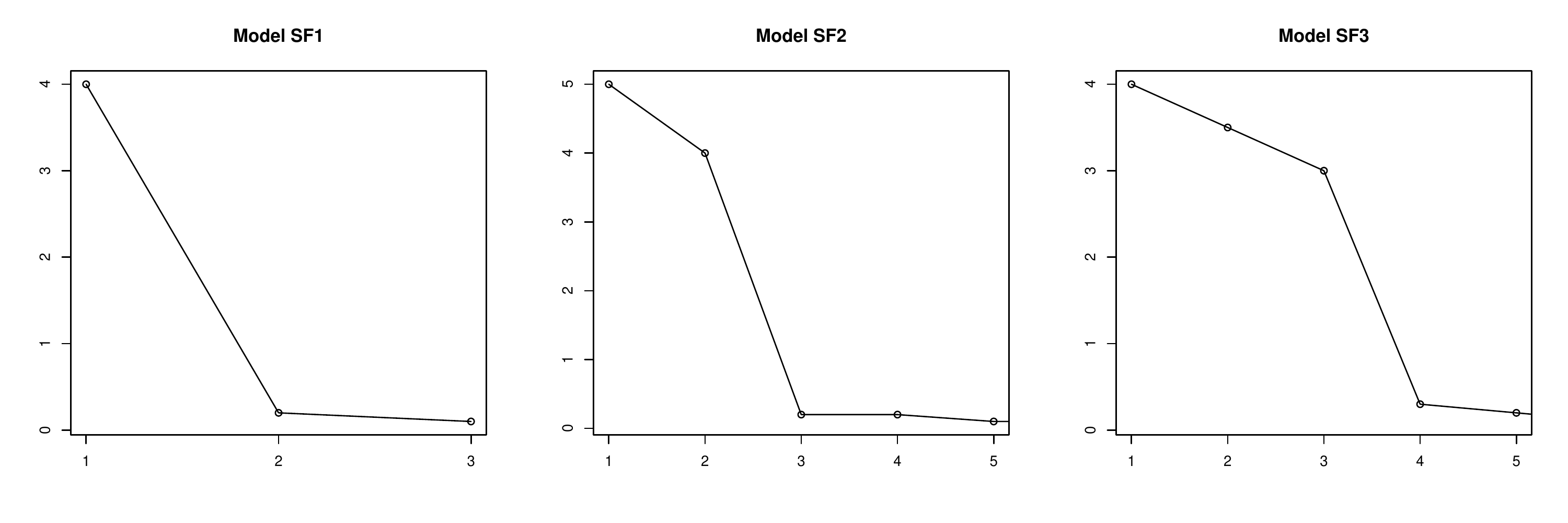}\vspace{-0.1in}
\caption{Scree plots for the spiked models}\label{eigen-plot-spiked}
\end{center}
\end{figure}

Table \ref{Tab-spiked} gives the empirical distribution of the selected rank for each of the above three models when $(n,L) = (150,25)$ and $(150,50)$ -- sparse and dense grids, contrasting our proposed methods with the benchmark methods. The vale of $M=10$ was used again in these settings. Intriguingly, it is observed that the proposed method yields near perfect estimation of the rank under all of the above models. The only exception is the most challenging model (SF3), and this only when the grid is sparse, where our method returns the true rank or the true rank minus in approximately a 50-50 split. Note that this is the scenario where the smallest eigenvalue is 1/300 the size of the error variance \emph{and} the grid is sparse. The results suggests a certain degree of robustness of the proposed procedure against extreme forms of the spectrum. 

By comparison, the benchmark procedures markedly underperform relative to our procedure. The $AIC_m$, $PC_{p1}$ and $IC_{p1}$ procedures yield poor results as in the previous two subsections. The performance of $AIC_{yao}$ performs similarly as in earlier simulations, namely, it does well when the rank is large and the grid is sparse and the error variance is small. Its performance degrades significantly if the grid is dense or the rank is small, which results in either overestimation. When the error variance is not small, $AIC_{yao}$ underestimates the true rank. The most striking difference in performance is observed for the $BIC_m$ procedure, which now heavily underestimates the true rank in the spiked regime. The situation does not improve much even if we take dense grids (here $L = 50$). While this can be explained by the fact that $BIC_m$ is a model selection procedure targeting parsimonious models, {it does also show that the consistency of $BIC_m$ may be slow to manifest in unbalanced spectra}. 

\begin{table}[!t] 
\caption{Table showing the true rank (in bold) and the empirical distribution of the estimated rank under Models SF1--SF3}
\label{Tab-spiked}
\begin{center}
\scalebox{0.8}{
\begin{tabular}{c*{24}{c}}
\hline
\multicolumn{24}{c}{$(n,L) = (150,25)$}\\
\hline
              & \multicolumn{5}{c}{Model SF1} & & \multicolumn{8}{c}{Model SF2} & & \multicolumn{8}{c}{Model SF3} \\
\hline
Selected rank & 1 & 2 & \bf{3} & 4 & $\geq 5$ & & 1 & 2 & 3 & 4 & 5 & \bf{6} & 7 & $\geq 8$ & & 1 & 2 & 3 & 4 & 5 & \bf{6} & 7 & $\geq 8$\\
\hline
Proposed      & 0 & 2 & 98 & 0 & 0        & & 0 & 0 & 0 & 0 & 0 & 100 & 0 & 0        & & 0 & 0 & 0 & 0 & 44 & 49 & 7 & 0 \\
$AIC_{yao}$   & 0 & 0 & 13 & 49 & 38      & & 0 & 0 & 0 & 0 & 6 & 94 & 0 & 0         & & 0 & 0 & 0 & 0 & 66 & 34 & 7 & 0 \\
$AIC_m$       & 100 & 0 & 0 & 0 & 0       & & 54 & 46 & 0 & 0 & 0 & 0 & 0 & 0        & & 59 & 27 & 14 & 0 & 0 & 0 & 0 & 0 \\
$BIC_m$       & 100 & 0 & 0 & 0 & 0       & & 0 & 98 & 2 & 0 & 0 & 0 & 0 & 0         & & 0 & 0 & 100 & 0 & 0 & 0 & 0 & 0 \\
$PC_{p1}$     & 100 & 0 & 0 & 0 & 0       & & 80 & 20 & 0 & 0 & 0 & 0 & 0 & 0        & & 75 & 21 & 4 & 0 & 0 & 0 & 0 & 0 \\
$IC_{p1}$     & 100 & 0 & 0 & 0 & 0       & & 64 & 36 & 0 & 0 & 0 & 0 & 0 & 0        & & 74 & 22 & 4 & 0 & 0 & 0 & 0 & 0 \\
\hline
\hline
\multicolumn{24}{c}{$(n,L) = (150,50)$}\\
\hline
              & \multicolumn{5}{c}{Model SF1} & & \multicolumn{8}{c}{Model SF2} & & \multicolumn{8}{c}{Model SF3} \\
\hline
Selected rank & 1 & 2 & \bf{3} & 4 & $\geq 5$ & & 1 & 2 & 3 & 4 & 5 & \bf{6} & 7 & $\geq 8$ & & 1 & 2 & 3 & 4 & 5 & \bf{6} & 7 & $\geq 8$\\
\hline
Proposed      & 0 & 0 & 100 & 0 & 0       & & 0 & 0 & 0 & 0 & 0 & 100 & 0 & 0        & & 0 & 0 & 0 & 0 & 0 & 100 & 0 & 0 \\
$AIC_{yao}$   & 0 & 0 & 1 & 0 & 99        & & 0 & 0 & 0 & 0 & 0 & 41 & 55 & 4        & & 0 & 0 & 0 & 0 & 0 & 13 & 40 & 47 \\
$AIC_m$       & 82 & 17 & 1 & 0 & 0       & & 22 & 78 & 0 & 0 & 0 & 0 & 0 & 0        & & 10 & 28 & 62 & 0 & 0 & 0 & 0 & 0 \\
$BIC_m$       & 67 & 6 & 27 & 0 & 0       & & 0 & 25 & 1 & 20 & 12 & 42 & 0 & 0      & & 0 & 0 & 1 & 9 & 87 & 3 & 0 & 0 \\
$PC_{p1}$     & 100 & 0 & 0 & 0 & 0       & & 45 & 55 & 0 & 0 & 0 & 0 & 0 & 0        & & 24 & 32 & 44 & 0 & 0 & 0 & 0 & 0 \\
$IC_{p1}$     & 98 & 2 & 0 & 0 & 0        & & 31 & 69 & 0 & 0 & 0 & 0 & 0 & 0        & & 23 & 32 & 45 & 0 & 0 & 0 & 0 & 0 \\
\hline
\end{tabular}}
\end{center}
\end{table}

\subsection{Infinite dimensional models}\label{infinite-simulations}

We now proble the finite sample performance of the procedure when the data are truly infinite dimensional, even prior to noise contamination; and we compare this with the output of model selection-based alternative procedures in such situations. To this aim, we consider infinite dimensional models $X(t) = \sum_{j=1}^{\infty} Y_{j}\varphi_{j}(t), t \in [0,1]$ with $E(Y_{j}) = 0$,  $Var(Y_{j}) = \lambda_{j}>0$ for all $j=1,2,\ldots$. The measurement error again satisfies $\epsilon_{ij} \stackrel{\mathrm{i.i.d}}{\sim} N(0,\sigma_{j}^{2})$ for each $i$. We consider four settings:

\begin{description}
\item[\textbf{Model I1}] $X$ is a standard Brownian motion, which features polynomial decay of eigenvalues and non-differentiable sample paths. Also, $\sigma_{j}^{2} = 1$ for $1 \leq j \leq L$. 

\smallskip
\item[\textbf{Model I2}] $X$ is a Gausian process with $k_X(t,s)= \exp\{-(t-s)^{2}/10\}$ -- which features exponential decay of eigenvalues and infinitely smooth paths. Also, $\sigma_{j}^{2} = 1$ for $1 \leq j \leq L$.  

\smallskip
\item[\textbf{Model I3}] $X$ is as in Model (I1). However, $\sigma_{j}^{2} = t_j$ for $1 \leq j \leq L$, where $t_1 < t_2 \ldots < t_L<1$ is the observation grid.

\smallskip
\item[\textbf{Model I4}] $X$ is as in Model (I2). However, $\sigma_{j}^{2} = t_j$ for $1 \leq j \leq L$, where $0<t_1 < t_2 \ldots < t_L$ is the observation grid. 

\end{description}
 Inspection of the off-diagonal scree plot in a trial run from each scenario suggested no evident elbow below $\lfloor (L-1)/2 \rfloor$, and so as per the recommendations of Section \ref{choice-of-d}, we chose $M=\lfloor (L-1)/2 \rfloor$ in each case. Tables \ref{Tab-infrank-1} and \ref{Tab-infrank-2} give the estimated ranks in $100$ iterations under Models (I1)-(I4) for $(n,L) = (150,25)$ and $(150,50)$. 

It is observed that the model selection procedures like $AIC_{yao}$, $AIC_m$ and $BIC_m$ target some level of parsimonious representation of the data, the degree of parsimony depending on the method used. Unsurprisingly, they fail to inform us on whether the model is truly infinite dimensional or not (similar to having low power error in the testing paradigm). In the majority of cases, regardless of scenario, the chosen rank is between 1 and 3, in fact. By contrast, the proposed method exhibits very good performance in terms of power, typically rejecting low-dimensional representations across all scenarios. In the case of dense grids ($L=50$), the procedure never chose a rank below 15. In the case of a sparser grid ($L=25$), the results varied somewhat between homoskedastic and heteroskedastic noise settings. In the two heteroskedastic scenarios, the procedure chose a rank of at least ten in 75\% and 85\% of runs. In the two homoskedastic scenarios, these percentages were modestly lower at about 56\% and 62\%. When we incorporated the assumption of homoskedasticity in the procedure (as per the comment in Section \ref{subsec2-4}, at the top of p. \pageref{hom}), the  performance surged in the two homoskedastic scenarios, with a rank of at least ten being chosen in  95\% and 96\% of runs. This suggests that, when operating with sparse grids, it can be beneficial in terms of power to make use of homoskedasticity if this can indeed be assumed.

\begin{table}[!t] 
\caption{Table showing the true rank (in bold) and the empirical distribution of the estimated rank under Models I1--I4 with L = 25. {The procedure labelled `Proposed (hom.)' corresponds to our bootstrap procedure, modified to make use of homoskedasticity, as per the comment in Section \ref{subsec2-4}, at the top of p. \pageref{hom}.}}
\label{Tab-infrank-1}
\begin{center}
\scalebox{0.85}{
\begin{tabular}{c*{11}{c}}
\hline
                 & \multicolumn{4}{c}{Model I1}  & & & \multicolumn{4}{c}{Model I2}   \\
\hline 
Selected rank    & 1-3 & 4-6 & 7-9 & $\geq$ 10 & & & 1-3 & 4-6 & 7-9 & $\geq$ 10 \\
\hline			
Proposed         & 2 & 9 & 33 & 56     & & & 0 & 7 & 31 & 62  \\
Proposed (hom.)  & 0 & 0 & 5 & 95      & & & 0 & 0 & 4 & 96   \\
$AIC_{yao}$      & 2 & 97 & 1 & 0       & & & 28 & 71 & 1 & 0  \\
$AIC_m$          & 96 & 4 & 0 & 0       & & & 100 & 0 & 0 & 0  \\
$BIC_m$          & 100 & 0 & 0 & 0      & & & 100 & 0 & 0 & 0  \\
$PC_{p1}$        & 88 & 12 & 0 & 0      & & & 100 & 0 & 0 & 0  \\
$IC_{p1}$        & 100 & 0 & 0 & 0      & & & 100 & 0 & 0 & 0  \\
\hline
\hline
\hline
                 & \multicolumn{4}{c}{Model I3}  & & & \multicolumn{4}{c}{Model I4}   \\
\hline 
Selected rank    & 1-3 & 4-6 & 7-9 & $\geq$ 10 & & & 1-3 & 4-6 & 7-9 & $\geq$ 10 \\
\hline			
Proposed         & 0 & 1 & 24 & 75      & & & 1 & 4 & 11 & 84 \\
$AIC_{yao}$      & 5 & 93 & 2 & 0       & & & 49 & 51 & 0 & 0  \\
$AIC_m$          & 64 & 36 & 0 & 0       & & & 100 & 0 & 0 & 0  \\
$BIC_m$          & 100 & 0 & 0 & 0       & & & 100 & 0 & 0 & 0  \\
$PC_{p1}$        & 42 & 58 & 0 & 0       & & & 100 & 0 & 0 & 0  \\
$IC_{p1}$        & 87 & 13 & 0 & 0       & & & 100 & 0 & 0 & 0  \\
\hline
\end{tabular}
}
\end{center}
\end{table}

\begin{table}[!t] 
\caption{Table showing the true rank (in bold) and the empirical distribution of the estimated rank under Models I1--I4 with L = 50}
\label{Tab-infrank-2}
\begin{center}
\scalebox{0.85}{
\begin{tabular}{c*{15}{c}}
\hline
                 & \multicolumn{6}{c}{Model I1}  & & & \multicolumn{6}{c}{Model I2}   \\
\hline 
Selected rank    & 1-3 & 4-6 & 7-9 & 10-12 & 13-15 & $>$ 15 & & & 1-3 & 4-6 & 7-9 & 10-12 & 13-15 & $>$ 15 \\
\hline			
Proposed         & 0 & 0 & 0 & 0 & 0 & 100      & & & 0 & 0 & 0 & 0 & 0 & 100 \\
$AIC_{yao}$      & 0 & 6 & 92 & 1 & 0 & 0       & & & 1 & 50 & 49 & 0 & 0 & 0  \\
$AIC_m$          & 54 & 46 & 0 & 0 & 0 & 0      & & & 100 & 0 & 0 & 0 & 0 & 0  \\
$BIC_m$          & 100 & 0 & 0 & 0 & 0 & 0      & & & 100 & 0 & 0 & 0 & 0 & 0  \\
$PC_{p1}$        & 30 & 70 & 0 & 0 & 0 & 0      & & & 100 & 0 & 0 & 0 & 0 & 0  \\
$IC_{p1}$        & 87 & 13 & 0 & 0 & 0 & 0      & & & 100 & 0 & 0 & 0 & 0 & 0  \\
\hline
\hline
\hline
                 & \multicolumn{6}{c}{Model I3}  & & & \multicolumn{6}{c}{Model I4}   \\
\hline 
Selected rank    & 1-3 & 4-6 & 7-9 & 10-12 & 13-15 & $>$ 15 & & & 1-3 & 4-6 & 7-9 & 10-12 & 13-15 & $>$ 15 \\
\hline			
Proposed         & 0 & 0 & 0 & 0 & 0 & 100     & & & 0 & 0 & 0 & 0 & 0 & 100 \\
$AIC_{yao}$      & 0 & 30 & 60 & 10 & 0 & 0     & & & 3 & 73 & 24 & 0 & 0 & 0  \\
$AIC_m$          & 15 & 76 & 9 & 0 & 0 & 0      & & & 99 & 1 & 0 & 0 & 0 & 0  \\
$BIC_m$          & 93 & 7 & 0 & 0 & 0 & 0       & & & 100 & 0 & 0 & 0 & 0 & 0  \\
$PC_{p1}$        & 2 & 64 & 34 & 0 & 0 & 0      & & & 100 & 0 & 0 & 0 & 0 & 0  \\
$IC_{p1}$        & 41 & 59 & 0 & 0 & 0 & 0      & & & 100 & 0 & 0 & 0 & 0 & 0  \\
\hline
\end{tabular}
}
\end{center}
\end{table}

\section{Data Analysis} \label{real-data}

\indent We will apply the bootstrap technique for estimating the rank to some benchmark data sets. The first of these is the well-known Tecator dataset which contains spectrometric curves for $n = 215$ samples of finely chopped meat (see \cite{FV06}). Each curve corresponds to the absorbances measured over $L = 100$ wavelengths. A standard functional PCA followed by a scree plot of the eigenvalues reveal an essentially finite dimensional structure since the eigenvalues decay to zero very fast. A scree-plot approach would suggest the underlying rank to be three/four. In fact, the top four eigenvalues are $0.2613$, $0.0024$, $0.0008$ and $0.0003$. The percentage of total variation explained by these principal components are $98.679\%$, $0.901\%$, $0.296\%$ and $0.114\%$, respectively. So the first four eigenvalues explain $99.99\%$ of the total variation.  
Since these data are recorded to high precision, and the curves are very smooth, it may be safely assumed that the measurements are essentially error-free. We will artificially add i.i.d. noise to the data and then apply our method and the alternative procedures considered in the previous section to evaluate their performance. Also, we will vary the error variance to investigate the effect of the magnitude of the signal-to-noise ratio on the rank selection algorithms.  

The errors are taken to be i.i.d. centered Gaussian with variances $1, 0.5, 0.1, 0.05, 0.01, 0.005, 0.001, 0.0005$ and $0.0001$. These values range from ``noise dominating signal completely" to ``noise smaller than fourth largest eigenvalue". For our procedure and each value of the noise variance, we choose $M = 10$ as suggested by the off-diagonal scree plot. \\

\indent Table \ref{Tab7} shows the estimated ranks obtained from the different procedures under the chosen levels of the error variance. It is seen that unless the error variance is very small (comparable to the fourth largest eigenvalue), $AIC_{yao}$ generally chooses unrealistically high values of the rank. In the other situations, the rank is chosen to be one. On the other hand, all of $AIC_m$, $PC_{p1}$ and $IC_{p1}$ select the rank to be one unless the error variance completely overwhelms the signal. The $BIC_m$ procedure always selects the rank as one. These observations can be explained by noting that the Tecator data is an example of a spiked functional dataset and the behaviour of these model selection procedures for such data was found to exhibit such behaviour in Section \ref{spiked}. The procedure proposed in the paper estimates the rank to be three or four in all cases where the error variance is interlaced and comparable with the second/third/fourth eigenvalues. Only when the error variance is very small (1/3 of the the fourth largest eigenvalue), is the rank overestimated (as being six), which is arguably modest a deviation.   

Thus, when the error variance is moderate (neither too small nor overwhelming the signal), only the proposed method seems to provide a proper estimate of the rank of the Tecator data. \\
\begin{table}[!t] 
\caption{Table showing the estimated rank of the Tecator data set under different error variances}
\label{Tab7}
\begin{center}
\scalebox{0.85}{
\begin{tabular}{l*{20}{c}}
\hline
Error variance   & 1 & 0.5 & 0.1 & 0.05 & 0.01 & 0.005 & 0.001 & 0.0005 & 0.0001 \\
\hline
Proposed method  & 2 & 2 & 2 & 2 & 3 & 3 & 4 & 4 & 6 \\
\hline
$AIC_{yao}$      & 7 & 8 & 11 & 12 & 12 & 12 & 9 & 1 & 1 \\
\hline
$AIC_m$          & 2 & 2 & 1 & 1 & 1 & 1 & 1 & 1 & 1 \\
\hline
$BIC_m$          & 1 & 1 & 1 & 1 & 1 & 1 & 1 & 1 & 1 \\
\hline
$PC_{p1}$        & 2 & 2 & 2 & 1 & 1 & 1 & 1 & 1 & 1 \\
\hline
$IC_{p1}$        & 2 & 2 & 1 & 1 & 1 & 1 & 1 & 1 & 1 \\
\hline
\end{tabular}}
\end{center}
\end{table}

\indent The next data set that we consider concerns the number of eggs laid by each of $1000$ female Mediterranean fruit flies (medflies), Ceratitis capitata, in a fertility study described in \cite{CLMWC98}. The data\footnote{Accessible at \texttt{http://anson.ucdavis.edu/$\sim$mueller/data/medfly1000.txt}} contain the total number of eggs laid by each medfly as well as the daily breakup of the number of eggs laid. It is discussed in \cite{CLMWC98} that there is a change in the pattern of egg production at day $51$ post birth for those medflies which lived past that age. Also, the variation in the number of eggs laid from day $51$ onwards is in general much larger than that before day $51$. Taking these observations into account, it seems more pertinent to look at the egg-laying data till the age $50$ days for those medflies that live past that age. This results in a sample of $n = 145$ medflies. Since the number of eggs laid in days $1$ to $3$ for these medflies equal zero, we only keep the number of eggs laid from day $4$ onwards for our analysis. \\
\indent Among the competing procedures, $AIC_{yao}$ estimates the rank of the data to be equal to $9$ while $BIC_m$ selects the rank to be $7$. All of $AIC_m$, $PC_{p1}$ and $IC_{p1}$ select the rank to be one, which appears way off based on a visual inspection of the data. The bootstrap procedure proposed in this paper is carried out by selecting $M = 10$. In fact, the off-diagonal scree plot as well as the results obtained from the competing methods indicate that the rank is likely smaller than $10$. Our procedure selects the rank to be $7$ at significance level $\alpha=1\%$. Further, our bootstrap test rejects the hypotheses $H_{0,q}$ for $q=1,2,\ldots,6$ with $p$-values that are numerically zero.\\
\indent Our procedure thus yields the same result as the $BIC_m$ approach, in this case, in addition to providing a confidence level. We compared the $AIC_m$, the $BIC_m$ and the $AIC_{yao}$ approaches by computing the average relative squared error 
$$ARSE = \frac{1}{n} \sum_{i=1}^{n} \frac{\sum_{j=1}^{L} (W_{ij} - \widehat{X_{i}}(t_{j}))^{2}}{\sum_{j=1}^{L} W_{ij}^{2}},$$
where $\widehat{X_{i}}(\cdot) = \widehat{\mu}(\cdot) + \sum_{j=1}^{\widehat{r}} \widehat{\xi}_{ij}\widehat{\phi}(\cdot)$ is the prediction of $X_{i}(s)$ using the PACE estimates of $\mu$, $\phi$ and $\xi_{ij}$'s (see \cite{yao2005}). For computing the $ARSE$ for each approach, we use the estimated value $\widehat{r}$ of the rank obtained from the corresponding approach. It is found that the $ARSE$ for the $AIC_{yao}$ approach (with $\widehat{r} = 9$) equals $0.200$ and the $ARSE$ for the $BIC_m$ approach (with $\widehat{r} = 7$) is $0.204$. Note that since our approach yields the same estimate of the rank as the $BIC_m$ approach, the $ARSE$ for our approach is also equal to $0.204$. Thus, there is no significant improvement in the $ARSE$ by considering $9$ principal components (obtained using $AIC_{yao}$) instead of $7$ (obtained using our approach or $BIC_m$). The $ARSE$ of the $AIC_m$ approach (as well as that of the $PC_{p1}$ and the $IC_{p1}$ approaches) equals $4.258$. It would seem that these three approaches perform poorly in determining the true rank of the process in this example.

\bibliographystyle{apalike}
\bibliography{biblio1}

\section{Appendix}

\subsection{Proofs of Formal Statements}  \label{proofs}

We will first state and prove some auxiliary results that will simplify the proofs of our main results.

\begin{lemma}\label{matrix-sampling}
If the covariance kernel $k_X$ is continuous and $\mathrm{rank}(k_X)\geq d$, then we can find $u_1<\ldots<u_d$ such that the matrix $\{k_X(u_i,u_j)\}_{i,j=1}^{d}$ is of full rank $d$. 
\end{lemma}

\begin{proof}
Using Mercer's theorem we may write
$$k_X(x_i,x_j)=\sum_{m=1}^{\mathrm{rank}(k_X)}\lambda_m \varphi_m(x_i)\varphi_m(x_j)=\underset{A(x_i,x_j)}{\underbrace{\sum_{m=1}^{d}\lambda_m \varphi_m(x_i)\varphi_m(x_j)}}+\underset{B(x_i,x_j)}{\underbrace{\sum_{m=d+1}^{\mathrm{rank}(k_X)}\lambda_m \varphi_m(x_i)\varphi_m(x_j)}}$$
for any collection of $d$ points $\{x_j\}_{j=1}^{d}$. Note that both $A$ and $B$ are non-negative definite matrices, so it suffices to prove that we can find $(u_1,...,u_d)$ such that the matrix $\{A(u_i,u_j)\}_{i,j=1}^{d}$ is of full rank $d$. We may write
\begin{equation}
A=UU\transpose  
\label{Krep}
\end{equation}
where 
\begin{equation} \label{Umatrix}
U= 
\left(  \begin{array}{cccc}
\lambda_1^{1/2}\varphi_1(x_1)& \hdots & \hdots &\lambda_d^{1/2}\varphi_{d}(x_1)   \\
\lambda_1^{1/2}\varphi_1(x_2) & \hdots & \hdots &\lambda_d^{1/2}\varphi_{d}(x_2)  \\
\vdots &  & &\vdots    \\
\lambda_1^{1/2}\varphi_1(x_d) & \hdots & \hdots &\lambda_d^{1/2}\varphi_{d}(x_d) 
\end{array}
\right)\end{equation}  
and of course $\mathrm{det}(A)=\mathrm{det}^2(U)$. We claim that there exists $d$-tuple such that $\mathrm{det}(U)\neq 0$. For suppose that $\mathrm{det}(U)=0$ for all $(x_1,...,x_d)$. Using the Leibniz formula for the determinant this translates to
$$\sum_{\pi \in \mathrm{Sym}(d)} \mathrm{sgn}(\pi) \prod_{i=1}^d \lambda_{i}^{1/2}\varphi_i (x_{\pi(i)})=0,\qquad\forall (x_1,\ldots,x_d).$$
where $\mathrm{Sym}(d)$ is permutation group on $d$ elements and $\mathrm{sgn}(\pi)$ is the signature of a permutation $\pi$. Keeping $(x_1,\ldots,x_{d-1})$ fixed, multiply both sides of the equation by $\lambda^{1/2}_d\varphi_d(x_d)$ and integrate with respect to $x_d$ to get: 
\begin{eqnarray*}
0&=& \sum_{\pi \in \mathrm{Sym}(d)} \mathrm{sgn}(\pi)\Bigg[ \prod_{i:\pi(i)\neq d}\lambda^{1/2}_{\pi(i)} \varphi_i (x_{\pi(i)})  \Bigg]\langle  \lambda^{1/2}_{\pi^{-1}(d)}\varphi_{\pi^{-1}(d)} , \lambda^{1/2}_d\varphi_d  \rangle\\
&=& \lambda_d\sum_{\pi \in \mathrm{Sym}(d-1)} \mathrm{sgn}(\pi)  \prod_{i=1}^{d-1}\lambda^{1/2}_i \varphi_{i} (x_{\pi(i)}) .
\end{eqnarray*}
Repeating the same process, multiplying both sides of the equation by $\lambda^{1/2}_{d-j}\varphi_j(x_{d-j})$ for $j\in\{1,...,d-1\}$ while keeping the remaining variables fixed, and then integrating with respect to $x_{d-j}$ eventually yields 
$$ \prod_{i=1}^d\lambda_i = 0.$$ 
 This last equality contradicts the fact that $\mathrm{rank}(k_X)\leq d$. Thus $\mathrm{det}(U)\neq 0$ for at least one $d$-tuple, say $(v_1,..,v_d)$. The elements of this $d$-tuple will necessarily be distinct, because $U$ would otherwise have two coincident lines, contradicting $\mathrm{det}(U)\neq 0$, and hence we may re-order them to get the sought $u_1<\ldots<u_d$.
\end{proof}

\begin{corollary}\label{ball-corollary}
If the covariance kernel $k_X$ is continuous and $\mathrm{rank}(k_X)\geq d$, then we can find $u_1<\ldots<u_d$ and $\delta>0$ such:
\begin{enumerate}
\item The balls $B_{\delta}(u_j)=[u_j-\delta,u_j+\delta]$ are pairwise disjoint;

\item The matrix $\{k_X(v_i,v_j)\}_{i,j=1}^{d}$ is of full rank $d$ for all $(v_1,...,v_d)$ such that $v_i\in B_{\delta}(u_i)$. 
\end{enumerate} 
\end{corollary}

\begin{proof}
By Lemma \ref{matrix-sampling} we know that there exist $u_1<\ldots<u_d$ such that $\{k_X(u_i,u_j)\}_{i,j=1}^{d}$ is of full rank. Define the function $\Delta: [0,1]^r\rightarrow\mathbb{R}$ as
 $$\Delta(v_1,...,v_r)=\mathrm{det}\big(\{k_X(v_i,v_j)\}_{i,j=1}^r\big).$$
 Since $k_X$ is uniformly continuous on $[0,1]^2$ 
 it follows that so is $\Delta$ on $[0,1]^r$. Now  
 $$|\Delta(u_1,...,u_r)|>0.$$
 It follows that there exists a $\delta>0$ depending on the modulus of continuity of $k_X$, such that $|\Delta(v_1,...,v_r)|>0$ whenever $|v_j-u_j|<\delta$, i.e. whenever $v_j\in B_{\delta}(u_j)$, the ball of radius $\delta$ centred at $u_j$. Since the $\{u_j\}$ are pairwise distinct, we can take $\delta$ sufficiently small so that the balls $B_{\delta}(u_j)$ are also pairwise disjoint.  
\end{proof}

\begin{proof}[Proof of Proposition \ref{lemma-discrete-rank}]
If $\mathrm{rank}(k_X)\geq d$, we can choose nodes $u_1<\ldots<u_d$  and corresponding balls $\{B_{\delta}(u_j)\}_{j=1}^{d}$ as in the statement of Corollary \ref{ball-corollary}. Since $\{t_1,\ldots,t_L\}$ are regularly spaced nodes for any $L$, there is a finite $L_*$ such that for $L>L_*$ each of the $r$ balls $B_{\delta}(u_j)$ contains at least $1$ grid point. Since the balls are disjoint, one can thus choose a subcollection $\{t_{j_1},...,t_{j_d}\}$ of $d$ distinct grid points such that the matrix $\{k_X(t_{j_p},t_{j_q})\}_{p,q=1}^d$ has full rank $d$, as ensured by Corollary \ref{ball-corollary}. It follows that $K_{X,L}$ has a non-vanishing minor of order $d$, and hence $\mathrm{rank}(K_{X,L})\geq d$.
\end{proof}

The next lemma will be used in the proof of Theorem 1. Informally, it states that a diagonal entry $a_{q,q}$ of an $L\times L$ matrix $A$ of rank $d< L$ can be imputed from the off-diagonal entries of $A$ provided there is a non-vanishing $d$-minor of $A$ that does not depend on $a_{q,q}$.

\begin{lemma}\label{imputation_lemma}
Let $A=\{a_{i,j}\}_{i,j=1}^{L}$ be an $L\times L$ matrix of rank $d< L$ and let $C=\{a_{i_p,j_p}\}_{p=1}^{d}$ be a $d\times d$ submatrix of $A$ that is also of rank $d$. If $i_p\neq q$ and $j_p\neq q$ for all $p\in\{1,\ldots,d\}$, it follows the diagonal element $a_{q,q}$ is uniquely determined as a continuous function of $C$ and the entries $\{a_{q,j_p}\}_{p=1}^{d}\cup\{a_{i_p,q}\}_{p=1}^{d}$.

\end{lemma}

\begin{proof}
Since the statement is invariant to conjugations of $PAP\transpose  $ by permutation matrices $P$, we may assume without loss of generality that $q>i_d$ and $q>j_d$. Thus we may extract a $(d+1)\times (d+1)$ submatrix $D$ of $A$, of the form 

\begin{equation}D=\left[\begin{array}{ccc|c} \,\,\,\,\,\, & u & \,\,\,\,\,\, & \,\, a_{q,q} \\ \hline \,\,\,\, & \,\,\,\,& \,\,\,\, & \,\,\,\, \\ \,\,\,\,\,\, & C & \,\,\,\,\,\, & \,\,v \\ \,\,\,\, & \,\,\,\, & \,\,\,\, & \,\,\,\,\end{array}\right],\label{block_matrix}\end{equation}
where $u=(a_{q,j_1},\ldots,a_{q,j_d})$, $v\transpose  =(a_{i_1,q},...,a_{i_d,q})$. 
Now $\mathrm{det}(D)=0$ because $\mathrm{rank}(A)=d$, so we may write
$$0=\mathrm{det}(D)=\underset{\neq 0}{\underbrace{\mathrm{det}(C)}} (a_{q,q}-u C^{-1}v)$$
showing that $a_{q,q}$ is uniquely determined as a rational function of the entries of $C$, $u$, and $v$.
\end{proof}


\begin{proof}[Proof of Theorem \ref{theorem-identifiability}]
Without loss of generality, we will prove the theorem for the largest possible value of $q$, i.e. for $q=d$. When $\mathrm{rank}(k_X)\geq d$, as is considered in the conclusions (1) and (2) of the Theorem's statement, we can choose $u_1<\ldots<u_d$  and $\{B_{\delta}(u_j)\}_{j=1}^{d}$ as in the statement of Corollary \ref{ball-corollary}. Since $\{t_1,\ldots,t_L\}$ are regularly spaced for any $L$, there is a finite $L_\dagger=L_\dagger(d)$ such that for all $L>L_\dagger$ each of the $d$ balls $B_{\delta}(u_j)$ contain at least $3$ grid points.  It follows that for  any $L\geq L_\dagger$ the matrix $K_{X,L}$ contains a $3d\times 3d$ submatrix $S_{X,L}$, which can be organised into an $d\times d$ matrix of $3\times 3$ blocks, with the property that: any $d\times d$ submatrix of  $S_{X,L}$ extracted by retaining one row from each of the $d$ consecutive triples of rows, and one column from each of the $d$ consecutive triples of columns, has rank $d$ (Figure 3 provides a visualisation).

 \begin{figure}
 \includegraphics[scale=0.2]{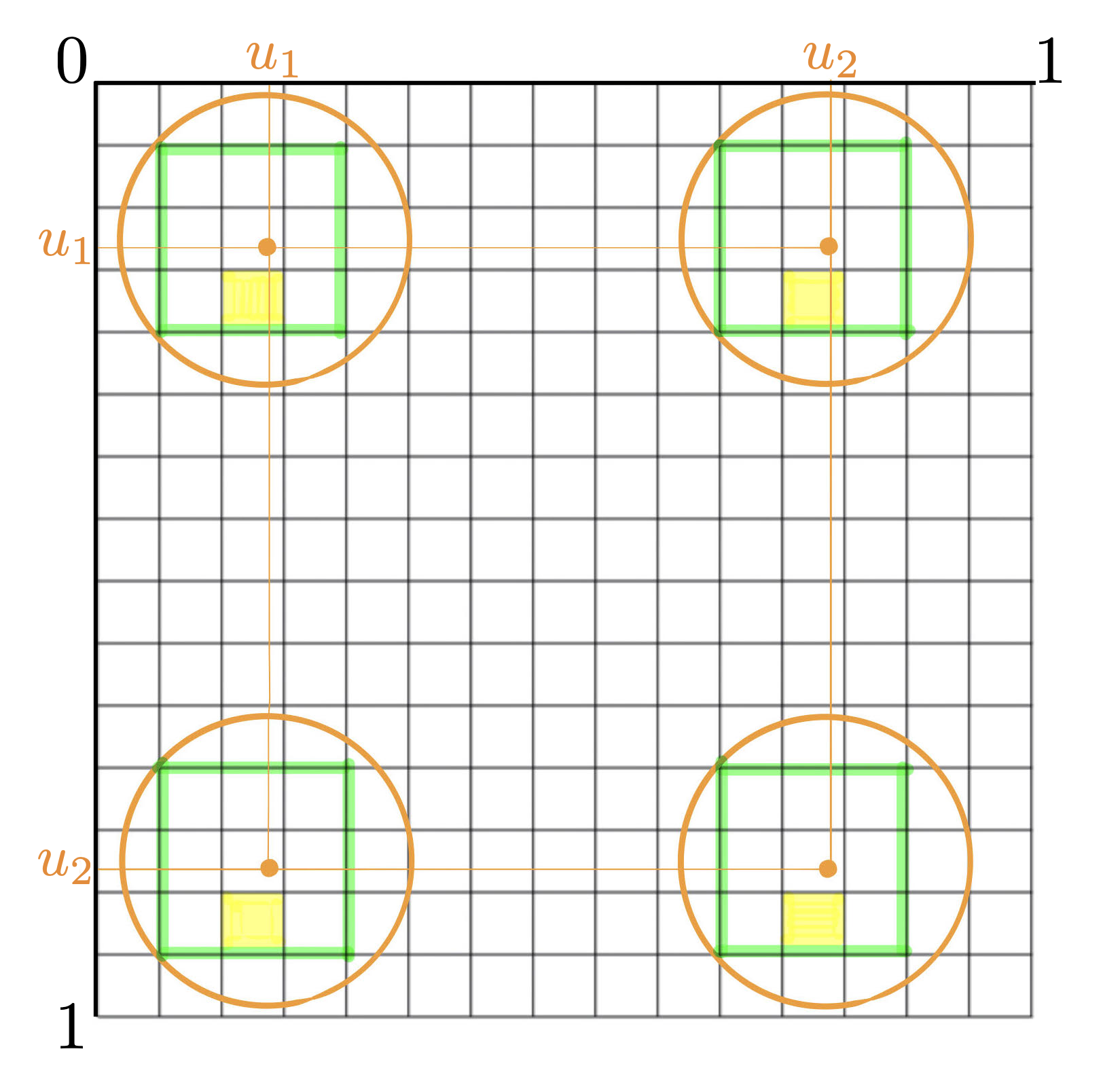}
 \caption{Schematic illustration of the construction of $S_{X,L}$ in a case where $d=2$ and $\mathrm{rank}(k_X)\geq 2$. The four pairs $\{(u_i,u_j)\}_{i,j=1}^{2}\subset[0,1]^2$ are indicated with orange dots, and the radius of the orange circles centred at pairs $(u_i,u_j)$ represents the value $\delta>0$. The overlaid regular grid with $L=15$ suffices to extract the submatrix $S_{X,L}$, which is constituted by the cells bounded by the light green frames. The yellow cells comprise an example of a submatrix of $S_{X,L}$ that stays clear of the diagonal of $K_{X,L}$.}
 \end{figure}
 
Moreover, it is possible to extract such rank-$d$ submatrices of $S_{X,L}$ that contain no diagonal cells of $K_{X,L}$ among their entries (simply by making sure that we do not choose the same order of row and column from corresponding consecutive triples of rows/columns, e.g. always picking the first row from each consecutive row triple and the second column from each consecutive column triple).

Finally, given any specific diagonal element $K_{X,L}(p,p)$ of $K_{X,L}$, we can extract a $d\times d$ submatrix of $S_{X,L}$ of rank-$d$ that: (a) contains no diagonal cells of $K_{X,L}$, and (b) has row/column indices distinct from the index $p$ of the diagonal entry $K_{X,L}(p,p)$. To see this, notice that any submatrix constructed as in the preceding paragraph satisfies both (a) and (b) whenever the index $p$ is not among the row/column indices forming $S_{X,L}$. Otherwise, the diagonal element $K_{X,L}(p,p)$ in question is contained on the diagonal of one of the $d$ blocks of size $3\times 3$ along the diagonal of $S_{X,L}$, say the bottom right block without loss of generality. In this case choose the first row from every row triple and the second column for column triple, except for the last row/column triples. From the last row/column triples: choose the third row and second column of this block, if $K_{X,L}(p,p)$ is the top-left element in that block; choose the first row and third column of this block, if $K_{X,L}(p,p)$ is the central element in that block; choose the first row and third column of this block, if $K_{X,L}(p,p)$ is the bottom-right element in that block. See Figure 4 for an illustration.\\
 
\begin{figure}\label{block}
 \includegraphics[scale=0.1]{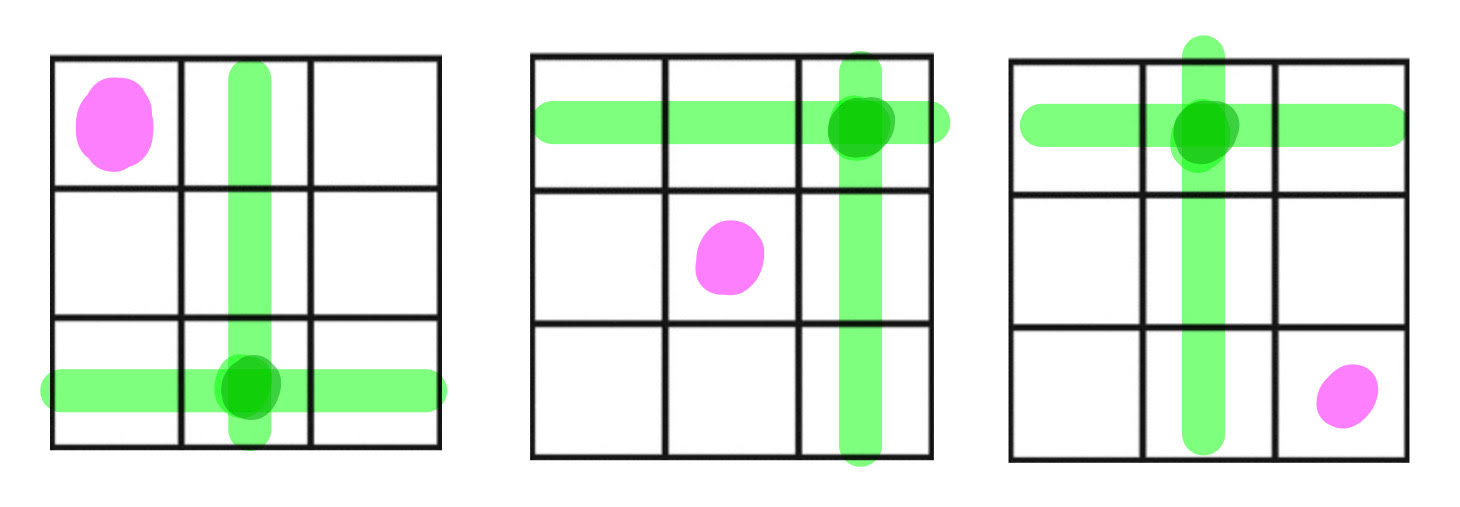}
 \caption{Ilustration of how one chooses elements from a diagonal block of $S_{X,L}$ in order to disinclude row/column $q$ of $K_{X,L}$, when indeed $K_{X,L}(q,q)$ is one of three possible diagonal elements of the said block}
 \end{figure}
  
\noindent Collecting all these facts, we can now see that for $L>L_\dagger$,
 \begin{enumerate}
 
  \item {When $\mathrm{rank}(k_X)=d$}, the matrix $K_{X,L}$ uniquely solves the equation $\|P_L\circ (K_{W,L}-\Theta)\|=0$ among $L\times L$ matrices $\Theta$ of rank $d$,. To see this, let $K_{X,L}(q,q)$ be an element on the diagonal of $K_{X,L}$, for some $q\leq L$. Then, based on the discussion in the previous paragraph, we can find a $d\times d$ submatrix $C_q$ of $S_{X,L}$ of rank $d$ that contains no diagonal elements of $K_{X,L}$ and no elements from the $q$th column or row of $K_{X,L}$. As a result, Lemma \ref{imputation_lemma} implies that we can determine $K_{X,L}(q,q)$ uniquely {as a continuous function of the} elements of $C_q$. Repeating this process for all $q\in \{ 1,...,L\}$ effectively shows that any rank $d$ matrix that coincides with $K_{X,L}$ off the diagonal must also coincide with $K_{X,L}$ on the diagonal {(equivalently that there exists a continuous function $\Xi$, such that $K_{X,L}=\Xi(P_L\circ K_{W,L})$)}.  
 
 \smallskip
 \item {When $\mathrm{rank}(k_X)>d$}, there is no $L\times L$ matrix $\Theta$ of rank less than $d$ such that $\|P_L\circ (K_{W,L}-\Theta)\|_F=0$. This is because there exists a $d\times d$ submatrix $A$ of $S_{X,L}$ of full rank $d$ that contains no diagonal elements of $K_{X,L}$. So $\|P_L\circ (K_{W,L}-\Theta)\|_F=0$ would imply that $\mathrm{rank}(\Theta)\geq d$. Indeed, when $\mathrm{rank}(k_X)>d$, we have the stronger statement that $\inf_{\mathrm{rank}(\Theta)\leq d-1}\|P_L\circ (K_{W,L}-\Theta)\|_F>0$. This is because 
$$\|P_L\circ (K_{W,L}-\Theta)\|_F^2\geq \|A-\Theta_A\|_F^2$$
where $A$ is the rank $d$ submatrix of $K_{X,L}$ as in point 1 above and $\Theta_A$ is the corresponding submatrix of $\Theta$. Since $\Theta$ is of rank at most $d-1$, so is $\Theta_A$. Hence 
$$\|A-\Theta_A\|_F^2\geq |\gamma_d(A)|^2$$
where $\gamma_d(A)\neq 0$ is the $d$-th singular value of $A$. The set of rank $d-1$ matrices being closed now shows that $\inf_{\mathrm{rank}(\Theta)\leq d-1}\|P_L\circ (K_{W,L}-\Theta)\|>0$.\\

 \end{enumerate}
Taken together, statements 1 and 2 above yield the theorem as stated, and thus complete the proof.

\end{proof}

In order to prove Theorm \ref{thm1}, we introduce some additional short hand notation, in the form of the following spaces and functionals:
\begin{eqnarray*}
&&{\cal M}_{q} := \{\Theta \in \mathbb{R}^{L \times L} : \Theta \succeq 0, \mathrm{rank}(\Theta) \leq q\} = \{CC\transpose  : C \in \mathbb{R}^{L \times q}\} \\
&&\pi(\Theta) := \|P_{L} \circ (K_{W,L} - \Theta)\|_{F}, \quad \widehat{\pi}(\Theta) := \|P_{L} \circ (\widehat{K}_{W,L} - \Theta)\|_{F} \\
 && \psi(C) := \pi(CC\transpose )=\|P_{L} \circ (K_{W,L} - CC\transpose )\|_{F}, \quad  {\psi}(C) := \widehat{\pi}(CC\transpose )= \|P_{L} \circ (\widehat{K}_{W,L} - CC\transpose )\|_{F}, \\
&& \Psi(C) := \psi^{2}(C), \quad \widehat{\Psi}(C) := \widehat{\psi}^{2}(C).
\end{eqnarray*}

\begin{proof}[Proof of Theorem \ref{thm1}]
Assumption C guarantees the validity of Theorem \ref{theorem-identifiability}. We will throughout assume that $L$ is finite and fixed, and satisfies $L\geq L_{\dagger}$, where $L_{\dagger}<\infty$ is the critical grid size whose existence is guaranteed by Theorem \ref{theorem-identifiability}. We will divide the proof into several parts. As defined earlier, $\psi(C) = \|P_{L} \circ (K_{X,L} - CC\transpose )\|_{F} = \|P_{L} \circ (K_{W,L} - CC\transpose )\|_{F}$. Also, $\widehat{\psi}(C) = \|P_{L} \circ (\widehat{K}_{W,L} - CC\transpose )\|_{F}$. Then, the test statistic can be written as 
$$T_q = \min_{\Theta^{L \times L} : \mathrm{rank}(\Theta) \leq q} \|P_{L} \circ (\widehat{K}_{W,L} - \Theta)\|_{F}^{2} = \min_{C \in \mathbb{R}^{L \times q}} \widehat{\psi}^{2}(C).$$

The proof will be broken down into the following sequence of steps:
\begin{enumerate}
\item First we will determine the gradients and Hessians of the functionals $\Psi=\psi^2$ and $\widehat\Psi=\widehat\psi^2$.

\item Then we will show the strong consistency any empirical minimiser $\widehat K_{X,L}$ of the functional $\widehat\pi$ for the minimiser $K_{X,L}$ of $\pi$.

\item We will translate this into consistency of an appropriately chosen factor $\widehat C$ of $K_{X,L}$ (i.e. $K_{X,L}=\widehat C\widehat C\transpose  $) to the factor $C_0$ of $K_{X,L}$ defined in Assumption (H).

\item Finally, we will use the penultimate step combined with a Taylor expansion of $\Psi$ and $\widehat\Psi$ in order to determine the sought weak convergence.
\end{enumerate}
\smallskip
\noindent\emph{Step 1:} We begin by determining the gradient and Hessian of $\Psi := \psi^{2}$, denoted by $\nabla\Psi$ and $\nabla^{2}\Psi$, respectively. Since $\psi$ is a real valued function of a matrix, $\nabla\Psi(C)$ is a matrix and $\nabla^{2}\Psi(C)$ is a tensor (Kronecker product). Note that for any $S \in \mathbb{R}^{L \times q}$, we have
\begin{eqnarray*}
\nabla\Psi(C)(S) &=& \langle\nabla\Psi(C),S\rangle_{F}\\
 &=& \lim_{t \rightarrow 0} t^{-1}\{\Psi(C + tS) - \Psi(C)\} \\
&=& \lim_{t \rightarrow 0} t^{-1}\{\|P_{L} \circ (K_{X,L} - (C+tS)(C+tS)\transpose )\|_{F}^{2} - \|P_{L} \circ (K_{X,L} - CC\transpose )\|_{F}^{2}\} \\ 
&=& \lim_{t \rightarrow 0} t^{-1}\{\|P_{L} \circ (K_{X,L} - CC\transpose  - t(CS\transpose  + SC\transpose ) - t^{2}SS\transpose )\|_{F}^{2} \\
&& \hspace{7cm} - \ \|P_{L} \circ (K_{X,L} - CC\transpose )\|_{F}^{2}\} \\
&=& \lim_{t \rightarrow 0} t^{-1}\{\|P_{L} \circ (K_{X,L} - CC\transpose )\|_{F}^{2} + t^{2}\|P_{L} \circ (CS\transpose  + SC\transpose  + tSS\transpose )\|_{F}^{2} \\
&& \hspace{2cm} - \ 2t\langle P_{L} \circ (K_{X,L} - CC\transpose ), P_{L} \circ (CS\transpose  + SC\transpose  + tSS\transpose )\rangle_{F} \\
&& \hspace{7cm} - \ \|P_{L} \circ (K_{X,L} - CC\transpose )\|_{F}^{2}\} \\
&=& -2\langle P_{L} \circ (K_{X,L} - CC\transpose ), P_{L} \circ (CS\transpose  + SC\transpose )\rangle_{F} \\
&=& -2\langle P_{L} \circ (K_{X,L} - CC\transpose ), CS\transpose  + SC\transpose \rangle_{F} \\
&=& -4\langle \{P_{L} \circ (K_{X,L} - CC\transpose )\}C, S\rangle_{F}.
\end{eqnarray*}
The last equality is obtained by using the fact that for a symmetric matrix $A$, we have $\langle A,CS\transpose \rangle_{F} = \mathrm{tr}(SC\transpose A) = \mathrm{tr}(ACS\transpose ) = \langle AC,S\rangle_{F}$ and $\langle A, SC\transpose \rangle = \mathrm{tr}(CS\transpose A) = \mathrm{tr}(ACS\transpose ) = \langle AC,S\rangle_{F}$. Thus, 
$$ \nabla\Psi(C) = -4\{P_{L} \circ (K_{X,L} - CC\transpose )\}C \quad \& \quad \nabla\widehat{\Psi}(C) = -4\{P_{L} \circ (\widehat{K}_{W,L} - CC\transpose )\}C,$$
where $\widehat{\Psi} := \widehat{\psi}^{2}$, and the form of $\nabla\widehat{\Psi}$ follows from the same calculations as above. \\

\noindent To determine the Hessian, we note that for any $R \in \mathbb{R}^{L \times q}$, we have
\begin{eqnarray*}
&&\lim_{t \rightarrow 0} t^{-1}\{\langle\nabla\Psi(C+tR) - \nabla\Psi(C),S\rangle_{F}\} \\
&=& -4\lim_{t \rightarrow 0} t^{-1}[\langle\{P_{L} \circ (K_{X,L} - (C+tR)(C+tR)\transpose )\}(C+tR),S\rangle_{F} - \langle\{P_{L} \circ (K_{X,L} - CC\transpose )\}C,S\rangle_{F}] \\
&=& -4\lim_{t \rightarrow 0} t^{-1}[\langle\{P_{L} \circ (K_{X,L} - CC\transpose  - tCR\transpose  - tRC\transpose  - t^{2}RR\transpose )\}(C+tR),S\rangle_{F} \\
&& \hspace{1cm} - \ \langle\{P_{L} \circ (K_{X,L} - CC\transpose )\}C,S\rangle_{F}] \\
&=& -4[\langle\{P_{L} \circ (K_{X,L} - CC\transpose )\}R,S\rangle_{F} - \langle\{P_{L} \circ (CR\transpose )\}C,S\rangle_{F} - \langle\{P_{L} \circ (RC\transpose )\}C,S\rangle_{F}].
\end{eqnarray*}
Now observe that for any $L \times L$ matrix $A$, we have $P_{L} \circ A = A - \sum_{j=1}^{L} \mathcal{P}_jA\mathcal{P}_j$, where $\mathcal{P}_j$ is the matrix whose $(j,j)$th entry is one and all other entries are zero. So,
\begin{equation}
\begin{aligned}
\langle(P_{L} \circ A)C,S\rangle_{F} &= \langle AC,S\rangle_{F} - \sum_{j=1}^{L} \langle \mathcal{P}_jA\mathcal{P}_jC,S\rangle_{F}\\
&=\langle AC,S\rangle_{F} - \sum_{j=1}^{L} \langle A\mathcal{P}_jC,\mathcal{P}_jS\rangle_{F}. 
\end{aligned}
\label{eq4}
\end{equation}

Next, recall that for compatible matrices $Q_{1}, Q_{2}, Q_{3}$ and $Q_{4}$, we have $\langle Q_{2}Q_{3}Q_{4}\transpose ,Q_{1}\rangle_{F} = \mathrm{tr}(Q_{1}\transpose Q_{2}Q_{3}Q_{4}\transpose ) = (\mathrm{vec}(Q_{1}))\transpose  (Q_{4} \otimes Q_{2}) \mathrm{vec}(Q_{3})$, where $\mathrm{vec}$ denotes the standard vectorization operator. Hence,
\begin{equation}
\left\{
\begin{aligned}
\langle(K_{X,L} - CC\transpose )R,S\rangle_{F} &= \langle(K_{X,L} - CC\transpose )RI_{q},S\rangle_{F}\\
& = (\mathrm{vec}(S))\transpose (I_{q} \otimes (K_{X,L} - CC\transpose )) \mathrm{vec}(R), \\
&
\\
\langle CR\transpose C,S\rangle_{F} &= (\mathrm{vec}(S))\transpose  (C\transpose  \otimes C) \mathrm{vec}(R\transpose )\\
& = (\mathrm{vec}(S))\transpose  \{(C\transpose  \otimes C)M\}\mathrm{vec}(R), \\
&\\
\langle RC\transpose C,S\rangle_{F} &= \langle I_{L}RC\transpose C,S\rangle_{F} \\
&= (\mathrm{vec}(S))\transpose  (C\transpose C \otimes I_{L}) \mathrm{vec}(R),
\end{aligned}
\right.
\label{eq5}
\end{equation}
where $M$ is the commutation matrix of order $(L,q)$, i.e. the permutation matrix satisfying $\mathrm{vec}(R\transpose ) = M\mathrm{vec}(R)$ for $R\in\mathbb{R}^{L\times q}$. Further, $\mathrm{vec}(\mathcal{P}_jR) = \mathrm{vec}(\mathcal{P}_jRI_{q}) = (I_{q} \otimes \mathcal{P}_j) \mathrm{vec}(R)$, which implies that $(\mathrm{vec}(\mathcal{P}_jS))\transpose  = (\mathrm{vec}(S))\transpose  (I_{q} \otimes \mathcal{P}_j)\transpose  = (\mathrm{vec}(S))\transpose  (I_{q}\transpose  \otimes \mathcal{P}_j\transpose ) = (\mathrm{vec}(S))\transpose  (I_{q} \otimes \mathcal{P}_j)$. So,

{
\begin{footnotesize}
\begin{equation} 
\left\{
\begin{aligned}
\langle(K_{X,L} - CC\transpose )\mathcal{P}_jR,\mathcal{P}_jS\rangle_{F} &= \mathrm{vec}(S)\transpose  (I_{q} \otimes \mathcal{P}_j)(I_{q} \otimes (K_{X,L} - CC\transpose ))(I_{q} \otimes \mathcal{P}_j) \mathrm{vec}(R), \\
&\\
\langle CR\transpose C,S\rangle_{F} &= (\mathrm{vec}(S))\transpose  (I_{q} \otimes \mathcal{P}_j)\{(C\transpose  \otimes C)M\}(I_{q} \otimes \mathcal{P}_j) \mathrm{vec}(R), \\
&\\
\langle RC\transpose C,S\rangle_{F} &= (\mathrm{vec}(S))\transpose  (I_{q} \otimes \mathcal{P}_j)(C\transpose C \otimes I_{L})(I_{q} \otimes \mathcal{P}_j) \mathrm{vec}(R).
\end{aligned}
\right.
\label{eq6}
\end{equation}
\end{footnotesize}
}

Observe that $\lim_{t \rightarrow 0} t^{-1}\{\langle\nabla\Psi(C+tR) - \nabla\Psi(C),S\rangle_{F}\}$ equals $\langle\nabla^{2}\Psi(C)\mathrm{vec}(R),\mathrm{vec}(S)\rangle $. Thus, using equations \eqref{eq4}, \eqref{eq5} and \eqref{eq6}, we have
{
\begin{footnotesize}
\begin{eqnarray*}
\nabla^{2}\Psi(C) &=& -4[(I_{q} \otimes (K_{X,L} - CC\transpose )) - (C\transpose  \otimes C)M - (C\transpose C \otimes I_{L}) \\
&& \hspace{0.5cm} - \ \sum_{j=1}^{L} (I_{q} \otimes \mathcal{P}_j)\{(I_{q} \otimes (K_{X,L} - CC\transpose )) - (C\transpose  \otimes C)M - (C\transpose C \otimes I_{L})\}(I_{q} \otimes \mathcal{P}_j)].
\end{eqnarray*}
\end{footnotesize}
}
Now note that $\sum_{j=1}^{L} (I_{q} \otimes \mathcal{P}_j) = I_{q} \otimes (\sum_{j=1}^{L} \mathcal{P}_j) = I_{q} \otimes I_{L} = I_{qL}$. Also, $I_{q} \otimes \mathcal{P}_j$ is the projection matrix onto the rows  $\{j,j+L,j+2L,\ldots,j+(q-1)L\}$ for each $j=1,2,\ldots,L$. Thus, for a matrix $B$ of order $qL$, we have $P_{qL} \circ B = B - \sum_{j=1}^{L} (I_{q} \otimes \mathcal{P}_j)B(I_{q} \otimes \mathcal{P}_j)$. Hence,
$$\nabla^{2}\Psi(C) = -4P_{qL} \circ [(I_{q} \otimes (K_{X,L} - CC\transpose )) - (C\transpose  \otimes C)M - (C\transpose C \otimes I_{L})].$$
Let us also note that for a matrix $A$ of order $L$,
$$ I_{q} \otimes A = 
\begin{bmatrix}
A & 0 & \cdots & 0 \\
0 & A & \cdots & 0 \\
\vdots & \vdots & & \vdots \\
0 & 0 & \cdots & A
\end{bmatrix}
,$$
and $P_{qL} \circ (I_{q} \otimes A)$ sets the diagonal entries of this matrix equal to zero, equivalently, the diagonal entries of each $A$ on the diagonal equal to zero. Thus, $P_{qL} \circ (I_{q} \otimes A) = I_{q} \otimes (P_{L} \circ A)$. Next, for a matrix $E = \{e_{i,j}\}_{i,j=1}^{q}$, we have 
$$ P_{qL} \circ (E \otimes I_{L}) = P_{qL} \circ 
\begin{bmatrix}
e_{11}I_{L} & e_{12}I_{L} & \cdots & e_{1q}I_{L} \\
e_{21}I_{L} & e_{22}I_{L} & \cdots & e_{2q}I_{L} \\
\vdots & \vdots & & \vdots \\
e_{q1}I_{L} & e_{q2}I_{L} & \cdots & e_{qq}I_{L}
\end{bmatrix}
=
\begin{bmatrix}
0 & e_{12}I_{L} & \cdots & e_{1q}I_{L} \\
e_{21}I_{L} & 0 & \cdots & e_{2q}I_{L} \\
\vdots & \vdots & & \vdots \\
e_{q1}I_{L} & e_{q2}I_{L} & \cdots & 0
\end{bmatrix}
= (P_{q} \circ E) \otimes I_{L}.$$
These two observations yield the form of the Hessian as
\begin{equation}\label{the-hessian}
\nabla^{2}\Psi(C) = -4I_{q} \otimes (P_{L} \circ (K_{X,L} - CC\transpose )) + 4P_{qL} \circ \{(C\transpose  \otimes C)M\} + 4(P_{q} \circ C\transpose C) \otimes I_{L}.
\end{equation}

\bigskip
\noindent\emph{{Step 2: Strong Consistency of Empirical Minimizers}}. We will now show that any minimizer $\widehat \Theta$ of the functional $\Theta \mapsto \widehat{\pi}^2(\Theta) := \|P_{L} \circ (\widehat{K}_{W,L} - \Theta)\|_{F}^{2}$ over the space of $L \times L$ matrices $\Theta$ with $\mathrm{rank}(\Theta) \leq q$ is consistent for $K_{X,L}$ as $n\rightarrow\infty$ when $H_{0,q}$ is valid.

To show this, let $\widehat K_{X,L}$ be the (unobservable) random $L\times L$ matrix
$$\widehat K_{X,L}(i,j)=\frac{1}{n}\sum_{m=1}^{n}X_m(t_i)X_m(t_j).$$
Under $H_{0,q}$, $\mathrm{rank}(\widehat K_{X,L})\leq q$. Thus, if $\widehat\Theta$ is a local minimiser of $\widehat\pi$, it must be that 
$$
\underset{\widehat\pi(\widehat\Theta)}{\underbrace{\| P_L\circ \widehat\Theta- P_L\circ\widehat K_{W,L}\|_F}}\leq \underset{\widehat\pi(\widehat K_{X,L})}{\underbrace{\|P_L\circ \widehat K_{X,L} -P_L \circ \widehat K_{W,L}\|_F}}
$$
\begin{equation}\label{empirical-objective-bound}
\leq  \|P_L\circ \widehat K_{X,L} -P_L \circ  K_{W,L}\|_F+\|P_L\circ \widehat K_{W,L} -P_L \circ K_{W,L}\|_F
\end{equation}

The right hand side, however, converges to zero almost surely by the strong law of large numbers and the continuous mapping theorem (note that since $k_X$ is continuous, the covariance operator of the process $X$ is trace-class, and so is any discretization thereof). Since $\widehat K_{W,L}$ converges to $K_{W,L}$ almost surely as $n\rightarrow\infty$, we therefore have
$$ \| P_L\circ \widehat\Theta-  P_L\circ K_{W,L}\|_F\leq \| P_L\circ \widehat\Theta- P_L\circ \widehat K_{W,L}\|_F+\|P_L\circ \widehat K_{W,L} -P_L \circ K_{W,L}\|_F\stackrel{\mathrm{a.s.}}{\longrightarrow}0,$$
as $n\rightarrow\infty$, which implies that 
$$P_L\circ \widehat\Theta \stackrel{\mathrm{a.s.}}{\longrightarrow}  P_L\circ K_{W,L}=P_L\circ K_{X,L}\quad n\rightarrow\infty.$$
This being the case, the event
$$A_n=\{\mathrm{rank}(\widehat B)\geq\mathrm{rank}(B)\,\mbox{for all corresponding submatrices }\widehat B\, \mbox{ of }P_L\circ\widehat \Theta\,\mbox{ and }B\, \mbox{ of }P_L\circ K_{X,L}\} $$
satisfies $P(\lim\inf A_n)=1$. This is because matrices of rank bounded by an integer form a closed set, and thus no sequence $B_n$ can converge almost surely to $B$ unless eventually $\mathrm{rank}(B_n)\geq \mathrm{rank}(B_n)$, almost surely.

Consequently, with probability 1, eventually in $n$, a minor of  $P_L\circ \widehat\Theta_n$ is non-vanishing if and only if the corresponding minor of $P_L\circ K_{X,L}$ is non-vanishing. But $\widehat\Theta$ is always at most rank $q$, by definition. Thus, under $H_{0,q}$, we follow the exact same steps as in the proof of Theorem \ref{theorem-identifiability}, by forming the submatrix $S_{X,L}$ of $K_{X,L}$ and corresponding submatrix $\widehat S_{X,L}$ of $\widehat\Theta$, in order to obtain that
\begin{itemize}
\item $K_{X,L}=\Xi(P_L\circ K_{X,L})$ for a continuous map $\Xi$ (i.e. the diagonal of $K_{X,L}$, and hence the entire matrix $K_{X,L}$ can be determined as a continuous function of the off-diagonal entries).

\item with probability 1, eventually in $n$, $\widehat\Theta=\Xi(P_L\circ \widehat\Theta)$ (i.e. the diagonal of $\Theta$, and hence the entire matrix $\Theta$, can be determined as a continuous function of the off-diagonal entries; indeed the function is the same function we use for $K_{X,L}$).

\end{itemize}
Thus, with probability 1, for all $n$ sufficiently large, $\Xi(P_L\circ\widehat\Theta)$ is well defined, equals $\widehat\Theta$, and thus we
$$\widehat\Theta= \Xi(P_L\circ\widehat\Theta) \stackrel{\mathrm{a.s.}}{\longrightarrow}  \Xi(P_L\circ K_{X,L})= K_{X,L}\quad n\rightarrow\infty. $$
In summary, we have established strong consistency of any minimising sequence $\widehat\Theta$ of $\widehat\pi$.


\bigskip
\noindent\emph{Step 3: Consistency of an Appropriately Chosen Factor} Since $\mathrm{rank}(\widehat{\Theta}) \leq q$, we can write $\widehat{\Theta} = \widecheck{C}\widecheck{C}\transpose $, where $\widecheck{C} \in \mathbb{R}^{L \times q}$. Thus, $\widehat{\pi}(\widehat{\Theta}) = \min_{\Theta} \widehat{\Pi}(\Theta) = \min_{C} \widehat{\Psi}(C) = \widehat{\Psi}(\widecheck{C})$. 
Of course, $\widecheck{C}U$ will also yield the same minimum value for any $q \times q$ orthogonal matrix $U$. 
Since we are not interested in the law of the argument $\widecheck{C}U$ itself, but with the law of the optimised objective $ \widehat{\Psi}(\widecheck{C}U)=\widehat{\Psi}(\widecheck{C})$, we can will work with any choice of $U$, even an ``oracle choice". Define, in particular,
$$\widehat{U} = \arg\min_{\substack{O \in \mathbb{R}^{q \times q}\\OO\transpose  =I_q}} \|\widecheck{C}O - C_{0}\|_{F}.$$
and subsequently, define $\widehat{C} = \widecheck{C}\widehat{U}$. Thus, $\widehat{C}$ is the version of $\widecheck{C}$ that is ``aligned'' with $C_{0}$ in the above sense. The solution above minimization problem, known as the orthogonal Procrustes problem, is given by $\widehat{U} = \widecheck{U}\widecheck{V}\transpose $, where $\widecheck{U}\widecheck{D}\widecheck{V}$ is the singular value decomposition of the matrix $C_{0}\transpose \widecheck{C}$. Furthermore, as earlier remarked, 
$$\min_{C \in \mathbb{R}^{L \times q}} \widehat{\Psi}(C) = \widehat{\Psi}(\widecheck{C}) = \widehat{\Psi}(\widehat{C})$$
so the asymptotic distributions of $\min_{C \in \mathbb{R}^{L \times q}} \widehat{\Psi}(C)$ and $\widehat{\Psi}(\widehat{C})$ will agree. \\

\indent We will now show that $\widehat{C} = \widehat{C}_n$ (to be explicit about the dependence on $n$) converges to $C_{0}$ almost surely as $n \rightarrow \infty$ under $H_{0,q}$. 
Note that since $\widehat{\Theta} = \widehat{\Theta}_{n} = \widehat{C}_{n}\widehat{C}_{n}\transpose $ converges almost surely to $K_{X,L} = C_{0}C_{0}\transpose $, we have that $\|\widehat{C}_{n}\|_{F} = \sqrt{\mathrm{tr}(\widehat{C}_{n}\widehat{C}_{n}\transpose )}$ converges almost surely to $\|C_{0}\|_{F} = \sqrt{\mathrm{tr}(K_{X,L})}$. 
Thus, the set $\Omega = \{\omega : \|\widehat{C}_{n}(\omega)\|_{F} \leq 2\|C_{0}\|_{F} \ \mbox{and} \ \widehat{C}_{n}(\omega)\widehat{C}_{n}(\omega)\transpose  \rightarrow K_{X,L} \ \mbox{as} \ n \rightarrow \infty\}$ has probability measure one. Fix any $\omega \in \Omega$. Since $\widehat{C}_{n}(\omega)$ lies in the compact set (closed ball of radius $2||C_0||_{F}$) for all large $n$, so in order to show that $\widehat{C}_n(\omega)$ converges to $C_0$, we will show that all subsequences of $\widehat{C}_n(\omega)$ converge to $C_0$. Suppose that 
there exists a subsequence $\{k'\}$ of $\{n\}$ such that $\widehat{C}_{k'}(\omega) \rightarrow C_{1}(\omega)$ as $k' \rightarrow \infty$, where $C_{1}(\omega) \neq C_0$. But then $\widehat{C}_{k'}(\omega)\widehat{C}_{k'}(\omega)\transpose  \rightarrow C_{1}(\omega)C_{1}(\omega)\transpose  = K_{X,L}$. Thus, $C_{1}(\omega) = C_{0}V(\omega)$ for some $q \times q$ orthogonal matrix $V(\omega)$.  Suppose that $C_{1}(\omega) \neq C_{0}$, equivalently, $V(\omega) \neq I_{q}$. Define $\widehat{C}_{k'}^{(0)}(\omega) = \widehat{C}_{k'}(\omega)V(\omega)\transpose $ for each $k' \geq 1$. Then,
$$\|\widehat{C}_{k'}^{(0)}(\omega) - C_{0}\|_{F} = \|\widehat{C}_{k'}(\omega)V(\omega)\transpose  - C_{0}\|_{F} \rightarrow \|C_{1}(\omega)V(\omega)\transpose  - C_{0}\|_{F} = 0.$$
On the other hand, $\|\widehat{C}_{k'}(\omega) - C_{0}\|_{F} \rightarrow \|C_{1}(\omega) - C_{0}\|_{F} > 0.$ Recall that $\widehat{C}_{k'}(\omega) = \widecheck{C}_{k'}(\omega)U(\omega)$ as per our construction. So, there exists $k'_{0} \geq 1$ such that 
\begin{eqnarray*}
\|\widecheck{C}_{k'_{0}}(\omega)\widehat{U}(\omega)V(\omega)\transpose  - C_{0}\|_{F} &=& \|\widehat{C}_{k'_{0}}^{(0)}(\omega) - C_{0}\|_{F} \\
&<& (1/4)\|C_{1}(\omega) - C_{0}\|_{F} \\
&<& \|\widehat{C}_{k'_{0}}(\omega) - C_{0}\|_{F} \ = \ \min_{\substack{O \in \mathbb{R}^{q \times q}\\OO\transpose = I_q}} \|\widecheck{C}_{k'_{0}}(\omega)O - C_{0}\|_{F}.
\end{eqnarray*}
This leads to a contradiction unless $V(\omega) = I_{q}$. Hence, $C_{1}(\omega) = C_{0}$ so that the limit does not depend on $\omega \in \Omega$. A standard subsequence argument (using the fact that $\widehat{C}_{n}$ lies in a compact set for all sufficiently large $n$ almost surely) now shows that the entire sequence $\widehat{C}_{n}$ must converge to $C_{0}$ as $n \rightarrow \infty$ on $\Omega$. \\

\noindent\emph{Step 4: Determination of the Asymptotic Distribution} We will now proceed to derive the asymptotic distribution of $\min_{C} \widehat{\Psi}(C)$ under $H_{0}$. First, observe that $\nabla^{2}\widehat{\Psi}(C) - \nabla^{2}\Psi(C)$ equals 
$$-4I_{q} \otimes (P_{L} \circ (\widehat{K}_{W,L} - K_{X,L})) = -4I_{q} \otimes (P_{L} \circ (\widehat{K}_{W,L} - K_{W,L}))$$ and is thus not dependent on $C$. Furthermore, it is $O_{\mathbb{P}}(n^{-1/2})$ as $n \rightarrow \infty$.

\smallskip
\noindent Observe that by Taylor's theorem,
\begin{eqnarray*}
\mathrm{vec}(\nabla\widehat{\Psi}(\widehat{C})) &=& \mathrm{vec}(\nabla\widehat{\Psi}(C_{0})) + \nabla^{2}\widehat{\Psi}(\widetilde{C})\mathrm{vec}(\widehat{C} - C_{0}) \\
\Rightarrow \ o_{\mathbb{P}}(1) &=& -4\mathrm{vec}(\{P_{L} \circ \sqrt{n}(\widehat{K}_{W,L} - K_{W,L})\}C_{0}) - 4I_{q} \otimes (P_{L} \circ \sqrt{n}(\widehat{K}_{W,L} - K_{W,L}))\mathrm{vec}(\widehat{C} - C_{0}) \\
&& + \ \nabla^{2}\Psi(\widetilde{C})(\sqrt{n}\mathrm{vec}(\widehat{C} - C_{0})),
\end{eqnarray*}
where $\widetilde{C} = \alpha\widehat{C} + (1-\alpha)C_{0}$ for some $0 < \alpha < 1$. Since $\sqrt{n}(\widehat{K}_{W,L} - K_{W,L})$ converges weakly to a centered Gaussian random matrix $Z$, and we have shown above that $\widehat{C} \rightarrow C_{0}$ in probability, it follows that
\begin{eqnarray} 
\nabla^{2}\Psi(\widetilde{C})(\sqrt{n}\mathrm{vec}(\widehat{C} - C_{0})) = 4\mathrm{vec}(\{P_{L} \circ \sqrt{n}(\widehat{K}_{W,L} - K_{W,L})\}C_{0}) + o_{\mathbb{P}}(1) \label{eq7}
\end{eqnarray}  
as $n \rightarrow \infty$.  Now $\nabla^{2}\Psi(C_{0})$ is non-singular by Assumption (H), so the inverse function theorem applied to the function $\nabla\Psi$ implies that:
\begin{itemize}
\item[(i)] the function $\nabla\Psi$ is invertible in a neighbourhood of $C_{0}$.
\item[(ii)] the function $(\nabla\Psi)^{-1}$ is continuously differentiable in that same neighbourhood.
\item[(iii)] $\nabla((\nabla\Psi)^{-1})(\nabla\Psi(C)) = (\nabla^{2}\Psi(C))^{-1}$ in that same neighbourhood. 
\end{itemize}
Since $\widehat{C} \rightarrow C_{0}$ in probability, we have $\mathbb{P}(\widetilde{C} \ \mbox{lies in the above neighbourhood}) \rightarrow 1$. Also, from the fact that $(\nabla\Psi)^{-1}$ is continuously differentiable in that neighbourhood, it follows that $\mathbb{P}(\nabla^{2}\Psi(\widetilde{C}) \ \mbox{is invertible}) \rightarrow 1$. Further,
\begin{eqnarray}
(\nabla^{2}\Psi(\widetilde{C}))^{-1} = \nabla((\nabla\Psi)^{-1})(\nabla\Psi(\widetilde{C})) \rightarrow \nabla((\nabla\Psi)^{-1})(\nabla\Psi(C_{0})) = (\nabla^{2}\Psi(C_{0}))^{-1} \label{eq8}
\end{eqnarray}
in probability as $n \rightarrow \infty$ by the continuity of $\nabla\Psi$ and the continuous differentiability of $(\nabla\Psi)^{-1}$. It now follows from \eqref{eq7} and \eqref{eq8} that
\begin{eqnarray}
\sqrt{n}\{\mathrm{vec}(\widehat{C}) - \mathrm{vec}(C_{0})\} &=& 4(\nabla^{2}\Psi(C_{0}))^{-1}\mathrm{vec}(\{P_{L} \circ \sqrt{n}(\widehat{K}_{W,L} - K_{W,L})\}C_{0}) + o_{\mathbb{P}}(1) \label{eq9} \\
&\stackrel{d}{\rightarrow}& 4(\nabla^{2}\Psi(C_{0}))^{-1}\mathrm{vec}(\{P_{L} \circ Z\}C_{0}) \label{eq10}
\end{eqnarray}
as $n \rightarrow \infty$. \\
\indent Next note that for some $\widetilde{C}_{1} = \beta\widehat{C} + (1-\beta)C_{0}$ with $0 < \beta < 1$, we have
\begin{eqnarray*}
\widehat{\Psi}(\widehat{C}) &=& \widehat{\Psi}(C_{0}) + \langle \mathrm{vec}(\nabla\widehat{\Psi}(C_{0})),\mathrm{vec}(\widehat{C} - C_{0})\rangle + \frac{1}{2}\langle \nabla^{2}\widehat{\Psi}(\widetilde{C}_{1})\mathrm{vec}(\widehat{C} - C_{0}),\mathrm{vec}(\widehat{C} - C_{0})\rangle \nonumber\\
&=& \widehat{\Psi}(C_{0}) + \langle \mathrm{vec}(\nabla\widehat{\Psi}(C_{0})),\mathrm{vec}(\widehat{C} - C_{0})\rangle + \frac{1}{2}\langle \nabla^{2}\Psi(\widetilde{C}_{1})\mathrm{vec}(\widehat{C} - C_{0}),\mathrm{vec}(\widehat{C} - C_{0})\rangle \nonumber \\
&& - \ 2\langle \{I_{q} \otimes (P_{L} \circ (\widehat{K}_{W,L} - K_{W,L}))\}\mathrm{vec}(\widehat{C} - C_{0}),\mathrm{vec}(\widehat{C} - C_{0})\rangle.
\end{eqnarray*}
It now follows from equations \eqref{eq8}, \eqref{eq9} and \eqref{eq10} that
\begin{eqnarray*}
\widehat{\Psi}(\widehat{C}) &=& \|P_{L} \circ (\widehat{K}_{W,L} - K_{X,L})\|_{F}^{2} \\
&& - \ 16~\langle \mathrm{vec}(\{P_{L} \circ (\widehat{K}_{W,L} - K_{W,L})\}C_{0}),[(\nabla^{2}\Psi(C_{0}))^{-1}\mathrm{vec}(\{P_{L} \circ (\widehat{K}_{W,L} - K_{W,L})\}C_{0}) + o_{\mathbb{P}}(n^{-1/2})]\rangle \\
&& + \ 8~\langle \nabla^{2}\Psi(\widetilde{C}_{1})[(\nabla^{2}\Psi(C_{0}))^{-1}\mathrm{vec}(\{P_{L} \circ (\widehat{K}_{W,L} - K_{W,L})\}C_{0}) + o_{\mathbb{P}}(n^{-1/2})],\\
&& \hspace{1cm} [(\nabla^{2}\Psi(C_{0}))^{-1}\mathrm{vec}(\{P_{L} \circ (\widehat{K}_{W,L} - K_{W,L})\}C_{0}) + o_{P}(n^{-1/2})]\rangle + o_{\mathbb{P}}(n^{-1}) \\
\Rightarrow n\widehat{\Psi}(\widehat{C}) &=& \|P_{L} \circ \sqrt{n}(\widehat{K}_{W,L} - K_{X,L})\|_{F}^{2} \\
&& - \ 16~\langle \mathrm{vec}(\{P_{L} \circ \sqrt{n}(\widehat{K}_{W,L} - K_{W,L})\}C_{0}),[(\nabla^{2}\Psi(C_{0}))^{-1}\mathrm{vec}(\{P_{L} \circ \sqrt{n}(\widehat{K}_{W,L} - K_{W,L})\}C_{0}) + o_{\mathbb{P}}(1)]\rangle \\
&& + \ 8~\langle \nabla^{2}\Psi(\widetilde{C}_{1})[(\nabla^{2}\Psi(C_{0}))^{-1}\mathrm{vec}(\{P_{L} \circ \sqrt{n}(\widehat{K}_{W,L} - K_{W,L})\}C_{0}) + o_{\mathbb{P}}(1)],\\
&& \hspace{1cm} [(\nabla^{2}\Psi(C_{0}))^{-1}\mathrm{vec}(\{P_{L} \circ \sqrt{n}(\widehat{K}_{W,L} - K_{W,L})\}C_{0}) + o_{P}(1)]\rangle + o_{\mathbb{P}}(1) \\
&=& \|P_{L} \circ \sqrt{n}(\widehat{K}_{W,L} - K_{W,L})\|_{F}^{2} \\
&& - \ 16~\langle \mathrm{vec}(\{P_{L} \circ \sqrt{n}(\widehat{K}_{W,L} - K_{W,L})\}C_{0}),(\nabla^{2}\Psi(C_{0}))^{-1}\mathrm{vec}(\{P_{L} \circ \sqrt{n}(\widehat{K}_{W,L} - K_{W,L})\}C_{0})\rangle \\
&& + \ 8~\langle \nabla^{2}\Psi(\widetilde{C}_{1})(\nabla^{2}\Psi(C_{0}))^{-1}\mathrm{vec}(\{P_{L} \circ \sqrt{n}(\widehat{K}_{W,L} - K_{W,L})\}C_{0}),\\
&& \hspace{1cm} (\nabla^{2}\Psi(C_{0}))^{-1}\mathrm{vec}(\{P_{L} \circ \sqrt{n}(\widehat{K}_{W,L} - K_{W,L})\}C_{0})\rangle + o_{\mathbb{P}}(1) \\
&\stackrel{d}{\rightarrow}& \|P_{L} \circ Z\|_{F}^{2} - 8~\langle \mathrm{vec}(\{P_{L} \circ Z\}C_{0}),(\nabla^{2}\Psi(C_{0}))^{-1}\mathrm{vec}(\{P_{L} \circ Z\}C_{0})\rangle
\end{eqnarray*}
as $n \rightarrow \infty$. Thus, as argued earlier, we have
\begin{eqnarray}
nT_q &=& n\min_{C \in \mathbb{R}^{L \times q}} \widehat{\Psi}(C) \nonumber\\
&\stackrel{d}{\rightarrow}& \|P_{L} \circ Z\|_{F}^{2} - 8~\langle \mathrm{vec}(\{P_{L} \circ Z\}C_{0}),(\nabla^{2}\Psi(C_{0}))^{-1}\mathrm{vec}(\{P_{L} \circ Z\}C_{0})\rangle \nonumber \\
&=& \|P_{L} \circ Z\|_{F}^{2} - 8~(\mathrm{vec}(P_{L} \circ Z))\transpose \{(C_{0} \otimes I_{L})(\nabla^{2}\Psi(C_{0}))^{-1}(C_{0}\transpose  \otimes I_{L})\}\mathrm{vec}(P_{L} \circ Z)   \label{eq11}
\end{eqnarray}
as $n \rightarrow \infty$.

The proof of the last statement will follow by establishing that
$$T_q \stackrel{a.s.}{\rightarrow} \inf_{\Theta\in\mathcal{M}_q} \|P_{L} \circ (K_{W,L} - \Theta)\|_{F}$$
as $n \rightarrow \infty$, the latter term being strictly positive under $H_{1,q}$ by statement (2) of Theorem \ref{theorem-identifiability}. To establish this limit, apply the reverse triangle inequality to obtain
\begin{eqnarray*}
\|P_L\circ \Theta - P_L\circ \widehat K_{W,L} \|_F&=&\|P_L\circ \Theta -P_L\circ K_{W,L}+P_L\circ K_{W,L} - P_L\circ \widehat K_{W,L} \|_F\\
&=&\|(P_L\circ \Theta -P_L\circ K_{W,L})- (P_L\circ \widehat K_{W,L}-P_L\circ K_{W,L}) \|_F\\
&\geq&\left|\|(P_L\circ \Theta -P_L\circ K_{W,L})\|_F-\| (P_L\circ \widehat K_{W,L}-P_L\circ K_{W,L}) \|_F \right|
\end{eqnarray*}
Thus
$$\inf_{\Theta\in\mathcal{M}_q}\| P_L\circ \Theta - P_L\circ  \widehat{K}_{W,L}\|_F\geq \inf_{\Theta\in\mathcal{M}_q}\left|\|(P_L\circ \Theta -P_L\circ K_{W,L})\|_F-\| (P_L\circ \widehat K_{W,L}-P_L\circ K_{W,L}) \|_F \right|$$
and notice that the right hand side converges almost surely to $\inf_{\Theta\in\mathcal{M}_q}\| P_L\circ \Theta - P_L\circ  {K}_{W,L}\|_F$. It follows that $nT_q$ diverges almost surely as $n \rightarrow \infty$, and the proof is complete.
\end{proof}

\begin{proof}[Proof of Theorem \ref{thm3}]

The proof will be broken down into the following  series of steps: 
\begin{enumerate}
\item  We will first show that the quantities $\widehat{\Theta}_q$, $\widehat{K}_{W,L}^{-1}$, $M$, $\widehat{\Theta}_{M}$, and $\widehat{D}$ either converge $\mathbb{P}$-almost surely as $n\rightarrow\infty$, or are $\mathbb{P}$-strongly tight as $n\rightarrow\infty$ (meaning they eventually lie in a compact set with $\mathbb{P}$-probability 1). This will be broken down into three sub-cases: when $H_{0,q}$ is true; when $H_{1,q}$ \emph{and} $H_{0}$ are both true; and, finally, when $H_1$ is true.
\item We will then prove statement (a) in the present theorem under the assumption that $H_{0,q}$ is true. This proof will be further broken down into the following five sub-steps: 
\begin{itemize}
\item[2a.] We will first show that the minimizer, say $\Theta^{*}_q$, of the bootstrap functional $\Theta \mapsto ||P_L \circ (\overline{K}_{\zeta,L} - \Theta)||_F$ over ${\cal M}_q$ converges to $K_{X,L}$ $\mathbb{P}^{*}$-almost surely, $\mathbb{P}$-almost surely. 
\item[2b.] Then we will show that an appropriately chosen rank $q$ factorization of $\Theta^{*}_q$ converges to the rank $q$ factorization $C_0$ of $K_{X,L}$ defined in Assumption (H). 
\item[2c.] We will next derive the $\mathbb{P}^{*}$-weak convergence of $\overline{K}_{\zeta,L}$ to an appropriate limit, $\mathbb{P}$-almost surely. 
\item[2d.] Step (2c) will be used along with Taylor's formula on the bootstrap functional in Step (2a) to derive the $\mathbb{P}^{*}$-weak limit of  the bootstrapped test statistic $nT^{*}_q$. 
\item[2e.] We will conclude step (2) by showing that the weak limit obtained in Step (2d) is the same as the weak limit of $nT_q$, thus establishing (a) in the theorem's statement. 
\end{itemize}
\item We will then prove statement (b) in the present theorem under the assumption that $H_{1,q}$ is true.
\end{enumerate}
\medskip

\noindent\textit{Step 1: Empirical estimators are either strongly convergent or strongly tight.} Consider first the case when $H_{0,q}$ is true. We will show that, in this case, we have strong consistency of the required quantities.  Suppose that we can show that under $H_{0,q}$, $M$ converges almost surely to $q$ as $n\rightarrow\infty$. Since $M$ is integer-valued, this wold imply that $M = q$ eventually, $\mathbb{P}$-almost surely. Consequently, $\widehat{\Theta}_{M} = \widehat{\Theta}_q$ eventually, $\mathbb{P}$-almost surely. So, following the same steps as in the proof of Theorem \ref{thm1}, we would obtain that $\widehat{\Theta}_{M}$ converges $\mathbb{P}$-almost surely to $K_{X,L}$, which in turn would imply that $\widehat{D}$ converges $\mathbb{P}$-almost surely to $D$.

To prove that $M$ converges almost surely to $q$, we claim that
$$m_n=\min\left\{m\geq 1: T_m\leq \epsilon\frac{\log n}{n}  \|P_L\circ  \widehat{K}_{W,L}\|^2_F\right\}\stackrel{a.s.}{\rightarrow} q.$$
To see this, we first note that $T_q$ converges to zero $\mathbb{P}$-almost surely as $n\rightarrow\infty$ by the continuous mapping theorem (applied to the strong convergence of $\widehat\Theta_q$ to $K_{X,L}$). Next, we observe that $\frac{n}{2\log\log n} T_q$ is eventually bounded $\mathbb{P}$-almost surely, by means of the law of the iterated logarithm\footnote{Recall that under $H_{0,q}$, it $\mathbb{P}$-almost surely holds that $\widehat\Theta=\Xi(P_L\circ\widehat\Theta)$ for all $n$ sufficiently large. This, combined with \eqref{empirical-objective-bound}, yields (for all $n$ sufficiently large)
$$
{\alpha(n)\widehat\pi^2(\widehat\Theta)=\alpha(n)\| P_L\circ \widehat\Theta- P_L\circ\widehat K_{W,L}\|^2_F\leq  \|\sqrt{\alpha(n)}(P_L\circ \widehat K_{X,L} -P_L \circ K_{X,L})\|^2_F= \|\sqrt{\alpha(n)}Y_n\|^2_F.}
$$
Now each entry of $Y_n$ is an average of $n$ iid random variavbles of mean zero and finite variance. So if we pick $\alpha(n)=n/(2\log\log n)$, the LIL will imply that $\lim\sup_{n\rightarrow\infty}\alpha(n)\widehat\pi^2(\widehat\Theta)$ is in a bounded set, almost surely. picking $\alpha(n)=n/\log n$ will then yield  convergence of $\alpha(n)T_q$ to zero, $\mathbb{P}$-almost surely}. So 
$$nT_q/(\epsilon \log n) = \frac{1}{\epsilon}\frac{n}{2\log\log n} T_q\times \frac{2\log\log n}{\log n}\stackrel{\mathbb{P}-a.s.}{\longrightarrow}0,$$
for any $\epsilon>0$. The term $ \|P_L\circ  \widehat{K}_{W,L}\|^2_F$ converges $\mathbb{P}$-almost surely to the positive constant $ \|P_L\circ  {K}_{W,L}\|^2_F$. Therefore $m_n$, and consequently $M$, converges $\mathbb{P}$-almost surely to $q$ under $H_{0,q}$.

Now consider the setting where both $H_{1,q}$ and the global null $H_0$ are true. In this case, one of $H_{0,q+1},...,H_{0,d}$ is true. Say it's $H_{0,r}$ for $r\in\{q+1,...,d\}$. By arguments similar to above, we can show that $M$ converges $\mathbb{P}$-almost surely to $r$. Thus, eventually, $\widehat\Theta_M$ equals $\widehat\Theta_r$. Since $H_{0,r}$ is true, $\widehat\Theta_r$ converges $\mathbb{P}$-almost surely to $K_{X,L}$, yielding strong consistency of $\widehat D$ for $D$. As for $\widehat{\Theta}_q$, this may not converge almost surely, since $q<r$, but we will show that it is strongly tight. Define $K^{(q)}_{X,L}(i,j)=\sum_{l=1}^{q}\lambda_l \varphi_l(t_i)\varphi_l(t_j)$ so that $\mathrm{rank}(K^{(q)}) \leq q$. Since $\widehat{\Theta}_q$ is a minimizer of the function $\Theta \mapsto ||P_L \circ (\widehat{K}_{W,L} - \Theta)||_F$ over matrices of rank at most $q$, 
\begin{eqnarray*}
\| P_L\circ \widehat\Theta_q - P_L\circ  {K}_{W,L}\|_F&\leq&\| P_L\circ \widehat\Theta_q - P_L\circ \widehat{K}_{W,L}\|_F + \| P_L\circ \widehat K_{W,L} - P_L\circ \widehat{K}_{W,L}\|_F\\
&\leq&\| P_L\circ {K}^{(q)}_{X,L} - P_L\circ \widehat{K}_{W,L}\|_F+\| P_L\circ \widehat K_{W,L} - P_L\circ {K}_{W,L}\|_F\\
&\leq&\| P_L\circ {K}^{(q)}_{X,L} - P_L\circ {K}_{W,L}\|_F+2\| P_L\circ {K}_{W,L} - P_L\circ \widehat{K}_{W,L}\|_F\\
&=&\| P_L\circ {K}^{(q)}_{X,L} - P_L\circ {K}_{X,L}\|_F+2\| P_L\circ {K}_{W,L} - P_L\circ \widehat{K}_{W,L}\|_F\\
&=&const + 2\| P_L\circ {K}_{W,L} - P_L\circ \widehat{K}_{W,L}\|_F
\end{eqnarray*}

The right hand side converges almost surely to a constant, therefore $P_L\circ \widehat\Theta_q$ eventually lies in a closed ball of finite radius centred at $P_L\circ {K}_{W,L}$. To see that the diagonal elements of $\widehat\Theta_q$ must also eventually lie in some compact set almost surely, consider an arbitrary diagonal such element $\widehat\theta_{ii}$ of $\widehat\Theta_q$. Since $L>L_\dagger$, where $L_\dagger$ is as in Theorem \ref{theorem-identifiability}, there exists a $(q+1)\times (q+1)$ submatrix $S$ of $\widehat\Theta_q$ that contains $\widehat\theta_{ii}$ but contains no other diagonal elements of $\widehat\Theta_q$ (see, e.g., the first part of the proof of Theorem\ref{theorem-identifiability}). The matrix $S$ is clearly of rank at most $q$, thus the column $S_i$ that contains $\widehat\theta_{ii}$ is in the span of the remaining columns of $S$, $\{S_j\}_{j\neq i}$. In other words, there exist $q$ coefficients $\{\alpha_j\}_{j\neq i}$ such that
$$S_i=\sum_{j\neq i} \alpha_j S_j.$$
Now all entries of $S$ other than $\widehat\theta_{ii}$ are elements of $P_L\circ \widehat\Theta_q$, and the latter eventually lies in a compact set almost surely. It follows that 
\begin{itemize}
\item[(i)] Any coefficients $\{\alpha_j\}_{j\neq i}$ such that $S_i=\sum_{j\neq i} \alpha_j S_j$ also eventually lie in a compact set almost surely. This because $S_{k,i}=\sum_{j\neq i} \alpha_j S_{k,j}$ for all $k\neq i$ and the $\{S_{k,j}\}_{k\neq i}$ are all eventually bounded almost surely.  

\item[(ii)] Thus, $\widehat\theta_{ii}$ lies in a compact set eventually almost surely, since $\widehat\theta_{ii}=\sum_{j\neq i} \alpha_j S_{i,j}$.
\end{itemize}

\noindent Since $\widehat\theta_{ii}$ was arbitrarily chosen, we establish that $\lim\sup_{n\rightarrow\infty}\|\widehat\Theta_q\|_F<\infty$ almost surely.  

\smallskip
To conclude Step 1 of the proof, consider now the case where the global alternative $H_1$ is true (i.e. the true rank exceeds $d$). In this case, we will show that $M$ converges almost surely to $d$. We will also show that $\widehat\Theta_d$ is eventually almost surely in a compact set. To prove that $M\stackrel{a.s.}{\rightarrow} d$, let us revisit the quantity
$$m_n=\min\left\{m\geq 1: T_m\leq \frac{\log n}{n}  \|P_L\circ  \widehat{K}_{W,L}\|^2_F\right\}$$
Recalling the reverse triangle inequality (applied to the Frobenius norm)
$$ \|u-v\|_F \geq \left| \|u\|_F-\|v\|_F \right|,$$
we may write
\begin{eqnarray*}
\|P_L\circ \Theta - P_L\circ \widehat K_{W,L} \|_F&=&\|P_L\circ \Theta -P_L\circ K_{W,L}+P_L\circ K_{W,L} - P_L\circ \widehat K_{W,L} \|_F\\
&=&\|(P_L\circ \Theta -P_L\circ K_{W,L})- (P_L\circ \widehat K_{W,L}-P_L\circ K_{W,L}) \|_F\\
&\geq&\left|\|(P_L\circ \Theta -P_L\circ K_{W,L})\|_F-\| (P_L\circ \widehat K_{W,L}-P_L\circ K_{W,L}) \|_F \right|
\end{eqnarray*}
Under $H_1$, for any $m\leq d$, we thus have
$$\inf_{\Theta\in\mathcal{M}_m}\| P_L\circ \Theta - P_L\circ  \widehat{K}_{W,L}\|_F\geq \inf_{\Theta\in\mathcal{M}_m}\left|\|(P_L\circ \Theta -P_L\circ K_{W,L})\|_F-\| (P_L\circ \widehat K_{W,L}-P_L\circ K_{W,L}) \|_F \right|$$
and we notice that the right hand side converges to $\inf_{\Theta\in\mathcal{M}_m}\| P_L\circ \Theta - P_L\circ  {K}_{W,L}\|_F$. This is a strictly positive quantity from part 2 of Theorem \ref{theorem-identifiability}. It follows that $m_n$ diverges almost surely, and thus $M$ converges to $d$ almost surely. The proof of the strong tightness of $\widehat\Theta_d$ is established in exactly in the same way as the proof of the strong tightness of $\widehat\Theta_q$ in the case where $H_{1,q}\cap H_0$ is valid, considered in the previous paragraph.

\medskip

\noindent\textit{Step 2: Asymptotic theory for bootstrap under $H_{0,q}$ (statement (a) of the Theorem).} As announced earlier in the proof, this step is broken down into five substeps, and we now establish these in order.

\smallskip

\noindent\textit{Step 2a: Consistency of minimizer of bootstrap functional.} We can, for simplicity, assume that $E(W) = 0$, and drop the term $\overline{{\bf W}}$ in the definition of $\widehat{m}({\bf W}_i)$. We will fix a set $\Omega_{0}$ of $\mathbb{P}$-measure one on which a.s. convergence and law of iterated logarithm results hold as will be required in the proof. Fix any $\omega \in \Omega_{0}$ and work with the resulting population $\{{\bf W}_{1}(\omega),\ldots,{\bf W}_{n}(\omega),\ldots\}$. All statements will be conditional on this population. We will drop the dependence on $\omega$ for simplicity of notation.

\indent Assume that $H_{0,q}$ is true, and recall that $\zeta_j = U^{*}_j + V^{*}_j$. Define 
\begin{eqnarray*}
\overline{K}_{\zeta,L} &=& n^{-1}\sum_{j=1}^{n} \zeta_j\zeta_j\transpose \\
&=& \underset{t_1}{\underbrace{n^{-1}\sum_{i=1}^{n} U^{*}_jU^{*\transpose}_j}} + \underset{t_2}{\underbrace{n^{-1}\sum_{i=1}^{n} V^{*}_jV^{*\transpose}_j}} + \underset{t_3}{\underbrace{n^{-1}\sum_{i=1}^{n} (U^{*}_j V^{*\transpose}_j + V^{*}_j U^{*\transpose}_j)}}.
\end{eqnarray*}
Now, observe that 
\begin{eqnarray*}
\E^{*}(U^{*}_1U^{*\transpose}_1) &=& n^{-1}\sum_{i=1}^{n} \widehat{m}({\bf W}_i)\widehat{m}({\bf W}_i)\transpose \\
&=& \widehat{\Theta}_q\widehat{K}_{W,L}^{-1}\widehat{\Theta}_q\transpose \\
\E^{*}(V^{*}_1V^{*\transpose}_1) &=& \widehat{D} + \widehat{A}, \quad \mbox{and} \\
\E^{*}(U^{*}_1V^{*\transpose}_1) &=& \E^{*}(V^{*}_1U^{*\transpose}_1) \ = \ 0,
\end{eqnarray*}
where the last equality follows from the independence of $U^{*}_1$ and $V^{*}_1$ along with the fact that $\E^{*}(V^{*}_1) = 0$. It is standard that all of $t_1 - \E^{*}(U^{*}_1U^{*\transpose}_1)$, $t_2 - \E^{*}(V^{*}_1V^{*\transpose}_1)$ and $t_3$ converge to zero $\mathbb{P}^{*}$-almost surely $\mathbb{P}$-almost surely (though, for completeness, we show this in Lemma \ref{aux-lemma1}, after this proof). Thus, combining these statements, it follows that $\overline{K}_{\zeta,L} - (\widehat{\Theta}_q + \widehat{D})$ converges to zero $\mathbb{P}^{*}$-almost surely.

\indent Now define the function $\overline{\Pi}(\Theta) = \|P_{L} \circ (\overline{K}_{\zeta,L} - \Theta)\|_{F}^{2}$. Define $\Theta^{*}_q$ to be a minimizer of the functional $\overline{\Pi}(\cdot)$ over $\mathcal{M}_q$. We will now prove that $\Theta^{*}_q$ converges to $K_{X,L}$ $\mathbb{P}^{*}$-almost surely, $\mathbb{P}$-almost surely under $H_{0,q}$. To show this, let $\overline{K}_{X,L}$ be the (unobservable) random $L \times L$ matrix satisfying
$$ \overline{K}_{X,L}(i,j) = n^{-1} \sum_{k=1}^{n} X_k^{*}(t_i)X_k^{*}(t_j),$$
where $\{{\bf X}_k^{*}: k=1,2,\ldots,n\}$ is the (unobservable) bootstrap sample from the (also unobservable) $\{{\bf X}_k: k=1,2,\ldots,n\}$ (i.e. with those indices sampled by the bootstrap). Under $H_{0,q}$, $\mathrm{rank}(\overline{K}_{X,L})\leq q$. Also, define $\widehat{K}_{X,L} = n^{-1}\sum_{i=1}^{n} {\bf X}_i{\bf X}_i\transpose$. Thus, it must hold that
\begin{eqnarray}\label{empirical-objective-bound-boot}
\underset{\overline\Pi(\Theta^{*}_q)}{\underbrace{\| P_L\circ \Theta^{*}_q - P_L\circ\overline K_{\zeta,L}\|_F}}&\leq& \underset{\overline\Pi(\overline K_{X,L})}{\underbrace{\|P_L\circ \overline K_{X,L} -P_L \circ \overline K_{\zeta,L}\|_F}} \nonumber \\
&\leq&  \|P_L\circ \overline K_{X,L} -P_L \circ  \widehat{K}_{X,L}\|_F + \|P_L\circ \widehat{K}_{X,L} - P_L \circ (\widehat{\Theta}_q + \widehat{D})\|_F \\
&& + \ \|P_L\circ (\widehat{\Theta}_q + \widehat{D}) -P_L \circ \overline{K}_{\zeta,L}\|_F. \nonumber 
\end{eqnarray}
Under $H_{0,q}$, it has been proved above that $\widehat{\Theta}_{q}$ converges $\mathbb{P}$-almost surely to $K_{X,L}$. Further, $\widehat{D}$ converges $\mathbb{P}$-almost surely to $D$. So, $\widehat{\Theta}_q + \widehat{D}$ converges $\mathbb{P}$-almost surely to $K_{X,L} + D = K_{W,L}$. Also, $P_{L} \circ K_{X,L} = P_{L} \circ K_{W,L}$. Thus, the right hand side converges to zero $\mathbb{P}^{*}$-almost surely, $\mathbb{P}$-almost surely by the classical as well as the bootstrap almost sure convergence results and the continuous mapping theorem (note that since $k_X$ is continuous, the covariance operator of the process $X$ is trace-class, and so is any discretization thereof). We therefore have
\begin{eqnarray*}
\| P_L\circ \Theta^{*}_q -  P_L\circ K_{W,L}\|_F &\leq& \| P_L\circ \Theta^{*}_q - P_L\circ \overline K_{\zeta,L}\|_F + \|P_L\circ \overline K_{\zeta,L} - P_L \circ (\widehat{\Theta}_q + \ \widehat{D})\|_F \\
&&+ \ \|P_L \circ (\widehat{\Theta}_q + \widehat{D}) - P_L \circ K_{W,L}\|_F \\
&\longrightarrow& 0,
\end{eqnarray*}
$\mathbb{P}^{*}$-almost surely, $\mathbb{P}$-almost surely as $n\rightarrow\infty$, which implies that 
$$P_L\circ \Theta^{*}_q \longrightarrow P_L\circ K_{W,L}=P_L\circ K_{X,L}\quad n\rightarrow\infty$$
$\mathbb{P}^{*}$-almost surely, $\mathbb{P}$-almost surely as $n\rightarrow\infty$. Now, the same arguments as in the remainder of the proof of Step 2 (Consistency of Empirical Minimizers) in Theorem \ref{thm1} show that $\Theta^{*}_q$ converges $\mathbb{P}^{*}$-almost surely, $\mathbb{P}$-almost surely as $n\rightarrow\infty$ to $K_{X,L}$.

\medskip

\noindent\textit{Step 2b: Consistency of appropriate rank factorization.} Since $\mathrm{rank}(\Theta^{*}_q) \leq q$, we can write $\Theta^{*}_q = \widetilde{C}\widetilde{C}\transpose $, where $\widetilde{C} \in \mathbb{R}^{L \times q}$. Thus, $\overline{\Pi}(\Theta^{*}_q) = \min_{\Theta} \overline{\Pi}(\Theta) = \min_{C} \overline{\Psi}(C) = \overline{\Psi}(\widetilde{C})$, where $\overline{\Psi}(C) = \|P_{L} \circ (\overline{K}_{\zeta,L} - CC\transpose )\|_{F}^{2}$. We now make the observation that $\widetilde{C}U$ will also yield the same minimum value for any $q \times q$ orthogonal matrix $U$. So, we will work with the following modified estimator instead. Define
$$\overline{U} = \arg\min_{\substack{O \in \mathbb{R}^{q \times q}\\OO\transpose = I_q}} \|\overline{C}O - C_{0}\|_{F},$$
and subsequently, define $\overline{C} = \widetilde{C}\overline{U}$. Thus, $\overline{C}$ is the version of $\widetilde{C}$ that is ``aligned'' with $C_{0}$ in the above Procrustes distance minimization sense. It is well known that the solution of the above minimization problem is given by $\overline{U} = \widetilde{U}\widetilde{V}\transpose $, where $\widetilde{U}\widetilde{D}\widetilde{V}$ is the singular value decomposition of the matrix $C_{0}\transpose \widetilde{C}$. Further, 
$$\min_{C \in \mathbb{R}^{L \times q}} \overline{\Psi}(C) = \overline{\Psi}(\widetilde{C}) = \overline{\Psi}(\overline{C}) \ \ \ \mathbb{P}\mbox{-almost surely}.$$
So, the bootstrap asymptotic distributions of $\min_{C \in \mathbb{R}^{L \times q}} \overline{\Psi}(C)$ and $\overline{\Psi}(\overline{C})$ will agree $\mathbb{P}$-almost surely. Since it is these asymptotic distributions that we wish to determine, we can work with $\overline{C}$, even though it is an oracle quantity (similar to Step 3 of Theorem \ref{thm1}).

\indent We will now show that $\overline{C} = \overline{C}_n$ (to be explicit about the dependence on $n$) converges to $C_{0}$ $\mathbb{P}^{*}$-almost surely as $n \rightarrow \infty$, $\mathbb{P}$-almost surely under $H_{0,q}$. Note that since $\Theta^{*}_q = \Theta^{*}_{q,n} = \overline{C}_{n}\overline{C}_{n}\transpose $ converges $\mathbb{P}^{*}$-almost surely to $K_{X,L} = C_{0}C_{0}\transpose$ $\mathbb{P}$-almost surely (proven in Step (2b)), we have that $\|\overline{C}_{n}\|_{F} = \sqrt{\mathrm{tr}(\overline{C}_{n}\overline{C}_{n}\transpose )}$ converges $\mathbb{P}^{*}$-almost surely to $\|C_{0}\|_{F} = \sqrt{\mathrm{tr}(K_{X,L})}$ $\mathbb{P}$-almost surely. 
Thus, the set $\Omega^{*} = \{\omega^{*} : \|\overline{C}_{n}(\omega^{*})\|_{F} \leq 2\|C_{0}\|_{F} \ \mbox{and} \ \overline{C}_{n}(\omega^{*})\overline{C}_{n}(\omega^{*})\transpose  \rightarrow K_{X,L} \ \mbox{as} \ n \rightarrow \infty\}$ has $\mathbb{P}^{*}$-probability measure one.

Fix any $\omega^{*} \in \Omega^{*}$. Since $\overline{C}_{n}(\omega^{*})$ lies in  the closed ball of radius $2||C_0||_{F}$ for all large $n$, in order to show that $\overline{C}_n(\omega^{*})$ converges to $C_0$, we will show that all subsequences of $\overline{C}_n(\omega^{*})$ converge to $C_0$. Suppose that 
there exists a subsequence $\{k'\}$ of $\{n\}$ such that $\overline{C}_{k'}(\omega^{*}) \rightarrow C_{1}(\omega^{*})$ as $k' \rightarrow \infty$, where $C_{1}(\omega^{*}) \neq C_0$. But then $\overline{C}_{k'}(\omega^{*})\overline{C}_{k'}(\omega^{*})\transpose  \rightarrow C_{1}(\omega^{*})C_{1}(\omega^{*})\transpose  = K_{X,L}$. Thus, $C_{1}(\omega^{*}) = C_{0}V(\omega^{*})$ for some $q \times q$ orthogonal matrix $V(\omega^{*})$. Suppose that $C_{1}(\omega^{*}) \neq C_{0}$, equivalently, $V(\omega^{*}) \neq I_{q}$. Define $\overline{C}_{k'}^{(0)}(\omega^{*}) = \overline{C}_{k'}(\omega^{*})V(\omega^{*})\transpose $ for each $k' \geq 1$. Then,
$$\|\overline{C}_{k'}^{(0)}(\omega^{*}) - C_{0}\|_{F} = \|\overline{C}_{k'}(\omega^{*})V(\omega^{*})\transpose  - C_{0}\|_{F} \rightarrow \|C_{1}(\omega^{*})V(\omega^{*})\transpose  - C_{0}\|_{F} = 0.$$
On the other hand, $\|\overline{C}_{k'}(\omega^{*}) - C_{0}\|_{F} \rightarrow \|C_{1}(\omega^{*}) - C_{0}\|_{F} > 0.$ Recall that $\overline{C}_{k'}(\omega^{*}) = \widetilde{C}_{k'}(\omega^{*})\overline{U}(\omega^{*})$ as per our construction. So, there exists $k'_{0} \geq 1$ such that 
\begin{eqnarray*}
\|\widetilde{C}_{k'_{0}}(\omega^{*})\overline{U}(\omega^{*})V(\omega^{*})\transpose  - C_{0}\|_{F} &=& \|\overline{C}_{k'_{0}}^{(0)}(\omega^{*}) - C_{0}\|_{F} \\
&<& (1/4)\|C_{1}(\omega^{*}) - C_{0}\|_{F} \\
&<& \|\overline{C}_{k'_{0}}(\omega^{*}) - C_{0}\|_{F} \ = \ \min_{\substack{O \in \mathbb{R}^{q \times q}\\OO\transpose = I_q}} \|\widetilde{C}_{k'_{0}}(\omega^{*})O - C_{0}\|_{F}.
\end{eqnarray*}
This leads to a contradiction unless $V(\omega^{*}) = I_{q}$. Hence, $C_{1}(\omega^{*}) = C_{0}$ so that the limit does not depend on $\omega^{*} \in \Omega^{*}$. A standard subsequence argument (using the fact that $\overline{C}_{n}$ lies in a compact set for all sufficiently large $n$ almost surely) now shows that the entire sequence $\overline{C}_{n}$ must converge to $C_{0}$ as $n \rightarrow \infty$ on $\Omega^{*}$. 

\medskip

We will now derive the asymptotic distribution of $\min_{C} \overline{\Psi}(C)$ under $H_{0}$. Define $\widecheck{\Psi}(C) = ||P_{L} \circ (\widehat{\Theta}_{q} + \widehat{D}) - CC\transpose)||_F^2 = ||P_{L} \circ \widehat{\Theta}_{q} - CC\transpose||_F^2$.

\smallskip

\noindent\textit{Step 2c: Asymptotic distribution of bootstrap covariance matrix}. We know determine the asymptotic distribution of the bootstrap covariance matrix $\overline{K}_{\zeta,L}$  First, use the form of the Hessian established in Step 1 of Theorem \ref{thm1} to observe that $\nabla^{2}\overline{\Psi}(C) - \nabla^{2}\widecheck{\Psi}(C) = -4I_{q} \otimes (P_{L} \circ (\overline{K}_{\zeta,L} - (\widehat{\Theta}_{q} + \widehat{D}))$ is free of $C$. We will now derive the asymptotic distribution of $\sqrt{n}(\overline{K}_{\zeta,L} - (\widehat{\Theta}_{q} + \widehat{D}))$. This will then imply that $\nabla^{2}\overline{\Psi}(C) - \nabla^{2}\widecheck{\Psi}(C) = O_{\mathbb{P}^{*}}(n^{-1/2})$  as $n \rightarrow \infty$ $\mathbb{P}$-almost surely. For simplicity of notation, denote $\widehat{\Theta}_{q} + \widehat{D}$ by $\widehat{K}$. In order to derive the $\mathbb{P}^{*}$-weak convergence of $\sqrt{n}(\overline{K}_{\zeta,L} - \widehat{K})$, it is enough to derive the $\mathbb{P}^{*}$-weak convergence of $\sqrt{n}({\bf a}\transpose\overline{K}_{\zeta,L}{\bf b} - {\bf a}\transpose\widehat{K}{\bf b})$ for each fixed ${\bf a}, {\bf b} \in \mathbb{R}^L$. Observe that from the definition of $\overline{K}_{\zeta,L}$, it follows that ${\bf a}\transpose\overline{K}_{\zeta,L}{\bf b} = n^{-1}\sum_{j=1}^{n} \mathbf{1}\transpose{\bf R}_j$, where ${\bf R}_j$ is a $4$-tuple given by
$${\bf R}_j = 
\begin{bmatrix}
({\bf a}\transpose U^{*}_j)({\bf b}\transpose U^{*}_j)\\
({\bf a}\transpose V^{*}_j)({\bf b}\transpose V^{*}_j)\\
({\bf a}\transpose U^{*}_j)({\bf b}\transpose V^{*}_j)\\
({\bf a}\transpose V^{*}_j)({\bf b}\transpose U^{*}_j)
\end{bmatrix}. $$
Clearly, 
$$ \E^{*}({\bf R}_j) = 
\begin{bmatrix}
{\bf a}\transpose\widehat{\Theta}_q\widehat{K}_{W,L}^{-1}\widehat{\Theta}_q{\bf b}\\
{\bf a}\transpose(\widehat{D}+\widehat{A}){\bf b}\\
0\\
0
\end{bmatrix}.$$
We now proceed to find $\mathrm{Cov}^{*}({\bf R}_j)$. Clearly,
$$
\E^{*}[({\bf a}\transpose U^{*}_j)^2({\bf b}\transpose U^{*}_j)^2] \\
= \E^{*}[\mbox{tr}\{{\bf a}{\bf a}\transpose U^{*}_jU^{*\transpose}_j{\bf b}{\bf b}\transpose U^{*}_jU^{*\transpose}_j\}].
$$
Using the fact that for compatible matrices $Q_1,Q_2,Q_3$ and $Q_4$, we have $\mbox{tr}\{Q_1\transpose Q_2Q_3Q_4\transpose\} = (\mbox{vec}(Q_1))\transpose(Q_4 \otimes Q_2)\mbox{vec}(Q_3)$, the above term equals
\begin{eqnarray*}
&&\E^{*}[(\mbox{vec}({\bf a}{\bf a}\transpose))\transpose(U^{*}_jU^{*\transpose}_j \otimes U^{*}_jU^{*\transpose}_j)\mbox{vec}({\bf b}{\bf b}\transpose)] \\
&=& (\mbox{vec}({\bf a}{\bf a}\transpose))\transpose\left[n^{-1}\sum_{i=1}^{n} \widehat{m}({\bf W}_i)\widehat{m}({\bf W}_i)\transpose \otimes \widehat{m}({\bf W}_i)\widehat{m}({\bf W}_i)\transpose\right]\mbox{vec}({\bf b}{\bf b}\transpose) \\
&=& (\mbox{vec}({\bf a}{\bf a}\transpose))\transpose\left[n^{-1}\sum_{i=1}^{n} \underset{A}{\underbrace{\widehat{\Theta}_q\widehat{K}_{W,L}^{-1}{\bf W}_i}}\underset{A\transpose}{\underbrace{{\bf W}_i\transpose\widehat{K}_{W,L}^{-1}\widehat{\Theta}_q}} \otimes \underset{B}{\underbrace{\widehat{\Theta}_q\widehat{K}_{W,L}^{-1}{\bf W}_i}}\underset{B\transpose}{\underbrace{{\bf W}_i\transpose\widehat{K}_{W,L}^{-1}\widehat{\Theta}_q}}\right]\mbox{vec}({\bf b}{\bf b}\transpose) 
\end{eqnarray*}
Using the properties of Kronecker products, it follows that $A A\transpose \otimes B B\transpose = (A \otimes B)(A\transpose \otimes B\transpose) = (A \otimes B)(A \otimes B)\transpose$. So, the above term equals
\begin{eqnarray*}
&& (\mbox{vec}({\bf a}{\bf a}\transpose))\transpose\left[n^{-1}\sum_{i=1}^{n} \left(\widehat{\Theta}_q\widehat{K}_{W,L}^{-1}{\bf W}_i \otimes \widehat{\Theta}_q\widehat{K}_{W,L}^{-1}{\bf W}_i \right)\left(\widehat{\Theta}_q\widehat{K}_{W,L}^{-1}{\bf W}_i \otimes \widehat{\Theta}_q\widehat{K}_{W,L}^{-1}{\bf W}_i\right)\transpose\right]\mbox{vec}({\bf b}{\bf b}\transpose) \\
&=& (\mbox{vec}({\bf a}{\bf a}\transpose))\transpose\left[n^{-1}\sum_{i=1}^{n} \left(\widehat{\Theta}_q\widehat{K}_{W,L}^{-1} \otimes \widehat{\Theta}_q\widehat{K}_{W,L}^{-1}\right)\left({\bf W}_i \otimes {\bf W}_i\right)\left({\bf W}_i \otimes {\bf W}_i\right)\transpose\left(\widehat{\Theta}_q\widehat{K}_{W,L}^{-1} \otimes \widehat{\Theta}_q\widehat{K}_{W,L}^{-1}\right)\transpose\right]\mbox{vec}({\bf b}{\bf b}\transpose) \\
&=& (\mbox{vec}({\bf a}{\bf a}\transpose))\transpose\underset{Q}{\underbrace{ \left(\widehat{\Theta}_q\widehat{K}_{W,L}^{-1} \otimes \widehat{\Theta}_q\widehat{K}_{W,L}^{-1}\right)\left[n^{-1}\sum_{i=1}^{n}\left({\bf W}_i \otimes {\bf W}_i\right)\left({\bf W}_i \otimes {\bf W}_i\right)\transpose\right]\left(\widehat{\Theta}_q\widehat{K}_{W,L}^{-1} \otimes \widehat{\Theta}_q\widehat{K}_{W,L}^{-1}\right)\transpose}}\mbox{vec}({\bf b}{\bf b}\transpose)
\end{eqnarray*} 
The matrix $Q$ above converges $\mathbb{P}$-almost surely under $H_{0,q}$ to 
\begin{eqnarray*}
(K_{X,L}K_{W,L}^{-1} \otimes K_{X,L}K_{W,L}^{-1})[\E\{({\bf W}_1 \otimes {\bf W}_1)({\bf W}_1 \otimes {\bf W}_1)\transpose\}](K_{X,L}K_{W,L}^{-1} \otimes K_{X,L}K_{W,L}^{-1})\transpose
\end{eqnarray*}
Assuming Gaussianity of the observations $\{{\bf W}_i\}$, we have that ${\bf W}_1{\bf W}_1\transpose$ follows a central $L$-dimensional Wishart distribution with parameter $K_{W,L}$. Also, observe that for any vector ${\bf x} \in \mathbb{R}^L$, we have ${\bf x} \otimes {\bf x} = \mbox{vec}({\bf x}{\bf x}\transpose)$. It now follows from equations (3.131), (3.132) and (3.135) in \citet[p. 64]{izenman_mult} that 
\begin{eqnarray*}
\E\{({\bf W}_1 \otimes {\bf W}_1)({\bf W}_1 \otimes {\bf W}_1)\transpose\} &=& \E\{\mbox{vec}({\bf W}_1{\bf W}_1\transpose)(\mbox{vec}({\bf W}_1{\bf W}_1\transpose))\transpose\} \\
&=& \mbox{Cov}(\mbox{vec}({\bf W}_1{\bf W}_1\transpose)) + \E\{\mbox{vec}({\bf W}_1{\bf W}_1\transpose)\}[\E\{\mbox{vec}({\bf W}_1{\bf W}_1\transpose)\}]\transpose \\
&=& ({I}_{L^2} + M_{L^2})(K_{W,L} \otimes K_{W,L}) + \eta_{L}\eta_{L}\transpose,
\end{eqnarray*}
where $\eta_{L} = \mbox{vec}(K_{W,L})$ is a $L^{2}$-dimensional vector, and $M_{L^{2}}$ is the $L^{2} \times L^{2}$ commutation matrix that satisfies $M_{L^2}\mbox{vec}(F) = \mbox{vec}(F\transpose)$ for any $L \times L$ matrix $F$. So, $Q$ converges $\mathbb{P}$-almost surely to 
\begin{eqnarray*}
&& K_{X,L}K_{W,L}^{-1}K_{X,L} \otimes K_{X,L}K_{W,L}^{-1}K_{X,L} + (K_{X,L}K_{W,L}^{-1} \otimes K_{X,L}K_{W,L}^{-1})M_{L^2}(K_{X,L} \otimes K_{X,L}) + \\
&& (K_{X,L}K_{W,L}^{-1} \otimes K_{X,L}K_{W,L}^{-1})\eta_L\eta_L\transpose(K_{X,L}K_{W,L}^{-1} \otimes K_{X,L}K_{W,L}^{-1})\transpose.
\end{eqnarray*}
Observe that $\E^{*}({\bf a}\transpose U^{*}_j{\bf b}\transpose U^{*}_j) = {\bf a}\transpose\widehat{\Theta}_q\widehat{K}_{W,L}^{-1}\widehat{\Theta}_q{\bf b}$ converges $\mathbb{P}$-almost surely to ${\bf a}\transpose K_{X,L}K_{W,L}^{-1}K_{X,L}{\bf b}$. Also, note that for vectors ${\bf x}, {\bf y} \in \mathbb{R}^{L}$ and for $L \times L$ matrices $E$ and $F$, we have 
$$({\bf a}\transpose E{\bf b})({\bf a}\transpose F{\bf b}) = \mbox{tr}\{{\bf a}{\bf a}\transpose E{\bf b}{\bf b}\transpose F\transpose\} = (\mbox{vec}({\bf a}{\bf a}\transpose))\transpose(E \otimes F)\mbox{vec}({\bf b}{\bf b}\transpose).$$
Thus, it follows that 
\begin{eqnarray*}
&&\mathrm{Var}^{*}[({\bf a}\transpose U^{*}_j)({\bf b}\transpose U^{*}_j)] \\
&=& \E^{*}[({\bf a}\transpose U^{*}_j)^2({\bf b}\transpose U^{*}_j)^2] - (\E^{*}[({\bf a}\transpose U^{*}_j)({\bf b}\transpose U^{*}_j)])^2 \\
&\stackrel{\mathbb{P}-a.s.}{\rightarrow}& (\mbox{vec}({\bf a}{\bf a}\transpose))\transpose(K_{X,L}K_{W,L}^{-1} \otimes K_{X,L}K_{W,L}^{-1})M_{L^2}(K_{X,L} \otimes K_{X,L})\mbox{vec}({\bf b}{\bf b}\transpose) + \\
&& (\mbox{vec}({\bf a}{\bf a}\transpose))\transpose(K_{X,L}K_{W,L}^{-1} \otimes K_{X,L}K_{W,L}^{-1})\eta_L\eta_L\transpose(K_{X,L}K_{W,L}^{-1} \otimes K_{X,L}K_{W,L}^{-1})\transpose\mbox{vec}({\bf b}{\bf b}\transpose).
\end{eqnarray*}
Next, let us consider $\mathrm{Var}^{*}[({\bf a}\transpose V^{*}_j)({\bf b}\transpose V^{*}_j)]$. This can be simplified as follows by using the fact that the $V^{*}_j$'s are themselves centered Gaussians.
\begin{eqnarray*}
&& \mathrm{Var}^{*}[({\bf a}\transpose V^{*}_j)({\bf b}\transpose V^{*}_j)]\\
&=& \E^{*}[({\bf a}\transpose V^{*}_j)^2({\bf b}\transpose V^{*}_j)^2] - (\E^{*}[({\bf a}\transpose V^{*}_j)({\bf b}\transpose V^{*}_j)])^2 \\
&=& (\mbox{vec}({\bf a}{\bf a}\transpose))\transpose[(\mathbb{I}_{L^2} + M_{L^2})(\widehat{D}+\widehat{A})\otimes(\widehat{D}+\widehat{A}) + \widehat{\gamma}_{L}\widehat{\gamma}_{L}\transpose]\mbox{vec}({\bf b}{\bf b}\transpose) - [{\bf a}\transpose(\widehat{D}+\widehat{A}){\bf b}]^{2} \\
&=& (\mbox{vec}({\bf a}{\bf a}\transpose))\transpose[M_{L^2}(\widehat{D}+\widehat{A})\otimes(\widehat{D}+\widehat{A}) + \widehat{\gamma}_{L}\widehat{\gamma}_{L}\transpose]\mbox{vec}({\bf b}{\bf b}\transpose) \\
&\stackrel{\mathbb{P}-a.s.}{\rightarrow}& (\mbox{vec}({\bf a}{\bf a}\transpose))\transpose[M_{L^2}(D+A)\otimes(D+A) + \gamma_{L}\gamma_{L}\transpose]\mbox{vec}({\bf b}{\bf b}\transpose),
\end{eqnarray*}
where $\widehat{\gamma}_{L} = \mbox{vec}(\widehat{D}+\widehat{A})$ and $\gamma_{L} = \mbox{vec}(D+A)$. \\
Next, we use the independence of the $U^{*}_j$'s and the $V^{*}_j$'s to write
\begin{eqnarray*}
&&\mathrm{Var}^{*}[({\bf a}\transpose U^{*}_j)({\bf b}\transpose V^{*}_j)]\\
&=& \E^{*}[({\bf a}\transpose U^{*}_j)^2]\E^{*}[({\bf b}\transpose V^{*}_j)^2] \\
&=& [{\bf a}\transpose\widehat{\Theta}_q\widehat{K}_{W,L}^{-1}\widehat{\Theta}_q{\bf a}][{\bf b}\transpose(\widehat{D}+\widehat{A}){\bf b}] \\
&\stackrel{\mathbb{P}-a.s.}{\rightarrow}& [{\bf a}\transpose K_{X,L}K_{W,L}^{-1}K_{X,L}{\bf a}][{\bf b}\transpose(D+A){\bf b}]. 
\end{eqnarray*}
Similarly,
\begin{eqnarray*}
\mathrm{Var}^{*}[({\bf a}\transpose V^{*}_j)({\bf b}\transpose U^{*}_j)] &\stackrel{\mathbb{P}-a.s.}{\rightarrow}& [{\bf b}\transpose K_{X,L}K_{W,L}^{-1}K_{X,L}{\bf b}][{\bf a}\transpose(D+A){\bf a}].
\end{eqnarray*}
Now note that by independence of $U^{*}_j$'s and $V^{*}_j$'s, and using the fact that the $V^{*}_j$'s are centered, we have
\begin{eqnarray*}
\mbox{Cov}^{*}[({\bf a}\transpose U^{*}_j)({\bf b}\transpose U^{*}_j),({\bf a}\transpose V^{*}_j)({\bf b}\transpose V^{*}_j)] &=& 0,\\
\mbox{Cov}^{*}[({\bf a}\transpose U^{*}_j)({\bf b}\transpose U^{*}_j),({\bf a}\transpose U^{*}_j)({\bf b}\transpose V^{*}_j)] &=& 0, \ \mbox{and}\\
\mbox{Cov}^{*}[({\bf a}\transpose U^{*}_j)({\bf b}\transpose U^{*}_j),({\bf a}\transpose V^{*}_j)({\bf b}\transpose U^{*}_j)] &=& 0.
\end{eqnarray*}
Further,
$$
\mbox{Cov}^{*}[({\bf a}\transpose V^{*}_j)({\bf b}\transpose V^{*}_j),({\bf a}\transpose U^{*}_j)({\bf b}\transpose V^{*}_j)] 
= \E^{*}[({\bf a}\transpose V^{*}_j)({\bf b}\transpose V^{*}_j)^2]\E^{*}[{\bf a}\transpose U^{*}_j] 
\stackrel{\mathbb{P}-a.s.}{\rightarrow} 0
$$
since $\E^{*}[({\bf a}\transpose V^{*}_j)({\bf b}\transpose V^{*}_j)^2]$ is bounded $\mathbb{P}$-almost surely and $\E^{*}[{\bf a}\transpose U^{*}_j] = {\bf a}\transpose\widehat{\Theta}_q\widehat{K}_{W,L}^{-1}\overline{W}$ converges to zero $\mathbb{P}$-almost surely since $\E({\bf W}_1) = 0$ by assumption. Similarly,
\begin{eqnarray*}
\mbox{Cov}^{*}[({\bf a}\transpose V^{*}_j)({\bf b}\transpose V^{*}_j),({\bf a}\transpose V^{*}_j)({\bf b}\transpose U^{*}_j)] &\stackrel{\mathbb{P}-a.s.}{\rightarrow}& 0. 
\end{eqnarray*}
Finally, 
\begin{eqnarray*}
\mbox{Cov}^{*}[({\bf a}\transpose U^{*}_j)({\bf b}\transpose V^{*}_j),({\bf a}\transpose V^{*}_j)({\bf b}\transpose U^{*}_j)] 
&=& \E^{*}[({\bf a}\transpose U^{*}_j)({\bf b}\transpose U^{*}_j)({\bf a}\transpose V^{*}_j)({\bf b}\transpose V^{*}_j)] \\
&=& [{\bf a}\transpose\widehat{\Theta}_q\widehat{K}_{W,L}^{-1}\widehat{\Theta}_q{\bf b}][{\bf a}\transpose(\widehat{D}+\widehat{A}){\bf b}] \\
&\stackrel{\mathbb{P}-a.s.}{\rightarrow}& [{\bf a}\transpose K_{X,L}K_{W,L}^{-1}K_{X,L}{\bf b}][{\bf a}\transpose(D+A){\bf b}]. 
\end{eqnarray*}
So, collecting all the expressions together, we get that
\begin{eqnarray*}
&&\mathrm{Var}^{*}({\bf 1}\transpose{\bf R}_j)\\
&=& {\bf 1}\transpose\mbox{Cov}^{*}({\bf R}_j){\bf 1} \\
&\stackrel{\mathbb{P}-a.s.}{\rightarrow}& (\mbox{vec}({\bf a}{\bf a}\transpose))\transpose(K_{X,L}K_{W,L}^{-1} \otimes K_{X,L}K_{W,L}^{-1})M_{L^2}(K_{X,L} \otimes K_{X,L})\mbox{vec}({\bf b}{\bf b}\transpose) + \\
&& (\mbox{vec}({\bf a}{\bf a}\transpose))\transpose(K_{X,L}K_{W,L}^{-1} \otimes K_{X,L}K_{W,L}^{-1})\eta_L\eta_L\transpose(K_{X,L}K_{W,L}^{-1} \otimes K_{X,L}K_{W,L}^{-1})\transpose\mbox{vec}({\bf b}{\bf b}\transpose) + \\
&& (\mbox{vec}({\bf a}{\bf a}\transpose))\transpose[M_{L^2}(D+A)\otimes(D+A) + \gamma_{L}\gamma_{L}\transpose]\mbox{vec}({\bf b}{\bf b}\transpose) + \\
&& [{\bf a}\transpose K_{X,L}K_{W,L}^{-1}K_{X,L}{\bf a}][{\bf b}\transpose(D+A){\bf b}] + [{\bf b}\transpose K_{X,L}K_{W,L}^{-1}K_{X,L}{\bf b}][{\bf a}\transpose(D+A){\bf a}] + \\
&& 2[{\bf a}\transpose K_{X,L}K_{W,L}^{-1}K_{X,L}{\bf b}][{\bf a}\transpose(D+A){\bf b}].
\end{eqnarray*}
Observe that for an $L \times L$ matrix $E$ and an $L$-dimensional vector ${\bf x}$, we have
\begin{eqnarray*}
M_{L^2}(E \otimes E)\mbox{vec}({\bf x}{\bf x}\transpose) 
&=& M_{L^2}(E \otimes E)({\bf x} \otimes {\bf x}) \\
&=& M_{L^2}(E{\bf x} \otimes E{\bf x}) \\
&=& M_{L^2}\mbox{vec}\{(E{\bf x})(E{\bf x})\transpose\} = \mbox{vec}\{(E{\bf x})(E{\bf x})\transpose\},
\end{eqnarray*}
by the definition of $M_{L^2}$ and noting that $(E{\bf x})(E{\bf x})\transpose$ is a symmetric $L \times L$ matrix. Thus,
\begin{eqnarray*}
&&\mathrm{Var}^{*}({\bf 1}\transpose{\bf R}_j)\\
&\stackrel{\mathbb{P}-a.s.}{\rightarrow}& ({\bf a}\transpose K_{X,L}K_{W,L}^{-1}K_{X,L}{\bf b})^2 + \\
&& ({\bf a}\transpose K_{X,L}K_{W,L}^{-1} \otimes {\bf a}\transpose K_{X,L}K_{W,L}^{-1})\eta_{L}\eta_{L}\transpose(K_{W,L}^{-1}K_{X,L}{\bf b} \otimes K_{W,L}^{-1}K_{X,L}{\bf b}) + \\
&& [{\bf a}\transpose(D+A){\bf b}]^2 + ({\bf a}\transpose \otimes {\bf a}\transpose)\gamma_{L}\gamma_{L}\transpose({\bf b}\otimes{\bf b}) + \\
&& [{\bf a}\transpose K_{X,L}K_{W,L}^{-1}K_{X,L}{\bf a}][{\bf b}\transpose(D+A){\bf b}] + [{\bf b}\transpose K_{X,L}K_{W,L}^{-1}K_{X,L}{\bf b}][{\bf a}\transpose(D+A){\bf a}] + \\
&& 2[{\bf a}\transpose K_{X,L}K_{W,L}^{-1}K_{X,L}{\bf b}][{\bf a}\transpose(D+A){\bf b}].
\end{eqnarray*}
Now, observe that for any two vectors ${\bf x}, {\bf y} \in \mathbb{R}^{L}$ and any $L \times L$ covariance matrix $\Sigma$ equalling $\E({\bf S}{\bf S}\transpose)$ for a centered $L$-dimensional random variable ${\bf S}$, we have
\begin{eqnarray*}
&&({\bf x}\transpose \otimes {\bf x}\transpose)\mbox{vec}(\Sigma)(\mbox{vec}(\Sigma))\transpose({\bf y} \otimes {\bf y}) \\
&=& ({\bf x}\transpose \otimes {\bf x}\transpose)\E[\mbox{vec}({\bf Y}{\bf Y}\transpose)](\E[\mbox{vec}({\bf Y}{\bf Y}\transpose)])\transpose({\bf y} \otimes {\bf y})\\
&=& ({\bf x}\transpose \otimes {\bf x}\transpose)\E[{\bf Y} \otimes {\bf Y}]\E[({\bf Y}\transpose \otimes {\bf Y}\transpose]({\bf y} \otimes {\bf y})\\
&=& \E[({\bf x}\transpose{\bf Y})^2]\E[({\bf Y}\transpose{\bf y})^2] \ = \ ({\bf x}\transpose\Sigma{\bf x})({\bf y}\transpose\Sigma{\bf y}),
\end{eqnarray*}
which follows by using the fact that for two $L$-dimensional vectors ${\bf u}$ and ${\bf v}$, we have $({\bf u}\transpose \otimes {\bf u}\transpose)({\bf v} \otimes {\bf v}) = ({\bf u}\transpose{\bf v})^2$. Thus, we have that
\begin{eqnarray*}
&&\mathrm{Var}^{*}({\bf 1}\transpose{\bf R}_j)\\
&\stackrel{\mathbb{P}-a.s.}{\rightarrow}& ({\bf a}\transpose K_{X,L}K_{W,L}^{-1}K_{X,L}{\bf b})^2 + \\
&& ({\bf a}\transpose K_{X,L}K_{W,L}^{-1}K_{W,L}K_{W,L}^{-1}K_{X,L}{\bf a})({\bf b}\transpose K_{X,L}K_{W,L}^{-1}K_{W,L}K_{W,L}^{-1}K_{X,L}{\bf b}) + \\
&& [{\bf a}\transpose(D+A){\bf b}]^2 + [{\bf a}\transpose(D+A){\bf a}][{\bf b}\transpose(D+A){\bf b}] + \\
&& [{\bf a}\transpose K_{X,L}K_{W,L}^{-1}K_{X,L}{\bf a}][{\bf b}\transpose(D+A){\bf b}] + [{\bf b}\transpose K_{X,L}K_{W,L}^{-1}K_{X,L}{\bf b}][{\bf a}\transpose(D+A){\bf a}] + \\
&& 2[{\bf a}\transpose K_{X,L}K_{W,L}^{-1}K_{X,L}{\bf b}][{\bf a}\transpose(D+A){\bf b}] \\
&=& ({\bf a}\transpose K_{X,L}K_{W,L}^{-1}K_{X,L}{\bf b})^2 + \\
&& ({\bf a}\transpose K_{X,L}K_{W,L}^{-1}K_{X,L}{\bf a})({\bf b}\transpose K_{X,L}K_{W,L}^{-1}K_{X,L}{\bf b}) + [{\bf a}\transpose(D+A){\bf b}]^2 + [{\bf a}\transpose(D+A){\bf a}][{\bf b}\transpose(D+A){\bf b}] + \\
&& [{\bf a}\transpose K_{X,L}K_{W,L}^{-1}K_{X,L}{\bf a}][{\bf b}\transpose(D+A){\bf b}] + [{\bf b}\transpose K_{X,L}K_{W,L}^{-1}K_{X,L}{\bf b}][{\bf a}\transpose(D+A){\bf a}] + \\
&& 2[{\bf a}\transpose K_{X,L}K_{W,L}^{-1}K_{X,L}{\bf b}][{\bf a}\transpose(D+A){\bf b}] \\
&=& ({\bf a}\transpose(K_{X,L}K_{W,L}^{-1}K_{X,L} + D + A){\bf b})^2 + \\
&& ({\bf a}\transpose(K_{X,L}K_{W,L}^{-1}K_{X,L} + D + A){\bf a})({\bf b}\transpose(K_{X,L}K_{W,L}^{-1}K_{X,L} + D + A){\bf b}) \\
&=& ({\bf a}\transpose K_{W,L}{\bf b})^2 + ({\bf a}\transpose K_{W,L}{\bf a})({\bf b}\transpose K_{W,L}{\bf b})\\
&=& v_{{\bf a},{\bf b}}, \quad (\mbox{say}).
\end{eqnarray*}
Now, in order to derive the asymptotic distribution of $\sqrt{n}({\bf a}\transpose\overline{K}_{\zeta,L}{\bf b} - {\bf a}\transpose\widehat{K}{\bf b})$ we verify the Lyapunov condition. Recall that ${\bf a}\transpose\overline{K}_{\zeta,L}{\bf b} = n^{-1}\sum_{j=1}^{n} \mathbf{1}\transpose{\bf R}_j$ and ${\bf a}\transpose\widehat{K}{\bf b} = \E^{*}(n^{-1}\sum_{j=1}^{n} \mathbf{1}\transpose{\bf R}_j)$. Also, $s_n^2 := \mathrm{Var}^{*}(n^{-1}\sum_{j=1}^{n} \mathbf{1}\transpose{\bf R}_j) = n^{-1}\mathrm{Var}^{*}(\mathbf{1}\transpose{\bf R}_1)$ converges $\mathbb{P}$-almost surely to a positive constant (derived previously). We will now show that 
$$ \frac{1}{n^4s_n^4}\sum_{j=1}^{n} \E^{*}[(\mathbf{1}\transpose{\bf R}_j)^4] \rightarrow 0 \ \mbox{as} \ n \rightarrow \infty \quad \mbox{$\mathbb{P}$-almost surely}.$$
So, it is enough to show that 
$$ \frac{1}{n^2}\sum_{j=1}^{n} \E^{*}[(\mathbf{1}\transpose{\bf R}_j)^4] \rightarrow 0 \ \mbox{as} \ n \rightarrow \infty \quad \mbox{$\mathbb{P}$-almost surely}.$$
Now, 
\begin{eqnarray*}
&&\E^{*}[(\mathbf{1}\transpose{\bf R}_j)^4] \\
&\leq& 64\{\E^{*}[({\bf a}\transpose U^{*}_j)^4({\bf b}\transpose U^{*}_j)^4] + \E^{*}[({\bf a}\transpose V^{*}_j)^4({\bf b}\transpose V^{*}_j)^4] + \\
&& \E^{*}[({\bf a}\transpose U^{*}_j)^4({\bf b}\transpose V^{*}_j)^4] + \E^{*}[({\bf a}\transpose V^{*}_j)^4({\bf b}\transpose U^{*}_j)^4]\} \\
&\leq& 64\{(\E^{*}[({\bf a}\transpose U^{*}_j)^8])^{1/2}(\E^{*}[({\bf b}\transpose U^{*}_j)^8])^{1/2} + (\E^{*}[({\bf a}\transpose V^{*}_j)^8])^{1/2}(\E^{*}[({\bf b}\transpose V^{*}_j)^8])^{1/2} + \\
&& \E^{*}[({\bf a}\transpose U^{*}_j)^4]\E^{*}[({\bf b}\transpose V^{*}_j)^4] + \E^{*}[({\bf a}\transpose V^{*}_j)^4]\E^{*}[({\bf b}\transpose U^{*}_j)^4]\} \\
&\leq& 64\{||{\bf a}||^4||{\bf b}||^4||\widehat{\Theta}_q\widehat{K}_{W,L}||_F^8[n^{-1}\sum_{i=1}^{n} ||{\bf W}_i||^8] + 105[{\bf a}\transpose(\widehat{D}+\widehat{A}){\bf a}]^2[{\bf b}\transpose(\widehat{D}+\widehat{A}){\bf b}]^2 + \\
&& \E^{*}[({\bf a}\transpose U^{*}_j)^4]\E^{*}[({\bf b}\transpose V^{*}_j)^4] + \E^{*}[({\bf a}\transpose V^{*}_j)^4]\E^{*}[({\bf b}\transpose U^{*}_j)^4]\}.
\end{eqnarray*}
Since the right hand side of the above expression converges $\mathbb{P}$-almost surely to a positive finite constant, it follows that $n^{-2}\sum_{j=1}^{n} \E^{*}[(\mathbf{1}\transpose{\bf R}_j)^4] \rightarrow 0$ as $n \rightarrow \infty$ $\mathbb{P}$-almost surely.

\indent Hence, by the Lindeberg CLT
\begin{eqnarray*}
\frac{{\bf a}\transpose\overline{K}_{\zeta,L}{\bf b} - {\bf a}\transpose\widehat{K}{\bf b}}{s_n} \stackrel{\mathbb{P}^{*}-weakly}{\longrightarrow} N(0,1) \quad \mbox{$\mathbb{P}$-almost surely} \\
\Longrightarrow \sqrt{n}({\bf a}\transpose\overline{K}_{\zeta,L}{\bf b} - {\bf a}\transpose\widehat{K}{\bf b}) \stackrel{\mathbb{P}^{*}-weakly}{\longrightarrow} N(0,v_{{\bf a},{\bf b}}) \quad \mbox{$\mathbb{P}$-almost surely},
\end{eqnarray*}
where the second statement follows upon using Slutsky's theorem combined with the fact that $ns_n^2$ converges to $v_{{\bf a},{\bf b}}$ as $n \rightarrow \infty$ $\mathbb{P}$-almost surely. This concludes Step (2c). As a side remark, observe that even if the Gaussian assumption is not true, we still have the above weak convergence (under $H_{0,q}$) albeit with a different expression for the asymptotic variance.

\medskip

\noindent\textit{Step 2d: Asymptotic distribution of $nT^{*}_q$.} 
Denote the Procrustes aligned rank factorization of $\widehat{\Theta}_q$ by $\mathring{C}$. Since $\widehat{\Theta}_q$ converges $\mathbb{P}$-almost surely to $K_{X,L}$ under $H_{0,q}$, it can be shown that $\mathring{C}$ converges $\mathbb{P}$-almost surely to $C_0$ by using arguments similar to those used to prove the almost sure convergence of $\overline{C}$.
\\

Recall that we denoted $\widehat{\Theta} + \widehat{D}$ by $\widehat{K}$. First, use Taylor's formula to observe that
\begin{eqnarray*}
\mathrm{vec}(\nabla\overline{\Psi}(\overline{C})) &=& \mathrm{vec}(\nabla\overline{\Psi}(\mathring{C})) + \nabla^{2}\overline{\Psi}(\widetilde{C}_{1})\mathrm{vec}(\overline{C} - \mathring{C}) \\
\Rightarrow \ o_{\mathbb{P}^{*}}(1) &=& -4\mathrm{vec}(\{P_{L} \circ \sqrt{n}(\overline{K}_{\zeta,L} - \widehat{K})\}\mathring{C}) - 4I_{q} \otimes (P_{L} \circ \sqrt{n}(\overline{K}_{\zeta,L} - \widehat{K}))\mathrm{vec}(\overline{C} - \mathring{C}) \\
&& + \ \nabla^{2}\widecheck{\Psi}(\widetilde{C}_{1})(\sqrt{n}\mathrm{vec}(\overline{C} - \mathring{C})), \\
&=& -4\mathrm{vec}(\{P_{L} \circ \sqrt{n}(\overline{K}_{\zeta,L} - \widehat{K})\}\widehat{C}) - 4I_{q} \otimes (P_{L} \circ \sqrt{n}(\overline{K}_{\zeta,L} - \widehat{K}))\mathrm{vec}(\overline{C} - \mathring{C}) \\
&& + \ \nabla^{2}\Psi(\widetilde{C}_{1})(\sqrt{n}\mathrm{vec}(\overline{C} - \mathring{C})) - 4I_{q} \otimes (P_{L} \circ (\widehat{K} - K_{W,L}))(\sqrt{n}(\mathrm{vec}(\overline{C} - \mathring{C})),
\end{eqnarray*}
where $\widetilde{C}_{1} = \alpha\overline{C} + (1-\alpha)\mathring{C}$ for some $0 < \alpha < 1$. We have already proved that $\sqrt{n}(\overline{K}_{\zeta,L} - \widehat{K})$ converges $\mathbb{P}^{*}$-weakly to a centered Gaussian random matrix $Z_\dagger$ as $n \rightarrow \infty$, $\mathbb{P}$-almost surely,  {in Step (2c) above}. Further, we have that $\overline{C} \rightarrow C_{0}$ $\mathbb{P}^{*}$-almost surely, $\mathbb{P}$-almost surely  {(Step (2b))}, and similarly that $\mathring{C}$ converges $\mathbb{P}$-almost surely to $C_0$. This statement implies that $\overline{C} -\mathring{C} \rightarrow 0$ $\mathbb{P}^{*}$-almost surely, $\mathbb{P}$-almost surely. As earlier, by the invertibility of $\nabla^{2}\Psi(C_{0})$, there is an open neighbourhood $\mathcal{N}$ of $C_{0}$ where 
(i) the function $\nabla\Psi$ is invertible, \\
(ii) the function $(\nabla\Psi)^{-1}$ is continuously differentiable, and \\
(iii) $\nabla((\nabla\Psi)^{-1})(\nabla\Psi(C)) = (\nabla^{2}\Psi(C))^{-1}$ for any $C$ in that neighbourhood. \\
Since, $\overline{C} \rightarrow C_{0}$ in $\mathbb{P}^{*}$-probability $\mathbb{P}$-almost surely, we have $$\mathbb{P}^{*}(\widetilde{C}_{1} \ \mbox{lies in the open neighbourhood }\mathcal{N}) \rightarrow 1$$ 
as $n \rightarrow \infty$, $\mathbb{P}$-almost surely. Also, from the fact that $(\nabla\Psi)^{-1}$ is continuously differentiable in that neighbourhood, it follows that $\mathbb{P}^{*}(\nabla^{2}\Psi(\widetilde{C}) \ \mbox{is invertible}) \rightarrow 1$ as $n \rightarrow \infty$, $\mathbb{P}$-almost surely. Moreover, $(\nabla^{2}\Psi(\widetilde{C}_{1}))^{-1} \rightarrow (\nabla^{2}\Psi(C_{0}))^{-1}$ in $\mathbb{P}^{*}$-probability as $n \rightarrow \infty$, $\mathbb{P}$-almost surely.

It now follows from the above equations that
\begin{eqnarray*}
&& (I_{qL} - 4(\nabla^{2}\Psi(\widetilde{C}_{1}))^{-1}[I_{q} \otimes (P_{L} \circ (\widehat{K} - K_{W,L}))])\sqrt{n}\{\mathrm{vec}(\overline{C}) - \mathrm{vec}(\mathring{C})\} \\
&=& 4(\nabla^{2}\Psi(\widetilde{C}_{1}))^{-1}\mathrm{vec}(\{P_{L} \circ \sqrt{n}(\overline{K}_{\zeta,L} - \widehat{K})\}\mathring{C}) + o_{\mathbb{P}^{*}}(1) \\
&=& 4(\nabla^{2}\Psi(C_{0}))^{-1}\mathrm{vec}(\{P_{L} \circ \sqrt{n}(\overline{K}_{\zeta,L} - \widehat{K})\}\mathring{C}) + o_{\mathbb{P}^{*}}(1) \\
&=& 4(\nabla^{2}\Psi(C_{0}))^{-1}\mathrm{vec}(\{P_{L} \circ \sqrt{n}(\overline{K}_{\zeta,L} - \widehat{K})\}C_{0}) + o_{\mathbb{P}^{*}}(1)
\end{eqnarray*}
Since $I_{q} \otimes (P_{L} \circ (\widehat{K} - K_{W,L})) \rightarrow 0$ $\mathbb{P}$-almost surely, and $(\nabla^{2}\Psi(\widetilde{C}_{1}))^{-1} \rightarrow (\nabla^{2}\Psi(C_{0}))^{-1}$ $\mathbb{P}^{*}$-almost surely as $n \rightarrow \infty$, $\mathbb{P}$-almost surely, it follows that $\|(\nabla^{2}\Psi(\widetilde{C}_{1}))^{-1}[I_{q} \otimes (P_{L} \circ (\widehat{K} - K_{W,L}))]\|_{F} \rightarrow 0$  in $\mathbb{P}^{*}$-probability as $n \rightarrow \infty$, $\mathbb{P}$-almost surely. Hence, $\mathbb{P}^{*}(I_{qL} - 4(\nabla^{2}\Psi(\widetilde{C}_{1}))^{-1}[I_{q} \otimes (P_{L} \circ (\widehat{K} - K_{W,L}))] \ \mbox{is invertible}) \rightarrow 1$ as $n \rightarrow \infty$, $\mathbb{P}$-almost surely. Also, the inverse converges to $I_{qL}$ in $\mathbb{P}^{*}$-probability as $n \rightarrow \infty$, $\mathbb{P}$-almost surely. Combining these facts, we get that
\begin{eqnarray}
\sqrt{n}\{\mathrm{vec}(\overline{C}) - \mathrm{vec}(\mathring{C})\} &=& 4(\nabla^{2}\Psi(C_{0}))^{-1}\mathrm{vec}(\{P_{L} \circ \sqrt{n}(\overline{K}_{\zeta,L} - \widehat{K})\}C_{0}) + o_{\mathbb{P}^{*}}(1) \nonumber \\
&\stackrel{\mathbb{P}^{*}-weakly}{\rightarrow}& 4(\nabla^{2}\Psi(C_{0}))^{-1}\mathrm{vec}(\{P_{L} \circ Z_\dagger\}C_{0}) \label{eq12}
\end{eqnarray}
as $n \rightarrow \infty$, $\mathbb{P}$-almost surely.  

\indent Next note that for some $\widetilde{C}_{1} = \beta\overline{C} + (1-\beta)\mathring{C}$ with $0 < \beta < 1$, we have
\begin{eqnarray*}
\overline{\Psi}(\overline{C}) &=& \overline{\Psi}(\mathring{C}) + \langle \mathrm{vec}(\nabla\overline{\Psi}(\mathring{C})),\mathrm{vec}(\overline{C} - \mathring{C})\rangle + \frac{1}{2}\langle \nabla^{2}\overline{\Psi}(\widetilde{C}_{1})\mathrm{vec}(\overline{C} - \mathring{C}),\mathrm{vec}(\overline{C} - \mathring{C})\rangle \nonumber\\
&=& \overline{\Psi}(\mathring{C}) + \langle \mathrm{vec}(\nabla\overline{\Psi}(\mathring{C})),\mathrm{vec}(\overline{C} - \mathring{C})\rangle + \frac{1}{2}\langle \nabla^{2}\widecheck{\Psi}(\widetilde{C}_{1})\mathrm{vec}(\overline{C} - \mathring{C}),\mathrm{vec}(\overline{C} - \mathring{C})\rangle \nonumber \\
&& - \ 2\langle \{I_{q} \otimes (P_{L} \circ (\overline{K}_{\zeta,L} - \widehat{K}))\}\mathrm{vec}(\overline{C} - \mathring{C}),\mathrm{vec}(\overline{C} - \mathring{C})\rangle \\
&=& \overline{\Psi}(\mathring{C}) + \langle \mathrm{vec}(\nabla\overline{\Psi}(\mathring{C})),\mathrm{vec}(\overline{C} - \mathring{C})\rangle + \frac{1}{2}\langle \nabla^{2}\Psi(\widetilde{C}_{1})\mathrm{vec}(\overline{C} - \mathring{C}),\mathrm{vec}(\overline{C} - \mathring{C})\rangle \nonumber \\
&& - \ 2\langle \{I_{q} \otimes (P_{L} \circ (\widehat{K} - K_{W,L}))\}\mathrm{vec}(\overline{C} - \mathring{C}),\mathrm{vec}(\overline{C} - \mathring{C})\rangle \\
&& - \ 2\langle \{I_{q} \otimes (P_{L} \circ (\overline{K}_{\zeta,L} - \widehat{K}))\}\mathrm{vec}(\overline{C} - \mathring{C}),\mathrm{vec}(\overline{C} - \mathring{C})\rangle 
\end{eqnarray*}
By \eqref{eq12}, the facts that $\widehat{K} \rightarrow K_{W,L}$ as $n \rightarrow \infty$ $\mathbb{P}$-almost surely, and $\overline{K}_{\zeta,L} - \widehat{K} \rightarrow 0$ $\mathbb{P}^{*}$-almost surely as $n \rightarrow \infty$ $\mathbb{P}$-almost surely, the convergence $\mathbb{P}^{*}$-almost surely of $\widetilde{C}_{1}$ to $C_{0}$ as $n \rightarrow \infty$ $\mathbb{P}$-almost surely, and the continuity of $\nabla^{2}\Psi$, it follows that
\begin{eqnarray*}
\overline{\Psi}(\overline{C}) &=& \overline{\Psi}(\mathring{C}) + \langle \mathrm{vec}(\nabla\overline{\Psi}(\mathring{C})),\mathrm{vec}(\overline{C} - \mathring{C})\rangle \\
&& + \ \frac{1}{2}\langle \nabla^{2}\Psi(C_{0})\mathrm{vec}(\overline{C} - \mathring{C}),\mathrm{vec}(\overline{C} - \mathring{C})\rangle + o_{\mathbb{P}^{*}}(n^{-1}) \\
&=& \|P_{L} \circ (\overline{K}_{\zeta,L} - \widehat{K})\|_{F}^{2} - 4\langle \mathrm{vec}(\{P_{L} \circ (\overline{K}_{\zeta,L} - \widehat{K})\}\mathring{C}),\mathrm{vec}(\overline{C} - \mathring{C})\rangle \\
&& + \ \frac{1}{2}\langle \nabla^{2}\Psi(C_{0})\mathrm{vec}(\overline{C} - \mathring{C}),\mathrm{vec}(\overline{C} - \mathring{C})\rangle + o_{\mathbb{P}^{*}}(n^{-1}) \\
&=& \|P_{L} \circ (\overline{K}_{\zeta,L} - \widehat{K})\|_{F}^{2} - 4\langle \mathrm{vec}(\{P_{L} \circ (\overline{K}_{\zeta,L} - \widehat{K})\}C_{0}),\mathrm{vec}(\overline{C} - \mathring{C})\rangle \\
&& + \ 4\langle \mathrm{vec}(\{P_{L} \circ (\overline{K}_{\zeta,L} - \widehat{K})\}(\mathring{C} - C_{0})),\mathrm{vec}(\overline{C} - \mathring{C})\rangle \\
&& + \ \frac{1}{2}\langle \nabla^{2}\Psi(C_{0})\mathrm{vec}(\overline{C} - \mathring{C}),\mathrm{vec}(\overline{C} - \mathring{C})\rangle + o_{\mathbb{P}^{*}}(n^{-1}) \\
&=& \|P_{L} \circ (\overline{K}_{\zeta,L} - \widehat{K})\|_{F}^{2} - 4\langle \mathrm{vec}(\{P_{L} \circ (\overline{K}_{\zeta,L} - \widehat{K})\}C_{0}),\mathrm{vec}(\overline{C} - \mathring{C})\rangle \\
&& + \ \frac{1}{2}\langle \nabla^{2}\Psi(C_{0})\mathrm{vec}(\overline{C} - \mathring{C}),\mathrm{vec}(\overline{C} - \mathring{C})\rangle + o_{\mathbb{P}^{*}}(n^{-1}) \\
&=& \|P_{L} \circ (\overline{K}_{\zeta,L} - \widehat{K})\|_{F}^{2} \\
&& - 16\langle \mathrm{vec}(\{P_{L} \circ (\overline{K}_{\zeta,L} - \widehat{K})\}C_{0}), \{(\nabla^{2}\Psi(C_{0}))^{-1}\mathrm{vec}(\{P_{L} \circ (\overline{K}_{\zeta,L} - \widehat{K})\}C_{0}) + o_{\mathbb{P}^{*}}(n^{-1/2})\}\rangle  \\
&& + \ 8\langle \nabla^{2}\Psi(C_{0})\{(\nabla^{2}\Psi(C_{0}))^{-1}\mathrm{vec}(\{P_{L} \circ (\overline{K}_{\zeta,L} - \widehat{K})\}C_{0}) + o_{\mathbb{P}^{*}}(n^{-1/2})\},\\
&& \hspace{4cm} \{(\nabla^{2}\Psi(C_{0}))^{-1}\mathrm{vec}(\{P_{L} \circ (\overline{K}_{\zeta,L} - \widehat{K})\}C_{0}) + o_{\mathbb{P}^{*}}(n^{-1/2})\}\rangle + o_{\mathbb{P}^{*}}(n^{-1}) \\
&=& \|P_{L} \circ (\overline{K}_{\zeta,L} - \widehat{K})\|_{F}^{2} \\
&& - \ 16\langle \mathrm{vec}(\{P_{L} \circ (\overline{K}_{\zeta,L} - \widehat{K})\}C_{0}), (\nabla^{2}\Psi(C_{0}))^{-1}\mathrm{vec}(\{P_{L} \circ (\overline{K}_{\zeta,L} - \widehat{K})\}C_{0})\rangle \\
&& + \ 8\langle \mathrm{vec}(\{P_{L} \circ (\overline{K}_{\zeta,L} - \widehat{K})\}C_{0}),(\nabla^{2}\Psi(C_{0}))^{-1}\mathrm{vec}(\{P_{L} \circ (\overline{K}_{\zeta,L} - \widehat{K})\}C_{0})\rangle + o_{\mathbb{P}^{*}}(n^{-1}) \\
&=& \|P_{L} \circ (\overline{K}_{\zeta,L} - \widehat{K})\|_{F}^{2} \\
&& - \ 8\langle \mathrm{vec}(\{P_{L} \circ (\overline{K}_{\zeta,L} - \widehat{K})\}C_{0}), (\nabla^{2}\Psi(C_{0}))^{-1}\mathrm{vec}(\{P_{L} \circ (\overline{K}_{\zeta,L} - \widehat{K})\}C_{0})\rangle  + o_{\mathbb{P}^{*}}(n^{-1}).
\end{eqnarray*}
Thus, 
\begin{eqnarray*}
n\overline{\Psi}(\overline{C}) &=& \|P_{L} \circ \sqrt{n}(\overline{K}_{\zeta,L} - \widehat{K})\|_{F}^{2} \\
&& - \ 8\langle \mathrm{vec}(\{P_{L} \circ \sqrt{n}(\overline{K}_{\zeta,L} - \widehat{K})\}C_{0}), (\nabla^{2}\Psi(C_{0}))^{-1}\mathrm{vec}(\{P_{L} \circ \sqrt{n}(\overline{K}_{\zeta,L} - \widehat{K})\}C_{0})\rangle  + o_{\mathbb{P}^{*}}(1) \\
&\stackrel{\mathbb{P}^{*}-weakly}{\rightarrow}& \|P_{L} \circ Z_\dagger\|_{F}^{2} - 8~\langle \mathrm{vec}(\{P_{L} \circ Z_\dagger\}C_{0}),(\nabla^{2}\Psi(C_{0}))^{-1}\mathrm{vec}(\{P_{L} \circ Z_\dagger\}C_{0})\rangle
\end{eqnarray*}
as $n \rightarrow \infty$ $\mathbb{P}$-almost surely. Thus, the bootstrap version of the test statistic, namely, $T^{*}_q = \min_{\Theta \in {\cal M}_{q}} \|P_{L} \circ (\overline{K}_{\zeta,L} - \Theta)\|_{F}^{2}$ satisfies
\begin{eqnarray}
nT^{*}_q &=& n\min_{C \in \mathbb{R}^{L \times q}} \overline{\Psi}(C) \nonumber\\
&\stackrel{\mathbb{P}^{*}-weakly}{\rightarrow}& \|P_{L} \circ Z_\dagger\|_{F}^{2} - 8~\langle \mathrm{vec}(\{P_{L} \circ Z_\dagger\}C_{0}),(\nabla^{2}\Psi(C_{0}))^{-1}\mathrm{vec}(\{P_{L} \circ Z_\dagger\}C_{0})\rangle \nonumber \\
&\stackrel{\label{eq13}}{=}& \|P_{L} \circ Z_\dagger\|_{F}^{2} - 8~(\mathrm{vec}(P_{L} \circ Z_\dagger))\transpose \{(C_{0} \otimes I_{L})(\nabla^{2}\Psi(C_{0}))^{-1}(C_{0}\transpose  \otimes I_{L})\}\mathrm{vec}(P_{L} \circ Z_\dagger)   
\end{eqnarray}
as $n \rightarrow \infty$ $\mathbb{P}$-almost surely.

\medskip

\noindent\textit{Step 2e: Bootstrap weak limit coincides with original weak limit in Theorem \ref{thm1}.} We will now conclude Step 2 by showing that $Z_\dagger$ has the same distribution as $Z$, which was the weak limit of $\sqrt{n}(\widehat{K}_{W,L} - K_{W,L})$ under $H_{0,q}$ in the statement of Theorem \ref{thm1}. For this, observe that it is enough to show that $\sqrt{n}({\bf a}\transpose\widehat{K}_{W,L}{\bf b} - {\bf a}\transpose K_{W,L}{\bf b})$ converges $\mathbb{P}$-weakly to $N(0,v_{{\bf a},{\bf b}})$. Note that ${\bf a}\transpose\widehat{K}_{W,L}{\bf b} = n^{-1}\sum_{i=1}^{n} ({\bf a}\transpose{\bf W}_i)({\bf b}\transpose{\bf W}_i)$. Now, under Gaussainity of the ${\bf W}_i$'s and the assumption that $\E({\bf W}_1) = 0$, we have
\begin{eqnarray*}
&& \mathrm{Var}[({\bf a}\transpose{\bf W}_i)({\bf b}\transpose{\bf W}_i)] \\
&=& \E[({\bf a}\transpose{\bf W}_i)^2({\bf b}\transpose{\bf W}_i)^2] - ({\bf a}\transpose K_{W,L}{\bf b})^2 \\
&=& (\mbox{vec}({\bf a}{\bf a}\transpose))\transpose\E[({\bf W}_1 \otimes {\bf W}_1)({\bf W}_1 \otimes {\bf W}_1)\transpose]\mbox{vec}({\bf b}{\bf b}\transpose) - ({\bf a}\transpose K_{W,L}{\bf b})^2 \\
&=& (\mbox{vec}({\bf a}{\bf a}\transpose))\transpose\{(\mathbb{I}_{L^2} + M_{L^2})(K_{W,L} \otimes K_{W,L}) + \eta_L\eta_L\transpose\}\mbox{vec}({\bf b}{\bf b}\transpose) - ({\bf a}\transpose K_{W,L}{\bf b})^2 \\
&=& (\mbox{vec}({\bf a}{\bf a}\transpose))\transpose \{M_{L^2}(K_{W,L} \otimes K_{W,L}) + \eta_L\eta_L\transpose\}\mbox{vec}({\bf b}{\bf b}\transpose) \\
&=& ({\bf a}\transpose K_{W,L}{\bf b})^2 + ({\bf a}\transpose K_{W,L}{\bf a})({\bf b}\transpose K_{W,L}{\bf b}) \ = \ v_{{\bf a},{\bf b}}.
\end{eqnarray*}
So, using the classical CLT, it follows that $\sqrt{n}({\bf a}\transpose\widehat{K}_{W,L}{\bf b} - {\bf a}\transpose K_{W,L}{\bf b})$ converges $\mathbb{P}$-weakly to $N(0,v_{{\bf a},{\bf b}})$. Hence, $Z\dagger \stackrel{d}{=} Z$ under Gaussianity of the observations and $H_{0,q}$. Consequently, the asymptotic distribution of the bootstrap statistic is the same as that of the original statistic under Gaussianity and $H_{0,q}$. \\

\indent Wrapping up Step 2 of our proof, we are now ready to establish statement (a) of the present theorem.
Denote by $H^{*}_q$ the empirical CDF of $nT^{*}_q$ and that of $T^{*}_q$ by $F^{*}_q$. Then, $F^{*}_q(x) = H^{*}_q(nx)$ for each $x \in \mathbb{R}$ and each $n \geq 1$. Let $G^{*}_q$ denote the generalized inverse CDF of $H^{*}_q$. From the usual properties of the generalized inverse of a cdf, it follows that the  {two events $\{F^{*}_{q}(T_q) \geq u\}$ (equivalently,  $\{H^{*}_q(nT_q) \geq u\}$) and $\{nT_q \geq G^{*}_q(u)\}$ are the same} for any $u \in (0,1)$. Let us denote the $u$-quantile of the asymptotic limit, say $Y_\dagger$, of $nT^{*}_q$ by $y_{u}$. Since $H^{*}_q$ converges $\mathbb{P}^{*}$-weakly, $\mathbb{P}$-almost surely to $Y_\dagger$, and $Y_\dagger$ has a continuous distribution (continuous map of the Gaussian random matrix $Z_\dagger$), it follows from Lemma 21.2 in \cite{vaart_asymp} that $G^{*}_q(u)$ converges to $y_u$ for all $u \in (0,1)$ $\mathbb{P}$-almost surely. Now, $nT_{q}$ converges $\mathbb{P}$-weakly to $Y$ which has the same distribution as $Y_\dagger$ (since $Z$ and $Z_\dagger$ have the same distributions). So, by Slutsky's theorem, $nT_q - G^{*}_q(u)$ converges $\mathbb{P}$-weakly to $Y - y_u$ for any $u \in (0,1)$. Hence, 
$$ \mathbb{P}[F^{*}_q(T_q) \geq u] \ = \ \mathbb{P}[H^{*}_q(nT_q) \geq u] \ = \ \mathbb{P}[nT_q \geq G^{*}_q(u)] \ \rightarrow \ \mathbb{P}[Y \geq y_u] = 1-u$$
for any $u \in (0,1)$. Observe that $p^{*}_q = 1 - F^{*}_q(T_q)$. Thus, 
$$ \mathbb{P}[p^{*}_q \leq u] = \mathbb{P}[F^{*}_q(T_q) \geq 1-u] \ \rightarrow \ u$$
for any $u \in (0,1)$. This completes the proof of the first conclusion of the present theorem.

\medskip

\noindent\textit{Step 3: Asymptotic theory for bootstrap under $H_{1,q}$.} For the second claim in the theorem, assume $H_{1,q}$. We will first consider the case when the global null is also still true, i.e., $H_{0,r}$ is true for some $r \in \{q+1,q+2,\ldots,d\}$. In this case, it has been proven  {in Step 1} that $M$ converges almost surely to $r$ so that $\widehat{D}$ converges almost surely to $D$. Though $\widehat{\Theta}_q$ might not converge, it has also been shown in Step 1 of the proof that $\widehat{\Theta}_q$ is almost surely tight. So, $\mathbb{P}(\widehat{\Theta}_q \ \mbox{lies in a compact set eventually as} \ n \rightarrow \infty) = 1$. We will work with an $\omega$ that satisfies all the almost sure convergences and laws of iterated logarithm as needed earlier along with the previous tightness requirement (such $\omega$'s  comprise an event of $\mathbb{P}$-measure one). So, by the compactness condition, there will exist a subsequence $\{n'\}$ (possibly depending on $\omega$) such that $\widehat{\Theta}_q$ converges to  {some} $K_q = K_q(\omega)$ as $n' \rightarrow \infty$. Clearly, $\widehat{D}$ converges to $D$ along this subsequence. We will work by viewing this subsequence as our original sequence, and all convergence statements will be as $n' \rightarrow \infty$.

\indent Since $\widehat{K}$ converges to $K_q + D$ along $\{n'\}$, observe that we will be able to prove (following the same arguments as  {in Step 2d}) that $\sqrt{n}({\bf a}\transpose\overline{K}_{\zeta,L}{\bf b} - {\bf a}\transpose\widehat{K}{\bf b})$ converges $\mathbb{P}^{*}$-weakly to $N(0,u_{{\bf a},{\bf b}})$ along $\{n'\}$ for each ${\bf a}, {\bf b} \in \mathbb{R}^L$. Note that the limiting variance term will be different from the that is Step 2, which assumes $H_{0,q}$ to be true. We can denote the limiting random matrix of $\sqrt{n}(\overline{K}_{\zeta,L} - \widehat{K})$ by $Z^{(1)}_{\dagger}$.

\indent Next observe that since $\mathrm{rank}(\widehat{\Theta}_q) \leq q$ (by construction), it follows  from the definition of a minimum that
\begin{eqnarray}
nT^{*}_q &=& n\min_{\Theta \in {\cal M}_q} ||P_L \circ (\overline{K}_{\zeta,L} - \Theta)||_F^2 \nonumber \\
&\leq& n||P_L \circ (\overline{K}_{\zeta,L} - \widehat{\Theta}_q)||_F^2 \nonumber \\
&=& n||P_L \circ (\overline{K}_{\zeta,L} - \widehat{K})||_F^2 \ = \ ||P_L \circ \{\sqrt{n}(\overline{K}_{\zeta,L} - \widehat{K})\}||_F^2 \label{tight-alt-eq1}
\end{eqnarray}
for each $n \geq 1$. So, $n'T^{*}_q$ is $\mathbb{P}^{*}$-tight along $\{n'\}$ since it is bounded above by a $\mathbb{P}^{*}$-weakly convergent sequence. Note that since $T_q > 0$ almost surely, the entire sequence $\{nT_q\}$ diverges to $+\infty$ as $n \rightarrow \infty$ $\mathbb{P}$-almost surely. So, for each $\omega$ in a $\mathbb{P}$-measure one set and along the corresponding sequence $\{n'\}$ (depending possibly on $\omega$ as discussed thus far in the proof), we have 
\begin{eqnarray*}
F^{*}_q(T_q) = H^{*}_q(n'T_q) = \mathbb{P}^{*}(n'T^{*}_q \leq nT_q) \rightarrow 1
\end{eqnarray*} 
as $n'\rightarrow \infty$. In fact, a stronger statement is actually true -- for each $u \in (0,1)$, there exists $n(\omega) \geq 1$ such that $F^{*}_q(T_q) \geq u$ for all $n > n(\omega)$. This is because of the following: if there exists a infinite sequence $\{\widetilde{n}\}$ (possibly depending on $u$ and $\omega$) such that $F^{*}_q(T_q) < u$ along $\{\widetilde{n}\}$, one can find a further subsequence, say $\{\widecheck{n}\}$ (obtained in a similar way as discussed previously in the case of the original subsequence $\{n'\}$) such that $F^{*}_q(T_q) \rightarrow 1$ along this subsequence $\{\widecheck{n}\}$ of $\{\widetilde{n}\}$. This would lead to a contradiction. Hence,
$$ \mathbb{P}\{F^{*}_q(T_q) \geq u \ \mbox{eventually}\} = 1 \ \mbox{for all} \ u \in (0,1).$$
Observe that $p^{*}_q = 1 - F^{*}_q(T_q)$. So, replacing $u$ by $1-u$ in the previous displayed equation, it follows that
$$ \mathbb{P}\{p^{*}_q \leq u \ \mbox{eventually}\} = 1 \ \mbox{for all} \ u \in (0,1).$$
Consequently, $\mathbb{P}\{p^{*}_q \leq u\} \rightarrow 1$ as $n \rightarrow \infty$ for each $u \in (0,1)$. This completes the proof of the second statement of the present theorem in case $H_{0,q}$ is not true but the global null is true.

\indent Finally, consider the situation when the global null is not true, i.e. $H_1$ is true. In this case, it follows that both $\widehat{\Theta}_q$ and $\widehat{D}$ are strongly tight. A simple extension of the subsequence arguments provided in the previous situation (to accommodate for $\widehat{D}$ in addition to $\widehat{\Theta}_q$) carries over and proves the second statement of the present theorem. 

\end{proof}

\begin{lemma} \label{aux-lemma1}
In the setting and notation of Theorem \ref{thm3} and its proof, all of $t_1 - \E^{*}(U^{*}_1U^{*\transpose}_1)$, $t_2 - \E^{*}(V^{*}_1V^{*\transpose}_1)$ and $t_3$ converges to zero $\mathbb{P}^{*}$-almost surely, $\mathbb{P}$-almost surely, where $t_1 = n^{-1}\sum_{j=1}^{n} U^{*}_jU^{*\transpose}_j$, $t_2 = n^{-1}\sum_{j=1}^{n} V^{*}_jV^{*\transpose}_j$ and $t_3 = n^{-1}\sum_{j=1}^{n} (U^{*}_jV^{*\transpose}_j + V^{*}_jU^{*\transpose}_j)$.
\end{lemma}

\begin{proof}
We will only give the proof of the convergence of $t_1 - \E^{*}(U^{*}_1U^{*\transpose}_1)$. The proofs of the convergence of the other two terms are similar. Note that $t_1 - \E^{*}(U^{*}_1U^{*\transpose}_1) = n^{-1}\sum_{j=1}^{n} S_j$, where $S_j = U^{*}_jU^{*\transpose}_j - \E^{*}(U^{*}_1U^{*\transpose}_1)$ has zero mean. We will show that
\begin{eqnarray}
\sum_{j=1}^{\infty} \E^{*}[(n^{-1}\sum_{j=1}^n {\bf a}\transpose S_j{\bf b})^4] < \infty  \label{aux-lemma1-eq0}
\end{eqnarray}
$\mathbb{P}$-almost surely for any choice of ${\bf a}, {\bf b} \in \mathbb{R}^{L}$. We will then be able to conclude the first convergence by using the Borel-Cantelli lemma. Define $s_{j} = {\bf a}\transpose S_j{\bf b}$ (omitting the dependence on ${\bf a}$ and ${\bf b}$ for simplicity of notation). Observe that 
\begin{eqnarray}
\E^{*}[(n^{-1}\sum_{j=1}^n s_j)^4] &=& n^{-4}\left[\sum_{j=1}^{n} \E^{*}(s_j^4) + 6\sum_{1 \leq j \neq j' \leq n} \E^{*}(s_j^2)\E^{*}(s_j'^2)\right] \nonumber \\
&\leq& n^{-3}\E^{*}(s_1^4) + 6n^{-2}[\E^{*}(s_1^2)]^2 \label{aux-lemma1-eq1}
\end{eqnarray}
since the other terms vanish by using the fact that the $s_j$'s are i.i.d. and have zero mean. Now,
\begin{eqnarray*}
\E^{*}(s_1^2) = \mathrm{Var}^{*}(s_1) = \mathrm{Var}^{*}[({\bf a}\transpose U^{*}_1)({\bf b}\transpose U^{*}_1)].
\end{eqnarray*}
We have already derived the expression of the above variance term while deriving the weak convergence of $\sqrt{n}({\bf a}\transpose\overline{K}_{\zeta,L}{\bf b} - {\bf a}\transpose\widehat{K}{\bf b})$ in the proof of Theorem \ref{thm3} (see Step (2c) for details). Since the above variance term converges $\mathbb{P}$-almost surely to a constant (depending on ${\bf a}$ and ${\bf b}$), the first term in (\ref{aux-lemma1-eq1}) is $O(n^{-2})$ as $n \rightarrow \infty$ $\mathbb{P}$-almost surely. Next, observe that
\begin{eqnarray}
\E^{*}(s_1^4) &\leq& 8[\E^{*}\{({\bf a}\transpose U^{*}_1)^4({\bf b}\transpose U^{*}_1)^4\} + [{\bf a}\transpose\E^{*}(U^{*}_1U^{*\transpose}_1){\bf b}]^4] \nonumber \\
&=& 8[\E^{*}\{({\bf a}\transpose U^{*}_1)^4({\bf b}\transpose U^{*}_1)^4\} + ({\bf a}\transpose\widehat{\Theta}_q\widehat{K}_{W,L}^{-1}\widehat{\Theta}_q{\bf b})^4]. \label{aux-lemma1-eq2}
\end{eqnarray}
We have already shown in the proof of Theorem 3 (see Step (2c) for details) that the first term in (\ref{aux-lemma1-eq2}) is bounded above by a quantity which converges to a constant $\mathbb{P}$-almost surely. Further, $\widehat{\Theta}_q\widehat{K}_{W,L}^{-1}\widehat{\Theta}_q$ converges to $K_{X,L}K_{W,L}^{-1}K_{X,L}$ $\mathbb{P}$-almost surely. These facts show that the first term in (\ref{aux-lemma1-eq1}) is $O(n^{-3})$ as $n \rightarrow \infty$ $\mathbb{P}$-almost surely. So, (\ref{aux-lemma1-eq0}) is true $\mathbb{P}$-almost surely. This completes the proof of the almost sure convergence of $t_1 - \E^{*}(U^{*}_1U^{*\transpose}_1)$.
\end{proof}

\subsection{On the Critical Grid Size}\label{critical_value}

As noted in Remark \ref{grid_remark}, the critical value $L_{\dagger}<\infty$ in Theorem \ref{theorem-identifiability} depends on the choice of hypothesis boundary $d$, and the spectrum of $k_X$. The purpose of this section is to show that, for a very wide variety of continuous kernels, we have 
$$L_\dagger\leq 2d+1.$$
Namely, we will show the following for continuous $k_X$:
\begin{itemize}
\item If $k_X$ is strictly positive definite, and hence certainly of infinite rank, then the bound $L_\dagger\leq 2d+1$ holds true for \emph{any} placement of the pairwise distinct grid nodes (not just for regular grids) without any additional assumptions on the form of the eigenfunctions.\\

\item {If $k_X$ is positive semidefinite (whether finite or infinite rank), and the Reproducing Kernel Hilbert Space (RKHS) of $k_X$  contains
}
\begin{itemize}

\smallskip
\item the collection of monomials $\{1,x,x^2,...,x^{d-1}\}$, then the bound $L_\dagger\leq 2d+1$ holds true for \emph{any} placement of the distinct grid nodes (not necessarily equi-spaced); if the monomials are replaced by $d$ linearly independent polynomials of highest degree greater than $d-1$, the bound $L_\dagger\leq 2d+1$ still holds true for all but finitely many configurations of the grid nodes.

\smallskip
\item  the collection of the first $d$ Fourier basis elements, then the bound $L_\dagger\leq 2d+1$ holds true for \emph{any} placement of the grid nodes; if the first $d$ Fourier elements are replaced by $d$ arbitrary linearly independent trigonometric polynomials, the bound $L_\dagger\leq 2d+1$ still holds true for all but finitely many configurations of the grid nodes.

\smallskip
\item a collection $\{1,F(x),F^2(x),...,F^{d-1}(x)\}$, where $F:[0,1]\rightarrow[0,1]$ is any strictly increasing function, then the bound $L_\dagger\leq 2d+1$ holds true for \emph{any} placement of the distinct grid nodes (not necessarily equi-spaced); if the exponents are replaced by $d$ arbitrary exponents, the bound $L_\dagger\leq 2d+1$ still holds true for all but finitely many configurations of the grid nodes. 

\smallskip
\item a collection of $d$ functions $\{h_j\}_{j=1}^d$ that are linearly independent on any subset $K\subseteq [0,1]$ of positive Lebesgue measure, then the bound $L_\dagger\leq 2d+1$ holds true for almost all configurations of the grid nodes. Collections of functions $\{h_j\}_{j=1}^{d}$ of this type are ubiquitous, and include collections of linearly independent splines, or more generally of linearly independent piecewise analytic functions; such collections need not be comprised of smooth functions alone. One can easily produce examples of collections that contain nowhere differentiable functions\footnote{To see a concrete case, take $\{h_j\}$ to be $d$ independent realisations of a standard Brownian motion on $[0,1]$.}.

\end{itemize}
\end{itemize}
Notice that that the eigenfunctions $\{\varphi_n\}$ of $k_X$ are by default elements of the RKHS of $k_X$. So if the eigensystem of $k_X$ includes $d$ (orthonormalised) functions as described in the cases above, then certainly so does $\mathrm{RKHS}(k_X)$.

\medskip
\noindent To show why the statements listed above hold true, let $r_{\mathrm{true}}\leq \infty$ be the true rank of $k_X$. Since $k_X$ is continuous, it admits the Mercer expansion
$$
k_X(s,t)=\sum_{m=1}^{r_{\mathrm{true}}}\lambda_m \varphi_m(s)\varphi_m(t).
$$
This yields
$$k_X(t_i,t_j)=\sum_{m=1}^{r_{\mathrm{true}}}\lambda_m \varphi_m(t_i)\varphi_m(t_j)$$
on our grid points $\{t_j\}_{j=1}^{L}$. It follows that the $L\times L$ matrix $K_{X,L}$ is represented as
\begin{equation}
K_{X,L}=U(t_1,...,t_L) U\transpose  (t_1,...,t_L)
\label{Krep}
\end{equation}\label{U-factor}
{where the $k$-th row of $U\in\mathbb{R}^{L\times r_{\mathrm{true}}}$ is comprised of the sequence $\{\lambda_n^{1/2}\varphi_n(t_k)\}_{n=1}^{r_{\mathrm{true}}}$, for $1\leq k\leq L$.} Schematically,
\begin{equation} \label{Umatrix}
U(t_1,...,t_L)= 
\left(  \begin{array}{cccc}
\lambda_1^{1/2}\varphi_1(t_1) & \lambda_2^{1/2}\varphi_2(t_1)  &\lambda_3^{1/2}\varphi_{3}(t_1) & \hdots  \\
\lambda_1^{1/2}\varphi_1(t_2) & \lambda_2^{1/2}\varphi_2(t_2)  &\lambda_3^{1/2}\varphi_{3}(t_2) & \hdots \\
\vdots & \vdots &\vdots &   \\
\lambda_1^{1/2}\varphi_1(t_L) & \lambda_2^{1/2}\varphi_2(t_L) &\lambda_3^{1/2}\varphi_{3}(t_L) & \hdots
\end{array}
\right)\end{equation}  
where the horizontal dots signify that there may be infinitely or finitely many columns depending on whether $r_{\mathrm{true}}<\infty$ or $r_{\mathrm{true}}=\infty$.

If $K_{X,L}^{A,B}$ is the submatrix of $K_{X,L}$ obtained by retaining rows in the index set $A\subseteq \{1,...,L\}$ and columns in the index set $B \subseteq \{1,...,L\}$, then
$$K_{X,L}^{A,B}=U^A (U^B)\transpose  ,$$
where $U^A$ (resp. $U^B$) represents the submatrix of $U$ obtained by retaining rows in the index set $A$ (resp. $B$).  {Formally, we can view the matrices $U^A$ and $U^B$ as linear operators from $(\ell_2)^d$ into $\mathbb{R}^d$. Continuity of $k_X$ ensures that they are indeed finite rank Hilbert-Schmidt}\footnote{To see this, note that 
$$\| U^A \|_{\mathrm{HS}}^2=\mathrm{trace}[ (U^A)\transpose   U^A] =\mathrm{trace}[U^A(U^A)\transpose  ]=\mathrm{trace}\left[ \{k_X(t_i,t_j)\}_{i,j\in A}\right]=\sum_{i\in A}k_X(t_i,t_i)<\infty.$$
}

So $\mathrm{det}(K_{X,L}^{A,B})\neq 0$ for a pair of index sets $A,B\subseteq \{1,...,L\}$ of cardinality $d\leq L$ if and only $U^A$ and $U^B$ are both of full column rank $d$. In summary, since $A$ and $B$ are arbitrary, we have the following implication:
\begin{center}
\begin{footnotesize}
\fbox{$U^A$ of column rank $d$ for any $A\subset \{1,...,L\}$ of cardinality $d$  $\implies$ all $d$-minors of $K_{X,L}$ are non-vanishing}
\end{footnotesize}
\end{center}
We will show in Subsection \ref{eigensystems} that $U^A(t_1,...,t_L)$ is indeed of full column rank $d$ for any index set $A$ of cardinality $d\leq r_{\mathrm{true}}$ for the scenarios described at the top of this Section. First, though we will show in Subsection \ref{minor-subsection} why $L_\dagger\leq 2d+1$ when the $d$-minors of $K_{X,L}$ are non-vanishing.

\subsubsection{Showing that $L_\dagger\leq 2d+1$ When the $d$-Minors of $K_{X,L}$ Are Non-Vanishing}\label{minor-subsection}
Said differently, let us show that whenever all order $d$ minors of $K_{X,L}$ can be guaranteed to be non-zero, the critical value satisfies $L_{\dagger}\leq 2d+1$. We will do this by showing that when $L\geq 2d+1$, each diagonal entry of $\{K_{X,L}(i,i)\}_{i=1}^{L}$ of $K_{X,L}$ is a (rational) function of some of the off-diagonal entries $\{K_{X,L}(i,j)\}_{i\neq j}$ (and thus the diagonal entries are uniquely imputed by the off-diagonal entries).  

Assume first that $L=2r_{\mathrm{true}}+1$ exactly. Let $D$ be the $(r_{\mathrm{true}}+1)\times (r_{\mathrm{true}}+1)$ submatrix of $K_{X,L}$ obtained by retaining the last $ (r_{\mathrm{true}}+1)$ rows and first $(r_{\mathrm{true}}+1)$ columns of $K_{X,L}$. Partition $D$ into four blocks, 
\begin{equation}D=\left[\begin{array}{ccc|c} \,\,\,\,\,\, & u & \,\,\,\,\,\, & \,\,x \\ \hline \,\,\,\, & \,\,\,\,& \,\,\,\, & \,\,\,\, \\ \,\,\,\,\,\, & C & \,\,\,\,\,\, & \,\,v \\ \,\,\,\, & \,\,\,\, & \,\,\,\, & \,\,\,\,\end{array}\right],\label{block_matrix}\end{equation}
where:
\begin{itemize}
\item  $C$ is the $r_{\mathrm{true}}\times r_{\mathrm{true}}$ submatrix of $D$ obtained by retaining the last $r_{\mathrm{true}}$ rows and last $r_{\mathrm{true}}$ columns of $K_{X,L}$.

\item $u$is the $r_{\mathrm{true}}\times 1$ row vector with the first $r_{\mathrm{true}}$ entries of the first row of $C$.

\item $v$ is the $1\times r_{\mathrm{true}}$ column vector with the last $r_{\mathrm{true}}$ entries of the last column of $C$

\item $x\in\mathbb{R}$ is the middle element on the diagonal of $K_{X,L}$, or equivalently the top right entry of the matrix $C$.
\end{itemize}
Note that $\mathrm{det}(D)=0$ (because $\mathrm{rank}(K_{X,L})=r_{\mathrm{true}}$) whereas $\mathrm{det}(C)\neq 0$ (because we are operating in the regime where $r_{\mathrm{true}}$-minors of $K_{X,L}$ are non-zero). It follows that
$$0=\mathrm{det}(D)=\underset{\neq 0}{\underbrace{\mathrm{det}(C)}} (x-u C^{-1}v)$$
showing that $x$ is a rational function of the entries of $C$, $u$, and $v$. It follows that the middle element of $K_{X,L}$ is uniquely specified by the off-diagonal elements of $K_{X,L}$.

Notice that any diagonal element of $K_{X,L}$ can be brought to the middle position of the diagonal by means of the conjugation $\Pi K_{X,L} \Pi\transpose  $, with $\Pi$ a suitable permutation matrix. This operation maps diagonal elements onto diagonal elements, and preserves the property that $r_{\mathrm{true}}$-minors of $K_{X,L}$ are non-vanishing. It follows that the diagonal elements of $K_{X,L}$ are uniquely determined by its off diagonal elements when $L=2r_{\mathrm{true}}+1$. 

For any $L>2r_{\mathrm{true}}+1$ one can apply the exact same procedure working with the top-left $(2r_{\mathrm{true}}+1)\times 2r_{\mathrm{true}}+1$ submatrix of $K_{X,L}$ instead of the entire matrix $K_{X,L}$, and using permutations to bring the remaining diagonal elements in-and-out of the said submatrix. It follows that $L_{\dagger}\leq 2r_{\mathrm{true}}+1$.

\subsubsection{Covariance Spectra Guaranteeing That The $d$-Minors of $K_{X,L}$ Are Non-Vanishing}\label{eigensystems} In any scenario, we need to show (in the notation introduced in the beginning of the main Section) that 
\begin{center}
\begin{footnotesize}
\fbox{$U^A(t_1,\ldots,t_L)$ is of full column rank $d$ for any index set $A\subset \{1,...,L\}$ of cardinality $d$}
\end{footnotesize}
\end{center}

\smallskip
\noindent We will do this by means of verifying an equivalent condition stated in the next lemma:

\begin{lemma}\label{lemma-rkhs} {Let $k_X$ be continuous, $\{t_1,...,t_L\}$ a collection of nodes, and $A\subset \{1,...,L\}$ an index set of cardinality $d$. The matrix $U^A(t_1,\ldots,t_L)$ is of full column rank $d$ if and only if there exist $d$ functions $h_1(\cdot),\ldots,h_d(\cdot) \in \mathrm{RKHS}(k_X)$ such that the matrix $\{h_j(t_i)\}_{i\in A, 1\leq j\leq d}$ is non-singular.}
\end{lemma}

\begin{proof} 

{Assume first that there are $d$ functions $h_j\in\mathrm{RKHS}(k_X)$ such that $\{h_j(t_i)\}$ is non-singular. The function $h_j$ being in the RKHS is equivalent the existence of a square summable sequence $\{\theta_{j,n}\}_{n\geq 1}$ such that $h_j=\sum_{n\geq 1}\theta_{j,n}\lambda_n^{1/2}\varphi_n$. Using the Cauchy-Schwarz inequality and Mercer's theorem, the series can be seen to converge uniformly and absolutely, as its square is bounded by}
$${\sum_{j\geq 1}\theta^2_{j,n}\sum_{j\geq 1}\lambda_j \varphi_j^2(x)=\sum_{j\geq 1}\theta^2_{j,n} \sum_{j\geq 1}\lambda_j |\varphi_j(x)\varphi_j(x)|\leq \sup_{u\in[0,1]}|k_X(u,u)|<\infty.}$$
{Hence, we may write $h_j(t_i)=\sum_{n\geq 1}\theta_{j,n}\lambda_n^{1/2}\varphi_n(t_i)$, which shows that the range of $U^A$ contains $d$ linearly independent vectors -- namely the columns of the $d\times d$ non-singular matrix $\{h_j(t_i)\}$. It follows that $\mathrm{rank}(U^A)=d$.}

{To prove the converse, assume that $\mathrm{rank}(U^A)=d$. Then $\mathrm{rank}(U^A (U^A)\transpose  )=d$ too. But $U^A (U^A)\transpose  =\{k_X(t_i,t_j)\}_{i,j\in A}$. Define $h_j(\cdot)=k_X(\cdot,x_j)$ for $j\in A$. It follows that there exist $d$ functions $h_j\in\mathrm{RKHS}(k_X)$ such that $\{h_j(t_i)\}$ is non-singular.}

\end{proof}

\noindent Let us now revisit the scenarios enumerated earlier, in light of the Lemma:

\medskip
\noindent{\textbf{Case where $k_X$ is strictly positive definite.}} {In this case, define $h_j(\cdot)=k_X(\cdot,t_j)$, and notice that each such function is an element of the RKHS of $k_X$. Indeed, the matrix $\{h_j(t_i)\}$ is simply the matrix $\{k_X(t_i,t_j)\}_{i,j\in A}$ which is strictly positive definite by strict positive definiteness of $k_X$ for any $d$ pairwise distinct nodes $\{t_j\}_{j\in A}$.}

\bigskip
\noindent\textbf{Cases where $k_X$ is positive semidefinite}. 

\begin{itemize}

\smallskip
\item If $\{1,x,x^2,...,x^{d-1}\}\in \mathrm{RKHS}(k_X)$, we can define $h_j(t_i)=t_i^{j-1}$. This is a $d\times d$ Vandermonde matrix, and hence non-singular for any pairwise distinct $d$-tuple $\{t_i\}$. Clearly, the same discussion applies if  $\{1,F(x),F^2(x),...,F^{d-1}(x)\}\in \mathrm{RKHS}(k_X)$ for any for $F$ strictly increasing, by defining $h_j(t_i)=[F(t_i)]^{j-1}$, and noting that $F(t_i)=\tau_i$ simply yields a different grid of distinct points, so that $h_j(\tau_i)$ is again a Vandermonde matrix over distinct nodes.

\smallskip
\item If $\mathrm{RKHS}(k_X)$ contains $d$ linearly independent polynomials $\{h_j\}_{j=1}^{d}$ of highest degree greater than $d-1$, the matrix is $\{h_j(t_i)\}$ is a polynomial matrix of generalised Vandermonde type, with a determinant proportional to $Q_A(\{t_m\}_{m\in A})\prod_{\{(i,j)\in A\times A:\, i<j\}}(t_i-t_j)$, where $Q_A$ is a finite degree polynomial. Hence, provided the grid points are distinct, this determinant vanishes nowhere but at the finitely many $\{t_i\}_{i\in A}$ satisfying polynomial restrictions dictated by the root structure of $Q_A$. Since there are finitely many index sets $A\subset\{1,\ldots,L\}$ of cardinality $d$, there are also finitely many corresponding polynomials $Q_A$, and hence only a finite number of grids $\{t_i\}_{i=1}^{L}$ for which $Q_A(\{t_j\}_{j\in A})$, for some choice of $A$, vanishes. The same discussion applies if $\mathrm{RKHS}(k_X)$ contains $d$ linearly independent polynomials of highest degree greater than $d-1$, each composed with a monotone map $F$, by switching to the grid $\tau_i=F(t_i)$.

\smallskip
\item  If $\mathrm{RKHS}(k_X)$ contains the collection of the first $d$ Fourier basis elements, or $d$ linearly independent trigonometric polynomials, the same discussion can be repeated as in the polynomial case, except with trigonometric polynomials (seen as polynomials of pairwise distinct unit modulus complex arguments).

\smallskip
\item If $\mathrm{RKHS}(k_X)$ contains a collection of $d$ functions $\{h_j\}_{j=1}^d$ that are linearly independent on any subset $K\subseteq [0,1]$ of positive Lebesgue measure, we claim that $\mathrm{det}(\{h_j(t_i)\}_{1\leq j\leq d,i\in A})$ is non-vanishing for almost all $d$-tuples $\{t_i\}_{i\in A}$. This will require a more lengthy argument, and to relax the indexing notation, we will write $\{t_j\}_{j\in A}=\{x_j\}_{j=1}^{d}$. For $q\in \{1,...,d\}$, write $H_q=\{h_{j}(x_i)\}_{i,j=1}^{q}$.To show that $\mathrm{det}(H_d)$ is non-vanishing for almost all $d$-tuples $\{x_1,...,x_d\}$, we will use induction:
\begin{enumerate}
\item We will first prove that $\mathrm{det}[H_1(x_1)]\neq 0$ almost everywhere on $[0,1]$. 

\item Then we will prove that if $\mathrm{det}[H_{q-1}(x_1,...,x_{q-1})]\neq 0$ almost everywhere on $[0,1]^{q-1}$, implies that $\mathrm{det}[H_{q}(x_1,...,x_{q})]\neq 0$ almost everywhere on $[0,1]^q$, for any $2\leq q\leq d$. 
\end{enumerate}

\smallskip
\noindent\emph{Step 1: Case $q=1$}. We need to show that $\mathrm{det}[H_1(x_1)]\neq 0$ for almost all $x_1\in [0,1]$. Equivalently, that $h_1(y)$ cannot vanish on a set $K\subseteq[0,1]$ of positive Lebesgue measure. Since the $d$ functions $\{h_j\}_{j=1}^{d}$ are linearly independent on any set of positive Lebesgue measure, $h_1(y)$ cannot vanish uniformly on such a set.

\smallskip
\noindent\emph{Step 2: Induction step}. Now take $2\leq q\leq d$, and suppose that $\mathrm{det}(H_{q-1})\neq 0$ almost everywhere on $[0,1]^{q-1}$, but that $\mathrm{det}(H_q)=0$ for all $(x_1,...,x_q)\in G_q$, for some $G_q\subseteq [0,1]^q$ of positive $q$-Lebesgue measure. We will obtain a contradiction. Note that, 
$$0<\mathrm{Leb}_q(G_q)=\int_{[0,1]^{q-1}} \mathrm{Leb}_{1}(G^{x_1,...,x_{q-1}})\mathrm{Leb}_{q-1}(dx_1\times\hdots\times dx_{q-1}),$$
where $G^{x_1,...,x_{q-1}}=\{y\in [0,1] (x_1,...,x_{q-1},y)\in G_q\}\subseteq[0,1]$ is the $(x_1,...,x_{q-1})$-section of $G_q$. It follows that $\mathrm{Leb}_{1}(G^{x_1,...,x_{q-1}})>0$ for all $(x_1,...,x_{q-1})$ in a set $G_{q-1}\subseteq [0,1]^{q-1}$ of positive $(q-1)$-Lebesgue measure, i.e. $\mathrm{Leb}_{q-1}(G_{q-1})>0$.

With this observation in mind, we use the Leibniz formula for the determinant to translate the statement that
$$\mathrm{det}(H_q)=0,\qquad \forall\, (x_1,...,x_q)\in G_q$$ 
into the equivalent statement
$$\sum_{\pi \in \mathrm{Sym}\{1,...,q\}} \mathrm{sgn}(\pi) \prod_{i=1}^q h_i (x_{\pi(i)})=0,\quad\forall (x_1,\ldots,x_q)\in G_q.$$
where $\mathrm{Sym}\{1,...,q\}$ denotes the group of permutations on $\{1,...,q\}$, and $\mathrm{sgn}(\pi)$ is the signature of a permutation $\pi$. Thus, for any $(q-1)$-tuple $(x_1,\ldots,x_{q-1})\in G_{q-1}$, we may view the last expression as a function of the last coordinate $y=x_q$, and write
$$0=\sum_{\pi \in \mathrm{Sym}\{1,...,q\}} \mathrm{sgn}(\pi)\Bigg[ \prod_{i:\pi(i)\neq q} h_i (x_{\pi(i)})  \Bigg]  h_{\pi^{-1}(q)}(y)$$
Regrouping the summations now yields
$$0=\sum_{i=1}^{q}\underset{=\alpha_{i}(x_1,...,x_{q-1})}{\underbrace{\left[\sum_{\pi\in \mathrm{Sym}\{1,...,q-1\}} \mathrm{sgn}(\pi)  \prod_{j\neq i} h_{j} (x_{\pi(\rho_i(j))})\right] }} h_i(y), \quad\forall y\in G^{x_1,...,x_{q-1}}\subseteq[0,1],
$$
where for $1\leq i \leq q$, the mapping $\rho_i(j)$ gives the rank (in the sense of sequentially increasing order) of any $j\in\{1,...,q\}\setminus\{i\}$, thus providing a bijection between $\{1,...,q\}\setminus\{i\}$ and $\{1,...,q-1\}$.

But we have observed that $G^{x_1,...,x_{q-1}}$ has positive $\mathrm{Leb}_1$-measure for any $(x_1,...,x_{q-1})\in G_{d-1}$, and the $\{h_i\}_{i=1}^q$ are linearly independent on any set of positive $\mathrm{Leb}_1$ measure. Hence, it must be that 
$$\alpha_1(x_1,...,x_{q-1})=\alpha_2(x_1,...,x_{q-1})=\hdots=\alpha_q(x_1,...,x_{q-1})=0$$
for any  $(q-1)$-tuple $(x_1,...,x_{q-1})\in G_{q-1}$, where $\mathrm{Leb}_{q-1}(G_{q-1})>0$. But now notice that 
$$\alpha_q(x_1,...,x_{q-1})=\left[\sum_{\pi\in \mathrm{Sym}\{1,...,q-1\}} \mathrm{sgn}(\pi)  \prod_{1\leq j\leq q-1} h_{j} (x_{\pi(j)})\right]=H_{q-1}(x_1,...,x_{q-1}),$$
because $\rho_q(j)=j$ for any $j\in \{1,...,q-1\}$.  We have thus arrived at a contradiction of our inductive induction assumption that $H_{q-1}(x_1,...,x_{q-1})\neq 0$ on any set of positive $\mathrm{Leb}_{q-1}$-measure.
\end{itemize}

\subsection{{On the Invertibility of the Hessian $\nabla^2\Psi$}} \label{hessian_invertibility} The purpose of this section is to further analyse Assumption (H), used to deduce the large sample distribution of $nT_q$ under $H_{0,q}$,
\begin{quote}
\begin{description}
\item[\textbf{Assumption (H)}:]Under $H_{0,q}$, there exists a factorisation $K_{X,L}=C_{0}C_{0}\transpose $, where $C_{0}\in\mathbb{R}^{L\times q}$, such that $\mathrm{det}(\nabla^{2}\Psi(C_{0}))\neq 0$.
\end{description}
\end{quote}
In particular, we will show:
\begin{enumerate}
\item That Assumption (H), is satisfied if Assumption (E) below holds true:
 \begin{quote}
\begin{description}
\item[\textbf{Assumption (E)}:]The $q$ leading eigenvectors of $K_{X,L}$ have non-zero entries. 
\end{description}
\end{quote}

\item That Assumption (E) is satisfied in the scenarios mentioned in Remark \ref{grid_remark} and listed in detail at the beginning of Section \ref{critical_value}, for almost all grids $t_1<...<t_L$.
\end{enumerate}
To show the first point, choose $C$ to be
\begin{equation} 
C= 
\left(  \begin{array}{ccc}
\lambda_1^{1/2}\varphi_1(t_1) & \hdots  &\lambda_q^{1/2}\varphi_{q}(t_1)   \\
\lambda_1^{1/2}\varphi_1(t_2) & \hdots  &\lambda_q^{1/2}\varphi_{q}(t_2)  \\
\vdots & \vdots &\vdots    \\
\lambda_1^{1/2}\varphi_1(t_L) & \hdots &\lambda_q^{1/2}\varphi_{q}(t_L) 
\end{array}
\right)\end{equation}  
i.e. exactly as in Equation \eqref{Umatrix}, which reduces to an equation for an $L\times q$ matrix under $H_{0,q}$. Mercer's theorem implies that, indeed, $K_{X,L}=CC\transpose  $. Let
$$\underset{L\times q}{C}= \underset{L\times q}{V}\, \underset{q\times q}{\Gamma}\, \underset{q\times q}{W\transpose  }$$
be the singular value decomposition of $C$, where $\Gamma$ is diagonal, and $V$ and $W$ are orthogonal. Define
$$\underset{L\times q}{H}=  CW ={V} {\Gamma}=\left(  \begin{array}{ccc}
\lambda_1^{1/2}\varphi_1(t_1) & \hdots  &\lambda_q^{1/2}\varphi_{q}(t_1)   \\
\lambda_1^{1/2}\varphi_1(t_2) & \hdots  &\lambda_q^{1/2}\varphi_{q}(t_2)  \\
\vdots &  &\vdots    \\
\lambda_1^{1/2}\varphi_1(t_L) & \hdots &\lambda_q^{1/2}\varphi_{q}(t_L) 
\end{array}
\right)\left(  \begin{array}{ccc}
w_{1,1} & \hdots  & w_{1,q}   \\
w_{2,1} & \hdots  &w_{2,q} \\
\vdots &  &\vdots    \\
w_{q,1} & \hdots &w_{q,q}
\end{array}
\right)=\left(  \begin{array}{ccc}
h_1(t_1) & \hdots  &h_{q}(t_1)   \\
h_1(t_2) & \hdots  &h_{q}(t_2)  \\
\vdots &  &\vdots    \\
h_1(t_L) & \hdots &h_{q}(t_L) 
\end{array}
\right)$$
and note that 
\begin{eqnarray*}
HH\transpose  &=&CWW\transpose  C\transpose  =  K_{X,L}\\
HH\transpose  &=&V \Gamma^2 V\transpose\\
H\transpose   H &=& \Gamma V\transpose   V \Gamma = \Gamma^2
\end{eqnarray*}
In particular note that the first and second line imply that $V$ has the leading $q$ eigenvectors of $K_{X,L}$ as its columns.
Our aim will be to show that choosing the  $L\times q$ factor $H$ of $K_{X,L}$ yields
$$ \mathrm{det}(\nabla^2\Psi(H))\neq 0,$$
provided $H$ has non-zero entries -- equivalently, provided $V$ has non-zero entries, i.e. Assumption (E) holds true. Note that the form of the Hessian at any $C\in \mathbb{R}^{L\times q}$ has been shown to be
$$\nabla^{2}\Psi(C) = -4I_{q} \otimes (P_{L} \circ (K_{X,L} - CC\transpose )) + 4P_{qL} \circ \{(C\transpose  \otimes C)M\} + 4(P_{q} \circ C\transpose C) \otimes I_{L}.$$
where $\circ$ is the Hadamard product, $P_m$ is the $m\times m$ matrix containing $0$'s on the diagonal and $1$'s everywhere else, and $M$ is the order $(L,q)$ commutation matrix, i.e. the unique permutation matrix satisfying
$$\mathrm{vec}(R\transpose  )=M\mathrm{vec}(R)$$
for any $R\in\mathbb{R}^{L\times q}$. Plugging in $H$, the Hessian $\nabla^2\Psi(H)$ reduces to 
$$\nabla^{2}\Psi(H) = 0+4P_{qL} \circ \{(H\transpose  \otimes H)M\}+0$$
because $HH\transpose  =K_{X,L}$ and  $H\transpose   H=\Gamma^2$ is a diagonal matrix. Therefore, it suffices to show that 
$$\mathrm{det}\Big\{  P_{qL} \circ \{(H\transpose  \otimes H)M\}  \Big\} \neq 0.$$
To this aim, we will make use of two Lemmas, the first of which probes the structure of $(H\transpose  \otimes H)M$:

\begin{lemma}\label{H-lemma}
The $(Lq)\times(Lq)$ matrix $(H\transpose  \otimes H)M$ is a $q\times q$ block matrix of $L\times L$ rank 1 blocks $\{H_{ij}\}_{i,j=1}^{q}$, defined as
$$\underset{L\times L}{H_{ij}}=\underset{L\times 1}{H_i} \underset{1\times L}{H_j\transpose  }\qquad\mbox{ \& }\qquad\underset{L\times 1}{H_i}=\left(  \begin{array}{c}
h_i(t_1)  \\
h_i(t_2)  \\
\vdots     \\
h_i(t_L) 
\end{array}
\right),$$
i.e. $H_j$ is the $j$th column of $H$. 
\end{lemma}

\begin{proof}
Let $v,u\in\mathbb{R}^{Lq}$ and define $A,B$ to be the $L\times q$ matrices such that $v=\mathrm{vec}(A)$ and $u=\mathrm{vec}(B)$. Then, for $\langle \cdot,\cdot\rangle_F$ the Frobenius inner product, and recalling that $M$ is the order $(L,q)$ commutation matrix, we may write
\begin{eqnarray*}
u\transpose  (H\transpose  \otimes H)M v &=& \mathrm{vec}(B)\transpose  (H\transpose  \otimes H)M \mathrm{vec}(A)\\
&=&\mathrm{vec}(B)\transpose  (H\transpose  \otimes H) \mathrm{vec}(A\transpose  )\\
&=&\mathrm{vec}(B)\transpose  \mathrm{vec}(HA\transpose  H)\\
&=&\langle B, HA\transpose   H \rangle_F\\
&=&\mathrm{trace}\{(B\transpose  H)(A\transpose  H)\}.
\end{eqnarray*}
Let $G$ be the stipulated block matrix. We will now show that 
$$\mathrm{trace}\{(B\transpose  H)(A\transpose  H)\} = u\transpose   G v,$$
thus establishing that $G=(H\transpose  \otimes H)M$, by arbitrary choice of $u,v$. To this aim, partition $A$ and $B$ as
$$A=\left(  \begin{array}{cccc}
A_1  & A_2 & \hdots & A_q
\end{array}
\right)\qquad \& \qquad B=\left(  \begin{array}{cccc}
B_1  & B_2 & \hdots & B_q
\end{array}
\right)$$
where $A_j,B_j\in\mathbb{R}^L$ are the $j$th columns of $A$ and $B$, respectively. This partitions the coordinates of $u$ and $v$ into groups of $L$,
\begin{eqnarray*}v&=&\mathrm{vec}(A)=\mathrm{vec}\left(  \begin{array}{cccc}
A_1  & A_2 & \hdots & A_q
\end{array}
\right)=\left(  \begin{array}{c}
A_1  \\
 A_2 \\
  \vdots \\
   A_q
\end{array}
\right)\\
u&=&\mathrm{vec}(B)=\mathrm{vec}\left(  \begin{array}{cccc}
B_1  & B_2 & \hdots & B_q
\end{array}
\right)=\left(  \begin{array}{c}
B_1  \\
 B_2 \\
  \vdots \\
   B_q
\end{array}
\right).
\end{eqnarray*}
We can now calculate
\begin{eqnarray*}
A\transpose   H = \left(  \begin{array}{c}
A_1\transpose    \\
  \vdots \\
   A_q\transpose  
\end{array}
\right)
\left(  \begin{array}{ccc}
H_1 & \hdots & H_q
\end{array}
\right)=
\left(  \begin{array}{ccc}
A_1\transpose   H_1 & \hdots & A_q\transpose   H_q\\
\vdots & \ddots & \vdots\\
A_q\transpose   H_1 & \hdots & A_q\transpose  H_q
\end{array}
\right)
\end{eqnarray*}
and
\begin{eqnarray*}
B\transpose   H = \left(  \begin{array}{c}
B_1\transpose    \\
  \vdots \\
   B_q\transpose  
\end{array}
\right)
\left(  \begin{array}{ccc}
H_1 & \hdots & H_q
\end{array}
\right)=
\left(  \begin{array}{ccc}
B_1\transpose   H_1 & \hdots & B_q\transpose   H_q\\
\vdots & \ddots & \vdots\\
B_q\transpose   H_1 & \hdots & B_q\transpose  H_q
\end{array}
\right)
\end{eqnarray*}
so that
$$
(B\transpose   H)(A\transpose   H) =
\left(  \begin{array}{ccc}
B_1\transpose   H_1 & \hdots & B_q\transpose   H_q\\
\vdots & \ddots & \vdots\\
B_q\transpose   H_1 & \hdots & B_q\transpose  H_q
\end{array}
\right)
\left(  \begin{array}{ccc}
A_1\transpose   H_1 & \hdots & A_q\transpose   H_q\\
\vdots & \ddots & \vdots\\
A_q\transpose   H_1 & \hdots & A_q\transpose  H_q
\end{array}
\right).
$$
The $i$th diagonal element of the last expression is seen to be equal  to $\sum_{j=1}^qB_i\transpose  H_jA_j\transpose   H_i$.
Consequently, 
\begin{equation}\label{1st-trace}
\mathrm{trace}\{(B\transpose   H)(A\transpose   H)\}=\sum_{i=1}^{q}\sum_{j=1}^{q}B_i\transpose  H_jA_j\transpose   H_i.
\end{equation}
Now we turn our attention to $v\transpose   G v$ which is seen to be
\begin{eqnarray*}
v\transpose   G v&=&\left(  \begin{array}{ccc}
B_1\transpose    & \hdots & B_q\transpose  
\end{array}
\right)
\left(  \begin{array}{ccc}
H_1 H_1\transpose   & \hdots & H_1 H_q\transpose  \\
\vdots & \ddots & \vdots\\
H_q H_1\transpose   & \hdots & H_q H_q\transpose  
\end{array}
\right)
\left(  \begin{array}{c}
A_1  \\
  \vdots \\
   A_q
\end{array}
\right)\\
&=&\left(  \begin{array}{ccc}
B_1\transpose    & \hdots & B_q\transpose  
\end{array}
\right)
\left(  \begin{array}{c}
\sum_{j=1}^q  H_1 H_j\transpose  A_1\\
 \vdots\\
\sum_{j=1}^q H_q H_j\transpose  A_q\\
\end{array}
\right)\\
&=&\sum_{i=1}^{q}\sum_{j=1}^{q}B_j\transpose  H_i H_j\transpose   A_i.
\end{eqnarray*}
Upon observing that the last line coincides with expression \eqref{1st-trace}, the proof is complete.

\end{proof}

The second Lemma is a standard fact about Hadamard products, stated here without proof for completeness (see, e.g. \citet{horn2012matrix}).

\begin{lemma}\label{hadamard-lemma}
Let $x,y\in\mathbb{R}^m$ and $A\in\mathbb{R}^{m\times m}$. Then,
$$A\circ (y x\transpose  ) = \Delta_{y}A\Delta_{x},$$
for $\Delta_v$ the $m\times m$ diagonal matrix with the elements of $v\in\mathbb{R}^m$ on its diagonal.
\end{lemma}

Using the fact that the entries of $H$ are assumed non-zero, and armed with the last two Lemmas, we will now show that $\mathrm{det}\Big\{  P_{qL} \circ \{(H\transpose  \otimes H)M\}  \Big\} \neq 0$. Clearly, it suffices to show that  $\mathrm{det}\Big\{Q  [P_{qL} \circ \{(H\transpose  \otimes H)M \}] Q\transpose   \Big\} \neq 0$ for any non-singular $(Lq)\times (Lq)$ matrix $Q$. Define $Q$ as a block diagonal matrix comprised of $q^2$ blocks of dimension $L\times L$,
$$\underset{Lq\times Lq}{Q}= \left(  \begin{array}{ccc}
\underset{L\times L}{Q_1} & \hdots  &\underset{L\times L}{\bm 0}   \\
\vdots & \ddots &\vdots    \\
\underset{L\times L}{\bm 0}& \hdots & \underset{L\times L}{Q_q}
\end{array}
\right),\qquad Q_i=\left(  \begin{array}{ccc}
0 & \hdots  &\displaystyle \frac{1}{h_i(t_L)}   \\
\vdots & \iddots &\vdots    \\
\displaystyle\frac{1}{h_i(t_1)}& \hdots & 0
\end{array}
\right)$$
The matrix $Q$ is well defined since the entries of $H$ are all non-zero, and is clearly of full rank. Moreover, by Lemma \ref{H-lemma},
$$
Q  [P_{qL} \circ \{(H\transpose  \otimes H)M \}] Q\transpose  =Q \left(  \begin{array}{ccccc}
P_L\circ H_{1,1} & H_{1,2}  & \hdots  & H_{1,q-1} & H_{1,q}\\
H_{2,1} & P_L\circ H_{2,2}  & \hdots  & H_{2,q-1} & H_{2,q}\\
\vdots & \vdots  & \ddots  & \vdots & \vdots\\
H_{q-1,1} & H_{q-1,2}  & \hdots  & H_{q-1,q-1} & H_{q-1,q}\\
H_{q,1} & H_{q,2}  & \hdots  & H_{q,q-1} & P_L\circ H_{q,q}\\
\end{array}
\right)  Q\transpose  
$$

$$
=\left(  \begin{array}{ccccc}
Q_1(P_L\circ H_{1,1}) Q_1\transpose   & Q_1 H_{1,2}Q_2\transpose    & \hdots  & Q_1H_{1,q-1}Q_{q-1}\transpose   & Q_1H_{1,q}Q_q\transpose  \\
Q_2H_{2,1}Q_1\transpose   & Q_2 (P_L\circ H_{2,2}) Q_2\transpose    & \hdots  & Q_2H_{2,q-1}Q_{q-1}\transpose   & Q_2H_{2,q}Q_q\transpose  \\
\vdots & \vdots  & \ddots  & \vdots & \vdots\\
Q_{q-1}H_{q-1,1} Q_1\transpose  & Q_{q-1}H_{q-1,2}Q_2\transpose    & \hdots  & Q_{q-1}H_{q-1,q-1}Q_{q-1}\transpose   & Q_{q-1}H_{q-1,q}Q_q\transpose  \\
Q_qH_{q,1} Q_1\transpose  & Q_qH_{q,2}Q_2\transpose    & \hdots  & Q_qH_{q,q-1}Q_{q-1}\transpose   & Q_q (P_L\circ H_{q,q}) Q_q\transpose  \\
\end{array}
\right)
$$
We claim that the last matrix is equal to $P_{qL}$. To show this, we will show that:
\begin{itemize}
\item \textbf{Diagonal blocks equal $\bm{P_L}$}. The typical diagonal block is of the form $Q_i(P\circ H_{i,i})Q_i\transpose  =Q_i(P\circ H_i H_i\transpose  )Q_i\transpose  =Q_i\Delta_{H_i}P_L\Delta_{H_i}Q_i\transpose$, where $\Delta_{H_i}$ is the diagonal matrix whose diagonal contains the elements of the column vector $H_i$, and we have made use of Lemma \ref{hadamard-lemma} to obtain the last equality. Now we calculate
$$Q_i \Delta_{H_i}=\left(  \begin{array}{ccc}
0 & & 1/h_i(t_L)\\
& \iddots & \\
1/h_i(t_1) & & 0
\end{array}\right)
\left(  \begin{array}{ccc}
h_i(t_1) & & 0\\
& \ddots & \\
0 & & h_i(t_L)
\end{array}\right)=
\underset{E}{\underbrace{\left(  \begin{array}{ccc}
0 & & \,\,1\,\,\\
& \iddots & \\
\,\,1\,\, & & 0
\end{array}\right)}}$$
where $E$ is known as the \emph{exchange matrix} of order $L$. Noting that $\Delta_{H_i}Q_i\transpose  =(Q_i\Delta_{H_i})\transpose  $ and $E\transpose  =E$, we conclude that the typical diagonal block equals $EP_LE=P_L$.

\item \textbf{Off-diagonal blocks equal $\mathbf{1}_L\mathbf{1}_L\transpose  $}. The typical off-diagonal block is of the form $Q_i H_{i,j}Q_j\transpose  =Q_iH_iH_j\transpose  Q_j\transpose  =(Q_iH_i)(Q_jH_j)\transpose  $, and the latter equals
$$
\left(  \begin{array}{ccc}
0 & & 1/h_i(t_L)\\
& \iddots & \\
1/h_i(t_1) & & 0
\end{array}\right)
\left(  \begin{array}{c}
h_i(t_1)\\
\vdots\\
h_i(t_L)
\end{array}\right)
\left[
\left(  \begin{array}{ccc}
0 & & 1/h_j(t_L)\\
& \iddots & \\
1/h_j(t_1) & & 0
\end{array}\right)
\left(  \begin{array}{c}
h_j(t_1)\\
\vdots\\
h_j(t_L)
\end{array}\right)
\right]\transpose  
=
\left(  \begin{array}{c}
1\\
\vdots\\
1
\end{array}\right)
\left(  \begin{array}{ccc}
1 & \hdots & 1
\end{array}\right).
$$ 
\end{itemize}

\bigskip
\noindent In summary, we have shown that 
$$Q  [P_{qL} \circ \{(H\transpose  \otimes H)M \}] Q\transpose  =P_{qL},$$
provided $H$ has non-zero entries. Now we use the fact that $\mathrm{det}(P_{qL})=(qL-1)(-1)^{qL-1}$ (which can be readily checked by row reduction),  to conclude that the determinant is non-zero provided $H$ has non-zero entries, or, equivalently provided that $H=V\Gamma$ has non-zero entries. In summary, we have established:
\begin{center}
\fbox{Assumption (E) $\implies$ Assumption (H)}
\end{center}

Now we move to the second point, i.e. showing that Assumption (E) is satisfied in the scenarios listed in the beginning of Section \ref{critical_value}, for almost all grids $t_1<...<t_L$. To see this, recall the definition of $H$ as
$$H=CW=  \left(  \begin{array}{ccc}
\lambda_1^{1/2}\varphi_1(t_1) & \hdots  &\lambda_q^{1/2}\varphi_{q}(t_1)   \\
\lambda_1^{1/2}\varphi_1(t_2) & \hdots  &\lambda_q^{1/2}\varphi_{q}(t_2)  \\
\vdots &  &\vdots    \\
\lambda_1^{1/2}\varphi_1(t_L) & \hdots &\lambda_q^{1/2}\varphi_{q}(t_L) 
\end{array}
\right)\left(  \begin{array}{ccc}
w_{1,1} & \hdots  & w_{1,q}   \\
w_{2,1} & \hdots  &w_{2,q} \\
\vdots &  &\vdots    \\
w_{q,1} & \hdots &w_{q,q}
\end{array}
\right)=\left(  \begin{array}{ccc}
h_1(t_1) & \hdots  &h_{q}(t_1)   \\
h_1(t_2) & \hdots  &h_{q}(t_2)  \\
\vdots &  &\vdots    \\
h_1(t_L) & \hdots &h_{q}(t_L) 
\end{array}
\right)$$
where the function $h_j$ is a linear combination of the scaled eigenfunctions,
$$h_j(\cdot)=w_{1,j}\lambda_1^{1/2}\varphi_1(\cdot) +w_{2,j}\lambda_2^{1/2}\varphi_2(\cdot)+ \hdots  +w_{q,j}\lambda_q^{1/2}\varphi_{q}(\cdot).$$
Now suppose that the collection $\{\varphi_1,...,\varphi_q\}$ remains linearly independent when restricted to any $[0,1]$-subset of positive Lebesgue measure (as is the case in all of the scenarios listed at the top of Section \ref{critical_value}). Then $\{h_1,...,h_q\}$ have the same property. To see this, let $G\subset [0,1]$ be any subset of positive measure and assume that $\alpha_1 h_1(u)+\hdots \alpha_q h_q(u)=(h_1(u),\ldots, h_q(u))\bm\alpha=0$ for all $u\in G$. By definition, this implies that $(\lambda^{1/2}_1\varphi_1(u),\ldots, \lambda^{1/2}_q\varphi_q(u))(W\bm\alpha)=0$ on $G$, which can only happen if $W\bm\alpha=0$, because $\{\varphi_1,...,\varphi_q\}$ are linearly independent on $G$. Since $W$ is orthogonal, $W\bm\alpha=0$ now implies that $\alpha_1=...=\alpha_q=0$. In particular, this property implies that $h_j(u)\neq 0$ almost everywhere on $[0,1]$. Therefore, $H$ has non-zero entries for almost all grids $\{t_1,...,t_L\}$, or equivalently, Assumption (E) is satisfied for almost all grids.

\end{document}